\documentclass[12pt]{article}
\title{Generalized Clapeyron's theorem}
\author{Yury Grabovsky\thanks{Department of Mathematics, Temple University, Philadelphia, PA 19122, USA} \and Lev Truskinovsky\thanks{PMMH, CNRS -- UMR 7636, ESPCI, PSL,  75005 Paris, France}}

\usepackage{amssymb,bm}
\usepackage{amsmath}
\usepackage{amsthm}
\usepackage{graphicx}
\usepackage{times} 
\usepackage{BOONDOX-cal} 
\usepackage{pdfsync}
\usepackage{hyperref} 
\counterwithin{equation}{section} 
\hypersetup{pdfborder=0 0 0}
\usepackage[top=1.05in,left=1.05in,right=1.05in,bottom=1.05in]{geometry}

\newtheorem{theorem}{{\sc Theorem}}[section]

\newtheorem{lemma}[theorem]{{\sc Lemma}}

\newtheorem{remark}[theorem]{Remark}
\newtheorem{example}[theorem]{Example}

\newtheorem{definition}[theorem]{Definition}

\newcommand{\bb}[1]{\mathbb{ #1}}

\bmdefine\Bone{1}

\newcommand{\Sym}{\mathrm{Sym}}

\newcommand{\dOm}{\partial\Omega}

\newcommand{\bra}[1]{\overline{#1}}

\newcommand{\Trc}{\mathrm{Tr}\,}

\newcommand{\cof}{\mathrm{cof}}

\newcommand{\tns}[1]{#1\otimes #1}
\newcommand{\hf}{\displaystyle\frac{1}{2}}
\newcommand{\nth}[1]{\displaystyle\frac{1}{#1}}

\newcommand{\dif}[2]{\displaystyle\frac{\partial #1}{\partial #2}}
\newcommand{\Grad}{\nabla}
\newcommand{\Div}{\nabla \cdot}

\newcommand{\Md}{\partial}
\renewcommand{\Hat}[1]{\widehat{#1}}
\newcommand{\Tld}[1]{\widetilde{#1}}
\newcommand{\hess}[2]{\displaystyle\frac{\partial^2 #1}{\partial #2^2}}
\newcommand{\mix}[3]{\displaystyle\frac{\partial^2 #1}{\partial #2\partial #3}}
\newcommand{\av}[1]{\langle #1 \rangle}

\def\Xint#1{\mathchoice
{\XXint\displaystyle\textstyle{#1}}%
{\XXint\textstyle\scriptstyle{#1}}%
{\XXint\scriptstyle\scriptscriptstyle{#1}}%
{\XXint\scriptscriptstyle\scriptscriptstyle{#1}}%
\!\int}
\def\XXint#1#2#3{{\setbox0=\hbox{$#1{#2#3}{\int}$ }
\vcenter{\hbox{$#2#3$ }}\kern-.6\wd0}}

\def\dashint{\Xint-}
\newcommand\myatop[2]{\genfrac{}{}{0pt}{}{#1}{#2}}

\newcommand{\jump}[1]{\lbrack\!\lbrack #1 \rbrack\!\rbrack}
\newcommand{\lump}[1]{\lbrace\skew{-14.7}\lbrace\!\!#1\!\!\skew{14.7}\rbrace\rbrace}

\newcommand{\mat}[4]{\left[\begin{array}{cc}
\displaystyle{#1}&\displaystyle{#2}\\[1ex]
\displaystyle{#3}&\displaystyle{#4}\end{array}\right]}
\newcommand{\vect}[2]{\left[\begin{array}{c}
\displaystyle{#1}\\[1ex]\displaystyle{#2}\end{array}\right]}

\newcommand{\bc}{boundary condition}
\newcommand{\bvp}{boundary value problem}
\newcommand{\rhs}{right-hand side}
\newcommand{\lhs}{left-hand side}

\newcommand{\WLOG}{without loss of generality}
\newcommand{\nbh}{neighborhood}
\newcommand{\IFF}{if and only if }

\newcommand{\Ga}{\alpha}
\newcommand{\Gb}{\beta}
\newcommand{\Gd}{\delta}
\newcommand{\Ge}{\epsilon}

\newcommand{\Gve}{\varepsilon}

\newcommand{\Gg}{\gamma}

\newcommand{\Gk}{\kappa}

\newcommand{\Gl}{\lambda}

\newcommand{\Gs}{\sigma}

\newcommand{\Go}{\omega}

\newcommand{\GD}{\Delta}

\newcommand{\GG}{\Gamma}

\newcommand{\GL}{\Lambda}

\newcommand{\GS}{\Sigma}

\newcommand{\GO}{\Omega}

\bmdefine\BGa{\alpha}
\bmdefine\BGb{\beta}
\bmdefine\BGd{\delta}
\bmdefine\BGe{\epsilon}
\bmdefine\BGve{\varepsilon}
\bmdefine\BGf{\phi}
\bmdefine\BGvf{\varphi}
\bmdefine\BGg{\gamma}
\bmdefine\BGc{\chi}
\bmdefine\BGi{\iota}
\bmdefine\BGk{\kappa}
\bmdefine\BGl{\lambda}
\bmdefine\BGn{\eta}
\bmdefine\BGm{\mu}
\bmdefine\BGv{\nu}
\bmdefine\BGp{\pi}
\bmdefine\BGth{\theta}
\bmdefine\BGvth{\vartheta}
\bmdefine\BGr{\rho}
\bmdefine\BGvr{\varrho}
\bmdefine\BGs{\sigma}
\bmdefine\BGvs{\varsigma}
\bmdefine\BGt{\tau}
\bmdefine\BGj{\tau}
\bmdefine\BGu{\upsilon}
\bmdefine\BGo{\omega}
\bmdefine\BGx{\xi}
\bmdefine\BGy{\psi}
\bmdefine\BGz{\zeta}
\bmdefine\BGD{\Delta}
\bmdefine\BGF{\Phi}
\bmdefine\BGG{\Gamma}
\bmdefine\BGL{\Lambda}
\bmdefine\BGP{\Pi}
\bmdefine\BGT{\Theta}
\bmdefine\BGS{\Sigma}
\bmdefine\BGU{\Upsilon}
\bmdefine\BGO{\Omega}
\bmdefine\BGX{\Xi}
\bmdefine\BGY{\Psi}

\bmdefine\BFM{\mathfrak{M}}
\bmdefine\BFb{\mathfrak{b}}
\bmdefine\BFk{\mathfrak{k}}
\bmdefine\BFm{\mathfrak{m}}
\bmdefine\BFu{\mathfrak{u}}
\bmdefine\BFv{\mathfrak{v}}

\newcommand{\CA}{{\mathcal A}}
\newcommand{\CB}{{\mathcal B}}

\newcommand{\CE}{{\mathcal E}}

\newcommand{\CH}{{\mathcal H}}

\newcommand{\CJ}{{\mathcal J}}

\newcommand{\CO}{{\mathcal O}}
\newcommand{\CP}{{\mathcal P}}

\newcommand{\CS}{{\mathcal S}}

\newcommand{\CW}{{\mathcal W}}

\bmdefine\BCA{{\mathcal A}}
\bmdefine\BCB{{\mathcal B}}
\bmdefine\BCC{{\mathcal C}}
\bmdefine\BCD{{\mathcal D}}
\bmdefine\BCE{{\mathcal E}}
\bmdefine\BCF{{\mathcal F}}
\bmdefine\BCG{{\mathcal G}}
\bmdefine\BCH{{\mathcal H}}
\bmdefine\BCI{{\mathcal I}}
\bmdefine\BCJ{{\mathcal J}}
\bmdefine\BCK{{\mathcal K}}
\bmdefine\BCL{{\mathcal L}}
\bmdefine\BCM{{\mathcal M}}
\bmdefine\BCN{{\mathcal N}}
\bmdefine\BCO{{\mathcal O}}
\bmdefine\BCP{{\mathcal P}}
\bmdefine\BCQ{{\mathcal Q}}
\bmdefine\BCR{{\mathcal R}}
\bmdefine\BCS{{\mathcal S}}
\bmdefine\BCT{{\mathcal T}}
\bmdefine\BCU{{\mathcal U}}
\bmdefine\BCV{{\mathcal V}}
\bmdefine\BCW{{\mathcal W}}
\bmdefine\BCX{{\mathcal X}}
\bmdefine\BCY{{\mathcal Y}}
\bmdefine\BCZ{{\mathcal Z}}

\bmdefine\Bzr{ 0}
\bmdefine\Ba{ a}
\bmdefine\Bb{ b}
\bmdefine\Bc{ c}
\bmdefine\Bd{ d}
\bmdefine\Be{ e}
\bmdefine\Bf{ f}
\bmdefine\Bg{ g}
\bmdefine\Bh{ h}
\bmdefine\Bi{ i}
\bmdefine\Bj{ j}
\bmdefine\Bk{ k}
\bmdefine\Bl{ l}
\bmdefine\Bm{ m}
\bmdefine\Bn{ n}
\bmdefine\Bo{ o}
\bmdefine\Bp{ p}
\bmdefine\Bq{ q}
\bmdefine\Br{ r}
\bmdefine\Bs{ s}
\bmdefine\Bt{ t}
\bmdefine\Bu{ u}
\bmdefine\Bv{ v}
\bmdefine\Bw{ w}
\bmdefine\Bx{ x}
\bmdefine\By{ y}
\bmdefine\Bz{ z}
\bmdefine\BA{ A}
\bmdefine\BB{ B}
\bmdefine\BC{ C}
\bmdefine\BD{ D}
\bmdefine\BE{ E}
\bmdefine\BF{ F}
\bmdefine\BG{ G}
\bmdefine\BH{ H}
\bmdefine\BI{ I}
\bmdefine\BJ{ J}
\bmdefine\BK{ K}
\bmdefine\BL{ L}
\bmdefine\BM{ M}
\bmdefine\BN{ N}
\bmdefine\BO{ O}
\bmdefine\BP{ P}
\bmdefine\BQ{ Q}
\bmdefine\BR{ R}
\bmdefine\BS{ S}
\bmdefine\BT{ T}
\bmdefine\BU{ U}
\bmdefine\BV{ V}
\bmdefine\BW{ W}
\bmdefine\BX{ X}
\bmdefine\BY{ Y}
\bmdefine\BZ{ Z}

\newcommand{\SFC}{\mathsf{C}}

\newcommand{\SFL}{\mathsf{L}}

\begin{document}
\maketitle

\begin{abstract}
  Clapeyron's Theorem in  classical   linear elasticity  provides a way 
 to explicitly express the energy stored in an equilibrium configuration in terms of
 the work of the forces applied on the boundary. We derive  several new integral relations   which can be viewed as nonlinear analogs of this classical result,  reinterpreting them as rather general statements within Calculus of Variations. In the framework of nonlinear elasticity these relations  reflect  various   partial symmetries of the material response, for instance,   scale-invariance or scaling homogeneity. In particular, when the energy functional is  scale-free, the  obtained result can be  interpreted as  the  Generalized Clapeyron's Theorem (GCT). Remarkably,   it  combines rather naturally  the work of  physical and configurational forces.  We present a series of illuminating case studies showing the variety of applications of various obtained relations in different seemingly unrelated problems of mechanics.

\end{abstract}

\section{Introduction}
\setcounter{equation}{0}
\label{sec:intro}

A well known theorem of linear elasticity, first attributed to Clapeyron in \cite{lame1852} 
states that the elastic energy stored in a body loaded by (dead) forces is equal to half of the work
done by the loading device  \cite{love27,Sokolnikoff:1983:MTE}.
More precisely, it states that 
\begin{equation}
\int_{\GO}W (\Bx,e(\Bu))d\Bx=\hf\int_{\dOm}(\BGs\Bn,\Bu)dS(\Bx), \label{Clap}
\end{equation}
where $W (\Bx,\BGve)=(1/2)(\SFC(\Bx)\BGve,\BGve),$  $\SFC(\Bx)$ is the elasticity
tensor , $\Bu(\Bx)$ is the displacement, $\BGs(\Bx)=\SFC e(\Bu)$ is the linear stress
tensor and
\begin{equation}
  \label{eofu}
  e(\Bu)=\hf(\Grad\Bu(\Bx)+(\Grad\Bu(\Bx))^{T})
\end{equation}
is the linear strain
tensor.    While this theorem is obtained by simple integration by parts and is widely used in applications \cite{brun65,maro97,huet01,momo22} it carries a certain flavor of mystery. 
 
First, it has never
been published by the author himself \cite[p.~418]{tope60}. Second, it contains this
puzzling disappearance of the half of the work, which can be linked to the fact that
elastostatics is only a weak limit of elastodynamics (as the appropriately normalized magnitudes of the
characteristic velocities diverge) \cite{fotr03}. Finally,  to the best of our knowledge,  a   relation   \eqref{Clap}   has been viewed as strictly limited to classical elasticity in the sense that an adequate  analog does not exist in geometrically and physically nonlinear elasticity.  
 
Here we  close the apparent  gap between linear and nonlinear elasticity perspectives on surface representation of the bulk energy and   expose a deeper structure of  the classical Clapeyron's Theorem unifying it with some other similar results in mechanics and physics. This  will not only allows us 
to generalize Clapeyron's Theorem beyond linear elasticity but will also open  the possibility to  view it as a general result in the broader Calculus of Variations. 

We begin with the observation that the abstract problem which one can  associate with the classical Clapeyron's Theorem is to find conditions ensuring that the energy of an extremal configuration can be expressed through the data on the boundary.
Our solution of the problem links the generalization of Clapeyron's Theorem with  scale invariance  of the elastic energy, which we show to be  sufficient alone for obtaining  a representation of the total energy
as a surface integral.  Remarkably, while in the classical Clapeyron's Theorem the implied boundary term contains the work of  physical  forces only, we show that  in our Genralized Clapeyron's Theorem (GCT), the energy of an extremal configuration can be expressed through what can be interpreted as the generalized  work of both physical and configurational forces \cite{chad75,gurt00,podi02,ligu06,ssd09,maug16}. This shows that GCT represents  an indirect generalization of the classical Clapeyron's Theorem. A direct generalization of this classical result,
that does not involve any reference to configurational forces, relies on what we call $p$-homogeneity of the energy density. We show that if both scale invariance and $p$-homogeneity are operative at the same time, there is a corresponding conserved quantity that links the work of physical and configurational forces. Eliminating one or the other from GCT using this link, we can express the equilibrium energy through the work of either only physical or only configurational forces.

The proposed  broader reading of the
Clapeyron's Theorem allows one to advance in several directions simultaneously, from  generalizing  the classical Green's formula in nonlinear elasto-statics \cite{green73} (see also \cite{hill86}),   to obtaining a far reaching
multidimensional generalization of the classical  result of the Calculus of Variations expressing the minimal value of the energy through
the Weierstrass excess function \cite{knst86,grtrnc,grtrmms,rosa_pre}.

The crucial step in our analysis is the realization that behind  the conventional Clapeyron's Theorem is a formula first obtained by E.  Noether \cite{noether18,olver86}. We recall, in this respect,  that the  well known Noether's Theorem \cite{noether18} links  the presence of  variational symmetries   with  the existence of conservation laws in the
form of divergence-free combinations of   field  variables. To obtain this result
Noether  considered the action of a continuous group  of
symmetries of the variational integral and showed  that when a symmetry parameter is varied,  the
  corresponding perturbation of the functional   can be written as a sum  of the  Euler-Lagrange operator and  a divergence \cite{ibra84,olver86,blku89}. 
The underlying     Noether's  formula    has been   used broadly to generate various   integral identities  in many different  domains of science \cite{wagn02,poho65,poho70,reic04,rell56,bomi07,puse86,vdvo91,knops03,olver86}.
Following some previous work, see for instance  \cite{bomi07},  we move away from symmetry groups and consider instead more general classes of transformations.

We begin by  using   Noether's formula to show that by means of an appropriate scaling of the dependent variable,  classical Clapeyron's Theorem  can be  linked to
 degree 2 homogeneity of the quadratic energy density of linear elasticity.
Our   GCT  then emerges from the same type of  analysis but now using  a   scaling  reparametrization transformation,  involving both dependent and independent variables, which  does not leave the functional invariant but  multiplies it by
  a scalar factor.

As we have already mentioned, an intriguing   feature of the GCT  is that it naturally  combines  what is known in solid mechanics as the work performed by the  Piola   
and  Eshelby stresses \cite{eric77,silh97,gurt99,podi01,stma05,gfa10}. Moreover, we show that there exists a higher dimensional representation of GCT featuring a new tensor   \emph{combining} the Piola stress   and the Eshelby stress. To obtain such a representation  we first observe that  for parametric problems energy density is a 1-homogeneous function of  deformation gradient and therefore one can express the total energy along an extremal as a surface integral. Then we show that  one can  take a  general non parametric problem and always map it into a higher dimensional parametric problem  obtaining a desired boundary integral representation in the higher dimensional space.  This higher dimensional representation  allows one to give to the GCT  a  compact form involving the implied  combination of the Piola  and  Eshelby stresses.

 Since  our work is motivated by applications which allow for surfaces of jump discontinuities of field gradients we account  in our representation of GCT   of inherent lack of smoothness of extremals \cite{knst78,know79,ball87,grtrmms}.
  This issue becomes  particularly important in our  analysis of GCT applicable to  nonlinear elastodynamics  where  the formation of singularities is   unavoidable   despite the presence of  the overall variational structure \cite{dafHCL,ggks25}.

To highlight the utility of the obtained results, we present a series of illuminating case studies. In particular, we consider a generic  problem of void formation  in a nonlinear  elastic body where  the possibility of  scaling reparametrization  allows one to  study  infinitely small voids  in a finite domain  as  finite voids  in the infinite space and show that the use of GCT allows one to unify several previous partial results.  We show further that   GCT opens the way  to addressing  in a new way some  important problems of  material stability ranging from  fracture to   phase nucleation. 
We also provide evidence that another important set of applications of the GCT  is associated with the possibility to  rule out the presence of strong local minimizers that are not global in nonlinear elasticity in the case of hard device loading.

Finally, we mention again  that  since partial invarinace   of the type studied in this paper  is a rather general property of variational problems, the obtained  nonlinear versions of the  classical Clapeyron's Theorem   applies to  a broad class of  variational problems not at all related to elasticity theory. For instance,  our methods offer a unified approach to  Rellich and Pohozaev type identities normally used  in reaction diffusion framework.  In this perspective,  the obtained integral  identities should be viewed  as general results  within Calculus of Variations reflecting particular structure  of the corresponding energy functionals.  Notwithstanding the immense breadth of possible application of GCT in different disciplines, we have chosen to focus our presentation on problems of solid mechanics while dedicating only brief comments to other applications.

The paper is organized as follows.  Our Section~\ref{sec:noether}  contains the background material including the  introduction of inner and outer variations, the definitions of   Piola and Eshelby stresses, the analysis of general transformation families  for variational problems and the derivation of the generalized Noether's formula,  which we illustrate in a detailed study of multidimensional elasto-dynamics.

In  Section\ref{sec:parinv}  we specialize  the generalized Noether's formula to  the case of  partial symmetries with   parametric Lagrangians as a prominant example. We then  show how different  types of partial symmetries allow one  to use the general Noether's formula to  derive  various invariant integrals, for instance, the nonlinear versions of $J$, $L$ and $M$ integrals,  whose linearized analogs  have long   been known in  elasticity theory.  As a case study, we present   in this section   a  derivation of a   generalized version of Pohozaev's identities  and show how to use them to prove uniqueness in a broad class of  nonlinear boundary value problems.

The actual GCT is derived only in Section~\ref{sec:clapeyron} where we introduce the idea of scale invariance and derive the  integral identity resulting in the corresponding generalized Noether's formula. We then  map an original   non parametric
problem into a higher dimensional parametric one and take  advantage of  the ensuing scale invariance  to  obtain a hiher dimensional version of the GCT  in the extended space,  naturally  combining the Piola    and   Eshelby stresses.  In the same section we present   versions of GCT applicable in the presence of constraints and in the case when extremals exhibit surfaces of jump discontinuities of field gradients. 
A special place among what we call ``partial symmetries'' is occupied by $p$-homogeneity which we use, for instance, to   simplify GCT   in the case of  linear elasticity where the  energy density is a quadratic function of the  strain tensor. The case study in this section is again the elastodynamics where now we focus on the derivation  of  a version of GCT  that accounts for the emergence of shock waves.

Our  Section~\ref{sec:app} is devoted to various applications of the obtained results. In particular we show how the use of GCT allows one to formulate a new necessary condition  of metastability, establish the non-existence  of strong local minima,  and gain nontrivial insight into notoriously inaccessible quasi-convex envelopes of elastic energies.  We then use   GCT  to obtain a new general representation  of the energy  release associated with the void formation in nonlinear elasticity theory. In particular, we   provide a novel resolution of the well known ``Griffith's error'' \cite{sned46,eshe51,sile67,kesi96}. The corresponding linear elastic problem is discussed in full detail as a case study.

Our conclusions are summarized in Section~\ref{sec:conc}. Several results of purely technical nature are presented in the form of Appendixes in Section~\ref{sec:apps}.

\section{Preliminaries}
\setcounter{equation}{0}
\label{sec:noether}
We begin by introducing  a general variational problem of minimizing the energy functional 
\begin{equation}
E[\By]=\int_{\Omega}W(\Bx,\By(\Bx),\Grad\By(\Bx))d\Bx,
\label{non-param}
\end{equation}
where, using the language of elasticity theory,  we interpret the
function $\By:\GO\to\bb{R}^{m}$ as a deformation field defined on  the
reference (Lagrangian) configuration $\GO\subset\bb{R}^{n}$. With elastodynamics and other applications in mind, we  deal here  with a more
general framework where $m$ and $n$ are arbitrary, the energy density $W$ is allowed to depend on both  $\Bx$ and $\By$, and most importantly, where $\By(\Bx)$ is assumed to be only Lipschitz continuous, most often with only singularities of $\Grad\By(\Bx)$ being smooth surfaces of jump discontinuity, as occurs in the presence of shocks and during martensitic phase transitions.  Since there are many examples of Lipschitz extremals with more general singular sets, even for strictly convex energy densities \cite{degio68,necas77,ball82,sivo92,svya00,musv03}, in each case we will specify the extra regularity assumptions placed on $\By(\Bx)$, which will otherwise always be assumed Lipschitz continuous.

The image $\GO^{*}=\By(\GO)$ is called actual or Eulerian configuration. 
In   classical   nonlinear elasticity it is assumed that
$W(\Bx,\By,\BF)=W(\BF)$ does not depend explicitly on $\Bx$ and $\By$ and is objective, i.e.
$W(\BR\BF)=W(\BF)$ for all $\BR\in SO(3)$ and all $\BF$, $\det\BF>0$. However, here we are not making these assumptions, as our conclusions hold for a far more general class of energy densities.

We will assume, in general, that $W$ is of class $C^{1}$, if not globally, then
on the open set containing the range of $(\Bx,\By(\Bx),\Grad\By(\Bx))$ for any
configuration $\By$ under consideration. In the context of three-dimensional nonlinear elasticity, for example, this means that the range of $\Grad\By(\Bx)$ will be a closed subset of $\{\BF\in\bb{R}^{3\times 3}:\det\BF>0\}$, eliminating the difficulties typically arising when one considers Sobolev solutions, whose existence is guaranteed by \cite{ball7677}.
In some situations that we will identify explicitly, the regularity in the $\Bx$ variable will not be assumed, as in the case of a heterogenous or composite media.

\subsection{First variation of a graph}

We begin with a brief derivation of the   first variation of the energy.
Following Noether, we adopt a more geometric point of view and regard the energy density $W$ as a function on the tangent bundle $T\GG$, where
\[
  \GG=\{(\Bx,\By(\Bx)):\Bx\in\GO\}\subset \bb{R}^{m+n}
\]
is the graph of $\By(\Bx)$.
We will now examine the effect of the graph perturbation on the energy functional. Hence, we consider
a smooth family of Lipschitz maps $\BGF_{\Ge}:\bb{R}^{m+n}\to \bb{R}^{m+n}$, such that $\BGF_{0}(\Bx,\By)=(\Bx,\By)$. When $\Ge$ is sufficiently small, $\GG_{\Ge}=\BGF_{\Ge}(\GG)$ would still be a graph of some Lipschitz function $\Tld{\By}_{\Ge}$, and our goal is to derive the formula for the first variation
\begin{equation}
  \label{firstvar}
\Gd E=\left.\frac{dE[T\GG_{\Ge}]}{d\Ge}\right|_{\Ge=0},
\end{equation}
where
\begin{equation}
  \label{firstvar2}
E[T\GG_{\Ge}]=\int_{\GO_{\Ge}}W(\Tld{\Bx},\By_{\Ge}(\Tld{\Bx}),\Grad\By_{\Ge}(\Tld{\Bx}))d\Tld{\Bx}.
\end{equation}
To explain the  notations in   \eqref{firstvar2}, we write
\begin{equation}
  \label{trgrp}
\BGF_{\Ge}(\Bx,\By)=(\BX(\Bx,\By,\Ge),\BY(\Bx,\By,\Ge)),
\end{equation}
where $\BX$ and $\BY$ are Lipschitz functions of $(\Bx,\By)$ and of class $C^{1}$ in $\Ge$, such that
\[
\BX(\Bx,\By,0)=\Bx,\quad\BY(\Bx,\By,0)=\By.
\]
According to (\ref{trgrp}) $\Tld{\By}_{\Ge}(\Tld{\Bx})=\BY(\Bx,\By(\Bx),\Ge)$, in other words,
\begin{equation}
  \label{yeps}
  \By_{\Ge}(\BX(\Bx,\By(\Bx),\Ge))=\BY(\Bx,\By(\Bx),\Ge),
\end{equation}
and hence, we obtain   \eqref{firstvar2},  
where $\GO_{\Ge}$ is the image of $\GO$ under the transformation $\BX(\Bx,\By(\Bx),\Ge)$.

To compute $\Gd E$ given by (\ref{firstvar}) 
we first introduce the notations 
\[
\left.\Gd\Bx(\Bx,\By)=\dif{\BX(\Bx,\By,\Ge)}{\Ge}\right|_{\Ge=0},\quad
\left.\Gd\By(\Bx,\By)=\dif{\BY(\Bx,\By,\Ge)}{\Ge}\right|_{\Ge=0},
\]
even though in our subsequent derivations   we will only use the Lipschitz functions
\[
\Gd\Bx(\Bx)=\Gd\Bx(\Bx,\By(\Bx)),\qquad \Gd\By(\Bx)=\Gd\By(\Bx,\By(\Bx)).
\]
We then make in  \eqref{firstvar2}   the change of variables $\Tld{\Bx}=\BX(\Bx,\By(\Bx),\Ge)$, and obtain
\[
E[T\GG_{\Ge}]=\int_{\GO}W(\BX(\Bx,\By(\Bx),\Ge),\By_{\Ge}(\BX(\Bx,\By(\Bx),\Ge)),\Grad\By_{\Ge}(\BX(\Bx,\By(\Bx),\Ge)))\det(\Grad\BX)d\Bx,
\]
where $\Grad\BX=\Grad_{\Bx}(\BX(\Bx,\By(\Bx),\Ge))$.
Differentiating (\ref{yeps}) in $\Bx$ we obtain
\[
\Grad\By_{\Ge}(\BX(\Bx,\By(\Bx),\Ge))\Grad\BX=\Grad\BY
\]
Therefore,
\begin{equation}
  \label{deform}
  E[T\GG_{\Ge}]=\int_{\GO}W(\BX,\BY,\Grad\BY(\Grad\BX)^{-1})\det(\Grad\BX)d\Bx,
\end{equation}
where $\BX$ and $\BY$ are evaluated at $(\Bx,\By(\Bx),\Ge)$.
We note that
\[
\left.\dif{\Grad\BX}{\Ge}\right|_{\Ge=0}=\Grad_{\Bx}\Gd\Bx(\Bx,\By(\Bx)),\quad \left.\dif{\Grad\BY}{\Ge}\right|_{\Ge=0}=\Grad_{\Bx}\Gd\By(\Bx,\By(\Bx)).
\]
Now, we can differentiate (\ref{deform}) in $\Ge$ at $\Ge=0$ and obtain
\[
    \Gd E=\int_{\GO}\{W_{\Bx}\cdot\Gd\Bx+W_{\By}\cdot\Gd\By+\av{W_{\BF},\Grad\Gd\By-\Grad\By\Grad\Gd\Bx}+
  W\Div\Gd\Bx\}d\Bx, 
\]
where $\Gd\Bx=\Gd\Bx(\Bx,\By(\Bx))$ and $\Gd\By=\Gd\By(\Bx,\By(\Bx))$.

We now introduce  important notations which will be essentially used in what follows. First, we define the tensor field 
\begin{equation}
  \BP(\Bx,\By,\BF)=W_{\BF}(\Bx,\By,\BF), 
\label{Esh-tensor1}
\end{equation}
to which we  refer as  the  \emph{Piola stress tensor} using its  narrow meaning in classical   elasticity theory; in a general field theory it is known as canonical momentum or current tensor.  Second, we define another tensor field 
\begin{equation}
  \BP^{*}(\Bx,\By,\BF)=W(\Bx,\By,\BF)\BI_{n}-\BF^{T}\BP(\Bx,\By,\BF),
\label{Esh-tensor}
\end{equation}
and refer to it as the \emph{Eshelby tensor}, again, referring to a more limited notion used in nonlinear  elasticity;  in a general field theory this object  is known as the energy momentum tensor. The corresponding background material, explaining the origin of these definitions in more detail,  can be found in the monographs \cite{silh97,gurt00,gfa10,maug16,gori17,krro19,seep20}. 

Using these definitions we can rewrite the formula for $\Gd E$ as follows
\begin{equation}
  \label{firstvar1}
  \Gd E=\int_{\GO}\{W_{\By}\cdot\Gd\By+W_{\Bx}\cdot\Gd\Bx+\av{\BP,\Grad\Gd\By}
  +\av{\BP^{*},\Grad\Gd\Bx}\}d\Bx,
\end{equation}
where $\BP$ and $\BP^{*}$ are evaluated at $(\Bx,\By(\Bx),\Grad\By(\Bx))$.
If all the functions above, $W$, $\By$, $\BX$, and $\BY$ are of class $C^{2}$, we can integrate by parts and rewite (\ref{firstvar1}) as
\begin{equation}
  \label{deriv}
\Gd E=\int_{\GO}\{\mathfrak{E}_{W}(\Bx)\cdot\Gd\By+\mathfrak{E}^{*}_{W}(\Bx)\cdot\Gd\Bx\}d\Bx+
\int_{\dOm}\{\BP\Bn\cdot\Gd\By+\BP^{*}\Bn\cdot\Gd\Bx\}d\Bx,
\end{equation}
where
\begin{equation}
  \label{EL}
  \mathfrak{E}_{W}(\Bx)=W_{\By}(\Bx,\By(\Bx),\Grad\By(\Bx))-\Div \BP(\Bx,\By(\Bx),\Grad\By),
\end{equation}
and
\begin{equation}
  \label{ELT}
\mathfrak{E}^{*}_{W}(\Bx)=W_{\Bx}(\Bx,\By(\Bx),\Grad\By(\Bx))-\Div\BP^{*}(\Bx,\By(\Bx),\Grad\By(\Bx)).
\end{equation}
The fields $ \mathfrak{E}_{W}(\Bx)$ and $\mathfrak{E}^{*}_{W}(\Bx)$ are linked through a relation  derived originally by  Noether \cite{noether18}:
\begin{equation}
  \label{Noether0}
\mathfrak{E}^{*}_{W}(\Bx)=-(\Grad\By)^{T}\mathfrak{E}_{W}(\Bx),
\end{equation}
Indeed,
\begin{multline*}
  \mathfrak{E}^{*}_{W}(\Bx)_{\Ga}=\dif{W}{x^{\Ga}}-\dif{}{x^{\Ga}}[W(\Bx,\By(\Bx),\Grad\By)]+
  \mix{y^{i}}{x^{\Ga}}{x^{\Gb}}P_{i}^{\Gb}+\dif{y^{i}}{x^{\Ga}}\dif{P_{i}^{\Gb}}{x^{\Gb}}=\\
\dif{W}{x^{\Ga}}-\dif{W}{x^{\Ga}}-\dif{W}{y^{i}}\dif{y^{i}}{x^{\Ga}}-P_{i}^{\Gb}\mix{y^{i}}{x^{\Ga}}{x^{\Gb}}+
  \mix{y^{i}}{x^{\Ga}}{x^{\Gb}}P_{i}^{\Gb}+\dif{y^{i}}{x^{\Ga}}\dif{P_{i}^{\Gb}}{x^{\Gb}}=\\
  -\dif{W}{y^{i}}\dif{y^{i}}{x^{\Ga}}+\dif{y^{i}}{x^{\Ga}}\dif{P_{i}^{\Gb}}{x^{\Gb}}=
  -\dif{y^{i}}{x^{\Ga}}\mathfrak{E}_{W}(\Bx)_{i},
\end{multline*}
which is (\ref{Noether0}) written in components.

 The  massive cancellations of terms in the
above derivation of (\ref{Noether0}) suggests that there is a deeper underlying reason for (\ref{Noether0}) to hold. We will therefore rederive it as a corollary of formula (\ref{firstvar1}), valid
for any Lipschitz $\bra{\By}\in W^{2,1}_{\rm loc}(\GO;\bb{R}^{m})$.
We present the result in the form of the  theorem:
\begin{theorem}
  \label{th:NI}
  Suppose $W$ is of class $C^{1}$, and $\bra{\By}\in W^{2,1}_{\rm loc}(\GO;\bb{R}^{m})\cap
  W^{1,\infty}(\GO;\bb{R}^{m})$. Then (\ref{Noether0})  holds in the
  sense of equality of distributions,
  where $(\Grad\By)^{T}\mathfrak{E}_{W}$ is understood as the distribution 
  \[
((\Grad\By)^{T}\mathfrak{E}_{W})_{i}=\dif{\bra{y}^{k}}{x^{i}}W_{y^{k}}-\dif{}{x^{j}}\left[\dif{\bra{y}^{k}}{x^{i}}P_{k}^{j}\right]+\mix{\bra{y}^{k}}{x^{i}}{x^{j}}P_{k}^{j},
\]
which makes sense, since our
regularity assumptions on $\bra{\By}$ imply that the third term above is in $L^{1}_{\rm loc}(\GO)$.
\end{theorem}
\begin{proof}
The idea is to construct a sufficiently large family of nontrivial deformations of
the $(\Bx,\By)$-space none of whose members change the graph of
$\bra{\By}(\Bx)$. Let $\BGf\in C_{0}^{\infty}(\GO;\bb{R}^{m})$ be arbitrary. Then
$\BX(\Bx,\Ge)=\Bx+\Ge\BGf(\Bx)$ is a diffeomorphism of $\GO$ onto itself for
sufficiently small $|\Ge|$, so that $\GO_{\Ge}=\GO$. Then the deformation
\[
\BGF_{\Ge}(\Bx,\By)=(\BX(\Bx,\Ge),\By+\bra{\By}(\BX(\Bx,\Ge))-\bra{\By}(\Bx)),
\]
leaves the graph of $\bra{\By}(\Bx)$ invariant. Indeed,
\[
  \BGF_{\Ge}(\Bx,\bra{\By}(\Bx))=(\BX(\Bx,\Ge),\bra{\By}(\Bx)+\bra{\By}(\BX(\Bx,\Ge))
  -\bra{\By}(\Bx))=(\BX(\Bx,\Ge),\bra{\By}(\BX(\Bx,\Ge))).
\]
Since $\BX(\Bx,\Ge)$ is a diffeomorphism of $\GO$, we conclude that $\BGF_{\Ge}(\GG)=\GG$, where $\GG$ is the graph of $\bra{\By}(\Bx)$. Thus, $\Gd E=0$ on the \lhs\ of (\ref{firstvar1}). We also compute
\[
  \Gd\Bx=\BGf,\quad\Gd\By=(\Grad\bra{\By})\BGf.
\]
In that case formula (\ref{firstvar1}) reads
\begin{equation}
  \label{preNI}
  0=\int_{\GO}\{W_{\Bx}\cdot\BGf+\av{\BP^{*},\Grad\BGf}+W_{\By}\cdot(\Grad\bra{\By})\BGf+
  \av{\BP,\Grad((\Grad\bra{\By})\BGf)}\}d\Bx,
\end{equation}
writing
\[
  \av{\BP,\Grad((\Grad\bra{\By})\BGf)}=P_{k}^{j}\dif{\bra{y}^{k}}{x^{i}}\dif{\phi^{i}}{x^{j}}
  +P_{k}^{j}\mix{\bra{y}^{k}}{x^{i}}{x^{j}}\phi^{i},
\]
we see that (\ref{preNI}) can be written as $0=\av{\mathfrak{E}^{*}_{W}|\BGf}+\av{(\Grad\By)^{T}\mathfrak{E}_{W}|\BGf}$, where the angular bracket notation denotes the action of distributions on test functions.
\end{proof}
\begin{remark} When $\bra{\By}$ is merely Lipschitz continuous, the proof of
  Theorem~\ref{th:NI} breaks down, since in this case
$\Grad\Gd\By$ is not an $L^{1}_{\rm loc}$ function and the derivation
of formula (\ref{firstvar1}) becomes invalid. Moreover, formula
(\ref{firstvar1}) itself does not make sense for a distributional
$\Grad\Gd\By$. Nonetheless, Theorem~\ref{th:NI} can still be used on
subdomains of $\GO$, where $\bra{\By}$ is of class $W^{2,1}$.
\end{remark}
Using this remark, we can apply the Noether formula
(\ref{deriv}) to the functional
\begin{equation}
  \label{ssf}
  E[\By]=\int_{\GO }W(\Grad\By)d\Bx
\end{equation}
where the field $\By(\Bx)$ is assumed to be of class $C^{2}$, outside of a smooth surface $\GS$ of jump discontinuity of $\Grad\By$ where the distributional parts of the Lagrangian quantities $\mathfrak{E}_{W}$ and $\mathfrak{E}^{*}_{W}$ supported on the surface $\GS$ must be taken into account.
The  identity (\ref{firstvar1}) can be then rewritten in the form similar to (\ref{deriv})
\begin{multline}
  \label{ClapId1}
\Gd E=-\int_{\GO}\{(\Div\BP)_{\rm reg}\cdot\Gd\By
  +(\Div\BP^{*})_{\rm reg}\cdot\Gd\Bx\}d\Bx\\
  +\int_{\dOm}\{\BP\Bn\cdot\Gd\By+\BP^{*}\Bn\cdot\Gd\Bx\}dS
-\int_{\GS}\{\jump{\BP}\Bn \cdot\Gd\By+\jump{\BP^{*}}\Bn \cdot\Gd\Bx\}dS.
\end{multline}
where
$(\Div\BP)_{\rm reg}$ and $(\Div\BP^{*})_{\rm reg}$ are $C^{1}$ functions on $\GO\setminus S$, suffering a jump discontinuity across $\GS$, defined by $\Div\BP(\Bx)$ and $\Div\BP^{*}(\Bx)$, $\Bx\in\GO\setminus \GS$.
If we follow the standard convention that the unit normal to a surface $\GS$ always points from the ``$-$'' side of $\GS$ to its ``$+$'' side, then 
\[\jump{\BP}=\BP_{+}(\Bx)-\BP_{-}(\Bx),\quad\Bx\in\GS.\]

\subsection{Extremal and stationary configurations}

Our primary interest lies in problems of Calculus of Variations arising from elastostatics and elastodynamics. The latter is governed by the action principle, whereby 
\begin{equation}
\Gd E=0  \label{ELeq1}
 \end{equation}
 for all outer variations $\Gd\By\in C_{0}^{\infty}(\GO;\bb{R}^{m})$, and $\Gd\Bx=0$;  the former---by the energy minimization procedure, one of whose consequences is (\ref{ELeq1}) for all variations $\Gd\By\in C_{0}^{\infty}(\GO;\bb{R}^{m})$ and $\Gd\Bx\in C_{0}^{\infty}(\GO;\bb{R}^{n})$. In both cases condition \eqref{ELeq1}   requires that 
\begin{equation}
    \label{ELeq}
\mathfrak{E}_{W}(\Bx)=0,
  \end{equation}
which serves in such theories as the main governing equation.
\begin{definition}
  \label{def:equil}
  We say that the Lipschitz configuration $\By(\Bx)$ is an extremal, if it satisfies the Euler-Lagrange equation (\ref{ELeq})
  in the sense of distributions.
\end{definition}
This definition allows for singularities of $\Grad\By$ corresponding, for instance, to shock waves and phase boundaries \cite{trus93}.
We recall that  stationarity  is fully compatible with jump discontinuities of this type as it  places restrictions on such singularities   but does not rule them out. This issue   is discussed in more detail in Section~\ref{sub:dyn}.

An important  implication of Thorem~\ref{th:NI} is that for any Lipschitz extremal, which is of class $W^{2,1}_{\rm loc}$ outside of a closed nowhere dense singular set $S$, $\mathfrak{E}^{*}_{W}(\Bx)$ is a distribution supported on $S$. In the special case when $\Grad\By(\Bx)$ suffers a jump discontinuity across a smooth surface $\GS\subset\GO$, outside of which $\By$ is of class $C^{2}$, the distribution $\mathfrak{E}^{*}_{W}(\Bx)$ has a particularly simple form. To derive it, we take an arbitrary point $\Bx_{0}\in\GS$, choose a direction of the normal $\Bn_{0}$ to $\GS$ at $\Bx_{0}$ and a radius $r>0$, such that $\GS$ splits $B(\Bx_{0},r)$ into two disjoint subdomains $B^{+}(\Bx_{0},r)$ and $B ^{-}(\Bx_{0},r)$ , with $\Bn_{0}$ pointing from $B^{-}(\Bx_{0},r)$ and into $B^{+}(\Bx_{0},r)$, by our convention. Then there is a unique choice of the smooth normal $\Bn(\Bx)$ to $\GS$ at $\GS\cap B(x_{0},r)$, such that $\Bn(\Bx_{0})=\Bn_{0}$. Then for any test function $\BGf\in C_{0}^{\infty}(B(\Bx_{0},r))$,
\begin{multline*}
    \av{\mathfrak{E}^{*}_{W}(\Bx),\BGf}=\int_{B(\Bx_{0},r)}\{W_{\Bx}\cdot\BGf+\av{\BP^{*},\Grad\BGf}\}d\Bx=\\
\int_{B^{+}(\Bx_{0},r)}\{W_{\Bx}\cdot\BGf+\av{\BP^{*},\Grad\BGf}\}d\Bx+\int_{B^{-}(\Bx_{0},r)}\{W_{\Bx}\cdot\BGf+\av{\BP^{*},\Grad\BGf}\}d\Bx. 
\end{multline*}
We can now integrate by parts on each of the domains $B^{\pm}(\Bx_{0},r)$, taking (\ref{Noether0}) and $\mathfrak{E}_{W}(\Bx)=0$ into account:
\begin{equation} \label{jumpP1}
  \av{\mathfrak{E}^{*}_{W}(\Bx),\BGf}=-\int_{\GS}\jump{\BP^{*}}\Bn\cdot\BGf\,dS.
 \end{equation}

The vector field $\jump{\BP}^{*}\Bn$ in \eqref{jumpP1} can be simplified further, taking into account again that $\mathfrak{E}_{W}(\Bx)=0$, which implies that
\begin{equation}
  \label{jumpP}
   \jump{\BP}\Bn=\Bzr.
 \end{equation}
In this case we can write 
 \[
\jump{\BP^{*}}\Bn=\jump{W}\Bn-\jump{\BF^{T}\BP}\Bn=\jump{W}\Bn-\jump{\BF}^{T}\BP\Bn,
\]
where we use the usual notation $\BF(\Bx)=\Grad\By(\Bx)$.
We now recall that on the smooth surface of jump discontinuity of $\Grad\By$ the following Hadamard relations hold:
\begin{equation}
  \label{Hadamard}
  \jump{\BF}=\Ba\otimes\Bn, 
\end{equation}
where $\Bn(\Bx)$ is the unit normal on an $n-1$-dimensional surface $\GS\subset\bb{R}^{n}$ and $\Ba:\GS\to\bb{R}^{m}$ is a $C^{1}$ vector field on $\GS$. Then
\[
\jump{\BP^{*}}\Bn=(\jump{W}-\BP\Bn\cdot\Ba)\Bn=(\jump{W}-\av{\BP_{\pm},\jump{\BF}})\Bn=p^{*}\Bn.
\]
Hence,
\begin{equation}
  \label{EWstar}
  \mathfrak{E}^{*}_{W}(\Bx)=-p^{*}_{\GS}\Bn\Gd_{\GS}, 
\end{equation}
where $\Gd_{\GS}$ is a distribution defined by
\[
\av{\Gd_{\GS},\phi}=\int_{\GS}\phi(\Bx)dS,\quad\forall\phi\in C_{0}^{\infty}(\GO).
\]
In (\ref{EWstar}) we also introduced a continuous scalar field
\begin{equation}
  \label{pstar}
  p^{*}_{\GS}=\jump{W}-\av{\BP_{\pm},\jump{\BF}}
\end{equation}
on the smooth surface of jump discontinuity $\GS$. Here $\BP_{\pm}$ indicates that it does not matter whether one uses $\BP_{+}$ or $\BP_{-}$ in the formula for $p^{*}_{\GS}$, as long as one uses this formula for an extremal, so that (\ref{jumpP}) holds.

As we have already mentioned, in static problems we are usually interested in minima of energy functionals, subject to specified \bc s. Then, Lipschitz minimizers of (\ref{non-param}) must satisfy $\Gd E=0$ for all variations
$\Gd\Bx\in C_{0}^{\infty}(\GO;\bb{R}^{n})$, and $\Gd\By\in C_{0}^{\infty}(\GO;\bb{R}^{m})$. Hence, in addition to (\ref{ELeq}) the solutions of such problems must also satisfy
\begin{equation}
  \label{stationary}
\mathfrak{E}^{*}_{W}(\Bx)=0,
\end{equation}
again understood in the sense of distributions.
\begin{definition}
  \label{def:exst}
We say that a Lipschitz configuration $\By(\Bx)$ is \emph{stationary} if
it is an extremal in the sense of Definition~\ref{def:equil}, and additionally satisfies (\ref{stationary}).
\end{definition}

Note that in view of Theorem   \ref{th:NI} if $\bra{\By}\in W^{2,1}_{\rm loc}(\GO;\bb{R}^{m})\cap
  W^{1,\infty}(\GO;\bb{R}^{m})$ is an extremal, then it is stationary
  in the sense of Definition~\ref{def:exst}.  Indeed $\Grad\bra{\By}\BGf\in W_{0}^{1,1}(\GO;\bb{R}^{m})$ and
  \[
\int_{\GO}\{W_{\By}\cdot\BGy+\av{\BP,\Grad\BGy}\}d\Bx=0,\quad\forall\BGy\in C_{0}^{\infty}(\GO;\bb{R}^{m}).
\]
By the density argument it follows that
 \[
\int_{\GO}\{W_{\By}\cdot\BGy+\av{\BP,\Grad\BGy}\}d\Bx=0,\quad\forall\BGy\in W_{0}^{1,1}(\GO;\bb{R}^{m}).
\]
Taking $\BGy=\Grad\bra{\By}\BGf$ we obtain
\[
\av{(\Grad\bra{\By})^{T}\mathfrak{E}_{W}|\BGf}=0\quad\forall\BGf\in C_{0}^{\infty}(\GO;\bb{R}^{n}).
\]
Theorem~\ref{th:NI} then implies that (\ref{stationary}) holds.
  
Note also that  formula (\ref{EWstar}) implies that for
stationary extremals with a smooth surface of jump discontinuity $\GS$ we must have
\begin{equation}
  \label{pst0}
  p^{*}_{\GS}=\jump{W}-\av{\BP_{\pm},\jump{\BF}}=0.
\end{equation}
In this sense 
 stationarity,   is not a consequence
 of equilibrium even for piecewise smooth $\By(\Bx)$.

While we introduced  in \eqref{Esh-tensor1} and \eqref{Esh-tensor}  two important tensorial fields, $\BP$ and $\BP^{*}$, which appear for instance   in   \eqref{deriv}, their physical meaning has not been explained. 
Observe first that while Piola stress $\BP$  acts on the virtual displacements $\Gd\By$ of the points on the boundary of the domain in the actual (Eulerian) space, the Eshely stress $\BP^{*}$  acts on the virtual changes of the shape of the domain $\GO$ in the reference (Lagrangian) space. The former is usually interpreted as the action of  ``physical forces'', and the latter -- as the action of ``configurational forces''. Therefore, we can conclude that  the two terms in the surface integral in \eqref{deriv}  represent the   work of both physical and configurational forces, see for instance, \cite{gfa10}.
More specifically, the implied two energy increments  can be associated with physical stress, respectively configurational prestress.  Here, while the meaning of the physical stress does not require any clarifications as it is related to the work of applied forces,  we associate  configurational prestress with the work that has been stored during the \emph{creation} of the physical body. To elucidate the  implied  fundamental  difference between $\BP$ and $\BP^{*}$ we present below to simple illustrative examples.

\begin{example}
  \label{ex0}
\end{example}

Consider a one-dimensional body, given in Lagrangian coordinates as an interval $[0,L]$.
The stored energy is  
\begin{equation}
E[u]=\int_{0}^{L}W(u')dx,
\end{equation}
where $u(x)$ is the displacement field. We assume that the energy density function $W(\Gve)$ is a smooth, even, strictly convex function, such that $W'(0)=0$, corresponding to the absence of the physical prestress. The configurational prestress is encoded by the residual energy $W_{*}=W(0)>0$. 
The relationship between the physical (Piola) stress
 \begin{equation}
P=W'(\Gve) 
\end{equation}
 and the configurational (Eshelby) stress \begin{equation}
P^{*}=W(\Gve)-\Gve W'(\Gve),
\end{equation}
 where the strain is  $\Gve=u'(x)$  is shown in the left panel of Fig.~\ref{fig:Clap1D}.
\begin{figure}[t]
  \centering
  \includegraphics[scale=0.25]{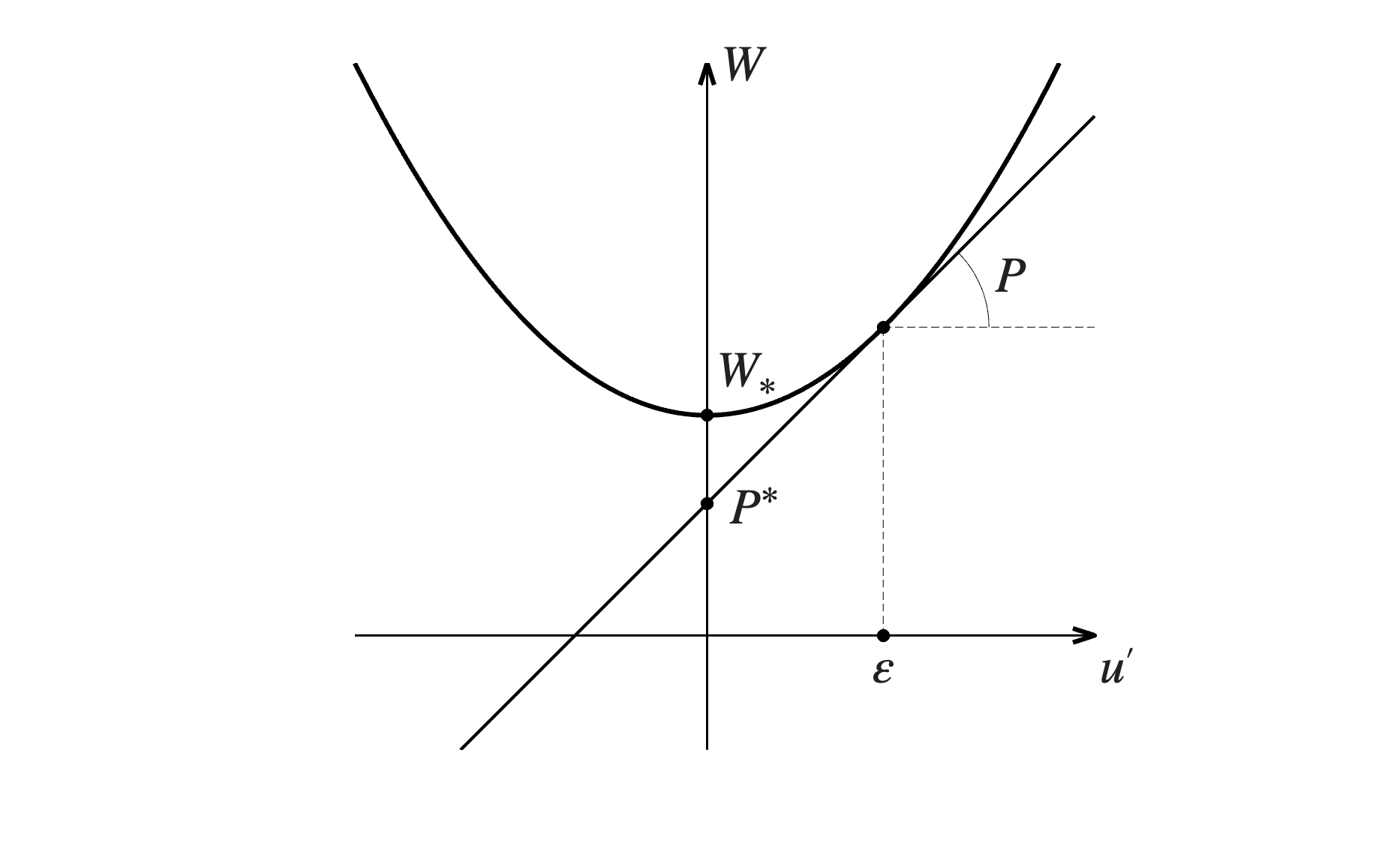}
  \includegraphics[scale=0.25]{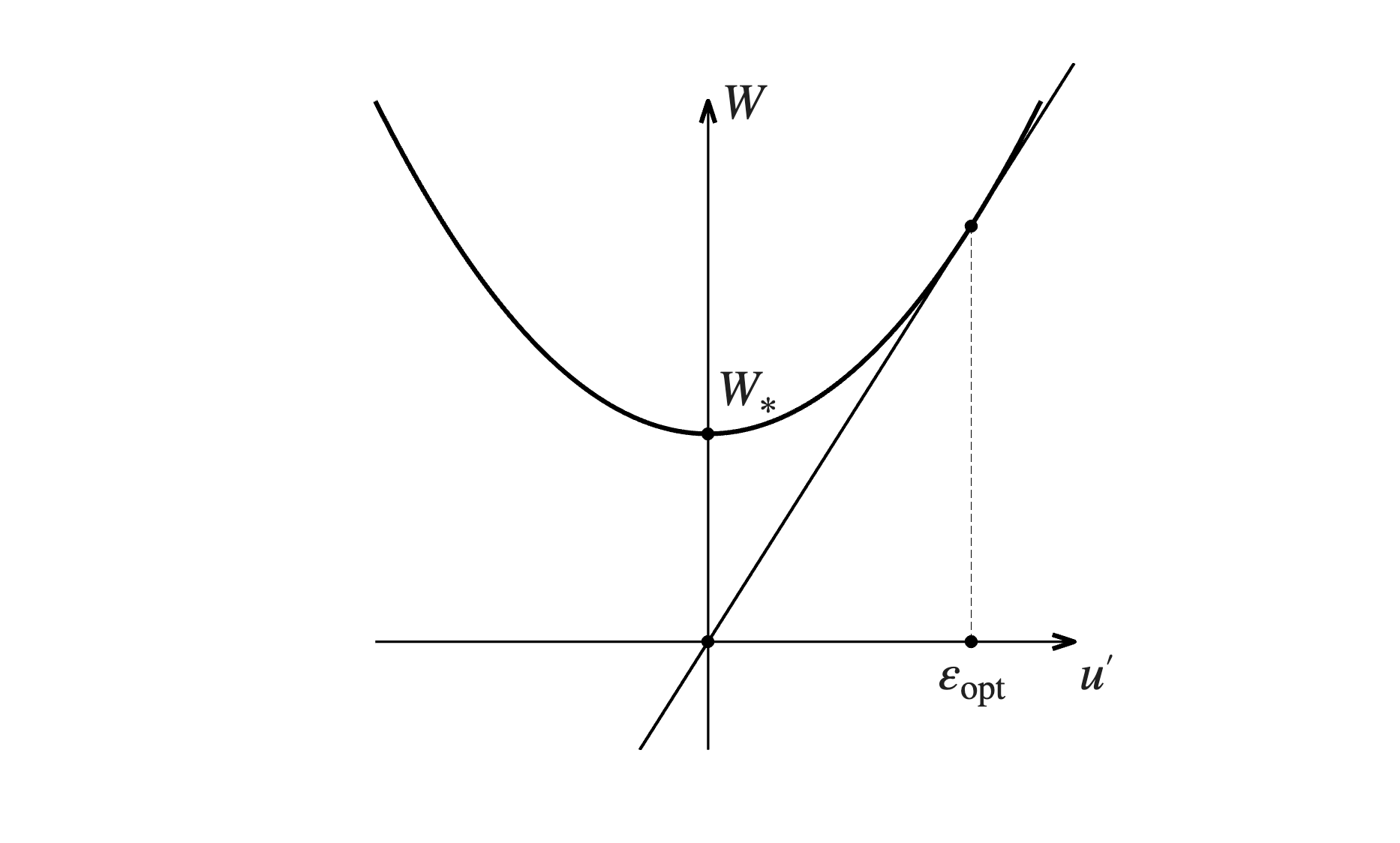}
  \caption{Relations between physical stress $P$, strain $\Gve$ and the configurational stress $P^{*}$ in equilibrium.}
  \label{fig:Clap1D}
\end{figure}
Behind  this figure is the observation that in equilibrium we always have $(W'(\Gve))'=0$ and due to the strict monotonicity of $W'(\Gve)$, expressing the strict convexity of $W(\Gve)$, all equilibrium configurations are the ones where $\Gve$ is constant on $[0,L]$. Note that while $P^{*}$ has the meaning of energy, $P$ is  a derivative of energy with respect to  $\Gve$.

Consider now a slightly different situation, where the deformation of points on $\dOm$ is prescribed,
but the domain $\GO$ itself is not fixed. In our 1D setting we can then consider that the numbers
$U_{0}=u(0)$ and $U_{1}=u(L)$ are prescribed, while  $L$ can vary. 

From a general point of view, in that case the admissible outer variations $\Gd\By$ in (\ref{deriv}) must vanish on $\dOm$, but the boundary values of admissible inner variations $\Gd\Bx$ are no longer restricted in any way. Then the vanishing of the first variation $\Gd E$ of the energy is equivalent to the Euler-Lagrange equation (\ref{ELeq}) and the Noether  equation (\ref{stationary}) augmented with the displacement boundary conditions together with the condition 
\begin{equation}
  \label{optimshape}
  \BP^{*}\Bn=0\text{ on }\dOm.
\end{equation}
This  leads to    a free boundary problem, where the shape and size of $\GO$ must be determined, in addition to the deformation $\By(\Bx)$.
\begin{figure}[t]
  \centering
  \includegraphics[scale=0.25]{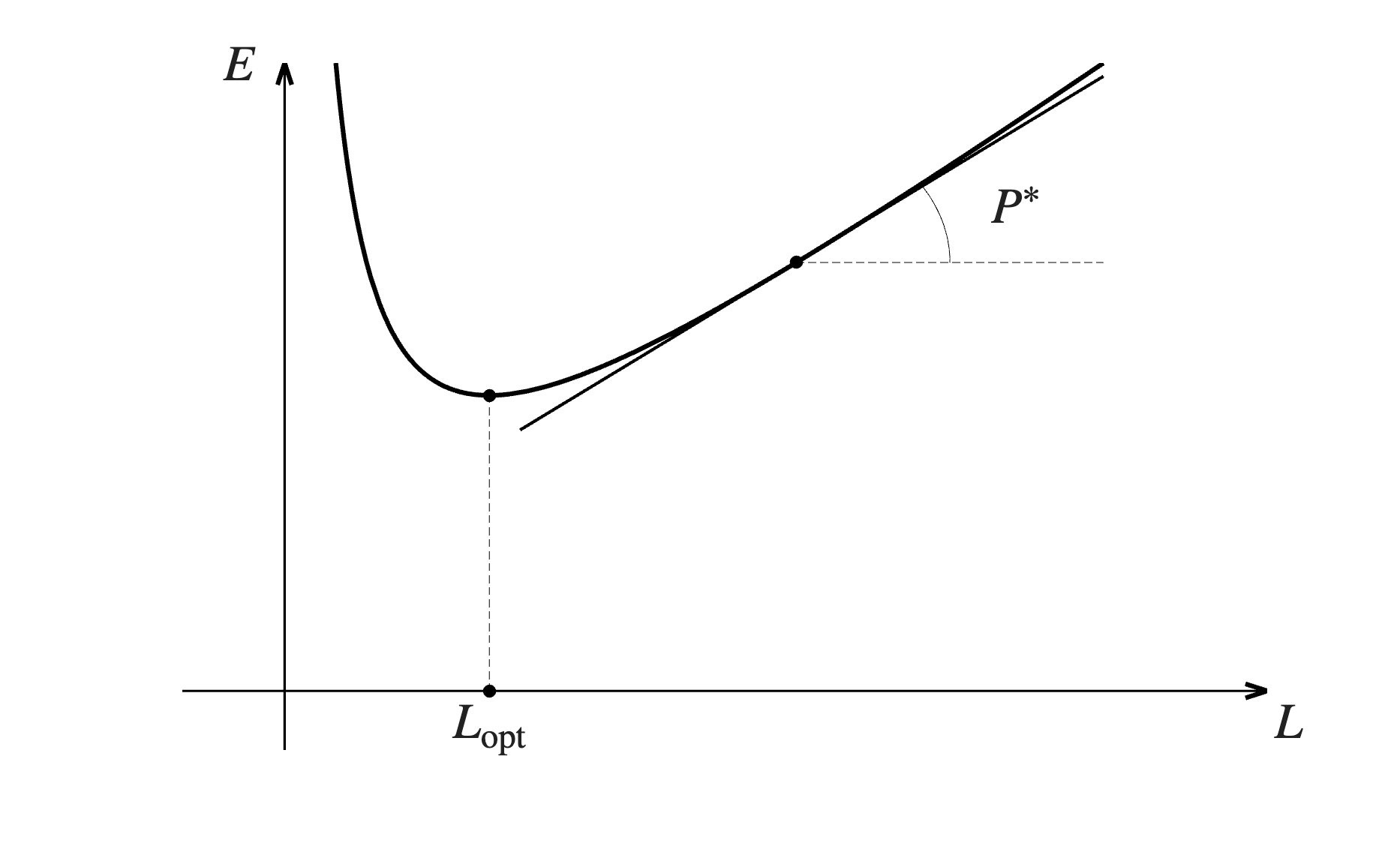}
  \caption{Energy as a function of the domain shape, when boundary displacements are fixed.}
  \label{fig:EofL}
\end{figure}

In our one dimensional example we can conclude that the energy extremum as a function of the shape/size of the domain is attained when $P^{*}=0$. The right panel in Fig.~\ref{fig:Clap1D} shows how the energy density function $W(\Gve)$ determines the value of $\Gve_{\rm opt}$ of the optimal strain.
This gives us the value
\[
L_{\rm opt}=\frac{U_{1}-U_{0}}{\Gve_{\rm opt}}.
\]
If $U_{0}$ and $U_{1}$ are prescribed, then, by strict convexity of $W(\Gve)$, there is a unique value of the constant equilibrium strain
\[
\Gve=\frac{U_{1}-U_{0}}{L}.
\]
Thus, the energy as a function of $L$ is
\[
E=LW\left(\frac{U_{1}-U_{0}}{L}\right).
\]
It is easy to see that
\[
\frac{dE}{dL}=P^{*},\qquad\frac{d^{2}E}{dL^{2}}=\frac{\Gve^{2}W''(\Gve)}{L}\geq 0,
\]
which shows that $P^{*}=0$ is indeed the condition of extremality of $E(L)$, and that $L_{\rm opt}$ delivers the minimum of the energy. The fact that $P^{*}\neq 0$ the system can still lower its energy by either growth or resorption    is illustrated in Fig.~\ref{fig:EofL}.

\begin{example}
  \label{ex1}
\end{example}
 Here we show how to construct   the simplest  $n$ dimensional   configuratiuonally prestressed elastic body.  For analytical transparency we limit our attention to spherically symmetric configurations and geometrically linear elasticity.
 
 Specifically,
  the goal is to obtain a  nontrivial radial stationary extremal  on the ball $B(0,R)\subset\bb{R}^{n}$  of the functional 
\begin{equation}
  \label{Erad}
  E[\Bu]=\int_{B(0,R)}W(\Bx,e(\Bu))d\Bx,
\end{equation}
where the notation $e(\Bu)$ was defined in (\ref{eofu}),
under the condition that the ``body'' $B(0,R)$ is not loaded by the  ``physical'' forces in the sense that the extremal  satisfies the boundary condition
\begin{equation}
 \label{Erad1}
  \BP\Bn=0,\quad|\Bx|=R. 
\end{equation}
Note that \eqref{Erad1} does not exclude the possibility that on the boundary we still have ``configurational'' loading in the sense that  
\begin{equation}
  \label{tracfree2}
  \BP^*\Bn\neq0,\quad|\Bx|=R.
\end{equation} 

For our goal, it is sufficient to choose the simplest  linearly elastic quadratic energy density of the form  
\begin{equation}
  \label{Wdisloc}
  W(\Bx,\BGve)=\hf|\BGve -\BGve_0(\Bx)|^{2},
\end{equation}
where $\BGve$ is a symmetric $n\times n$ matrix. The idea is to use the incompatible inhomogeneity $\BGve_0(\Bx)$ in the sense that 
\begin{equation}
  \label{Wdisloc1}
{\rm  Ink} [\BGve_0(\Bx)] \neq0  
\end{equation}
where
\[
  {\rm  Ink} [\BGve(\Bx)]_{ijkl}=\mix{\Gve_{ij}}{x^{k}}{x^{l}}-\mix{\Gve_{ik}}{x^{j}}{x^{l}}+
  \mix{\Gve_{lk}}{x^{i}}{x^{j}}-\mix{\Gve_{lj}}{x^{i}}{x^{k}}.
\]
It is clear that   \eqref{Wdisloc1} eliminates the trivial solution $\BGve(\Bx)=\BGve_{0}$ with zero energy.
\begin{figure}[t]
 \centering
 \includegraphics[scale=0.2]{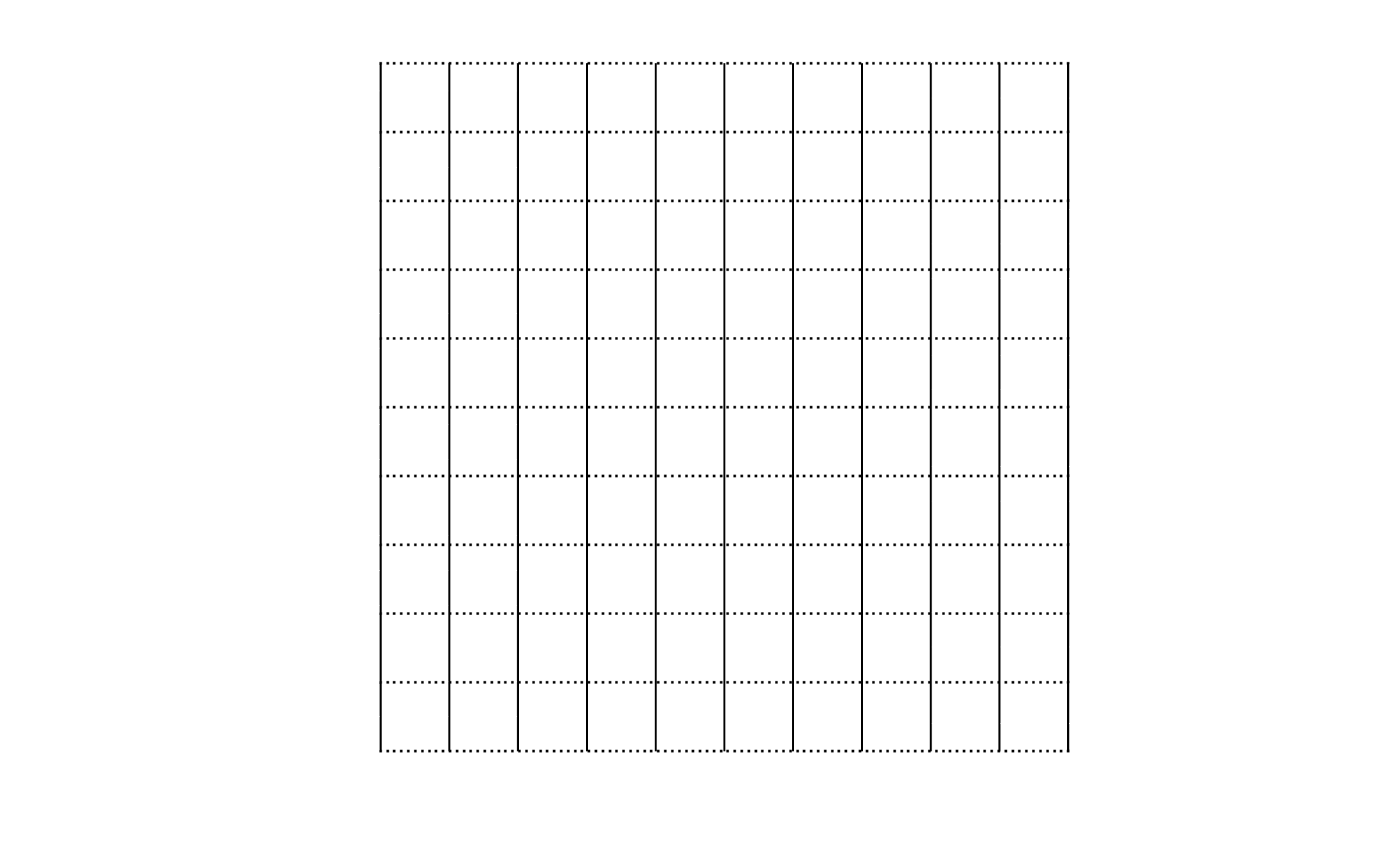}
 \includegraphics[scale=0.2]{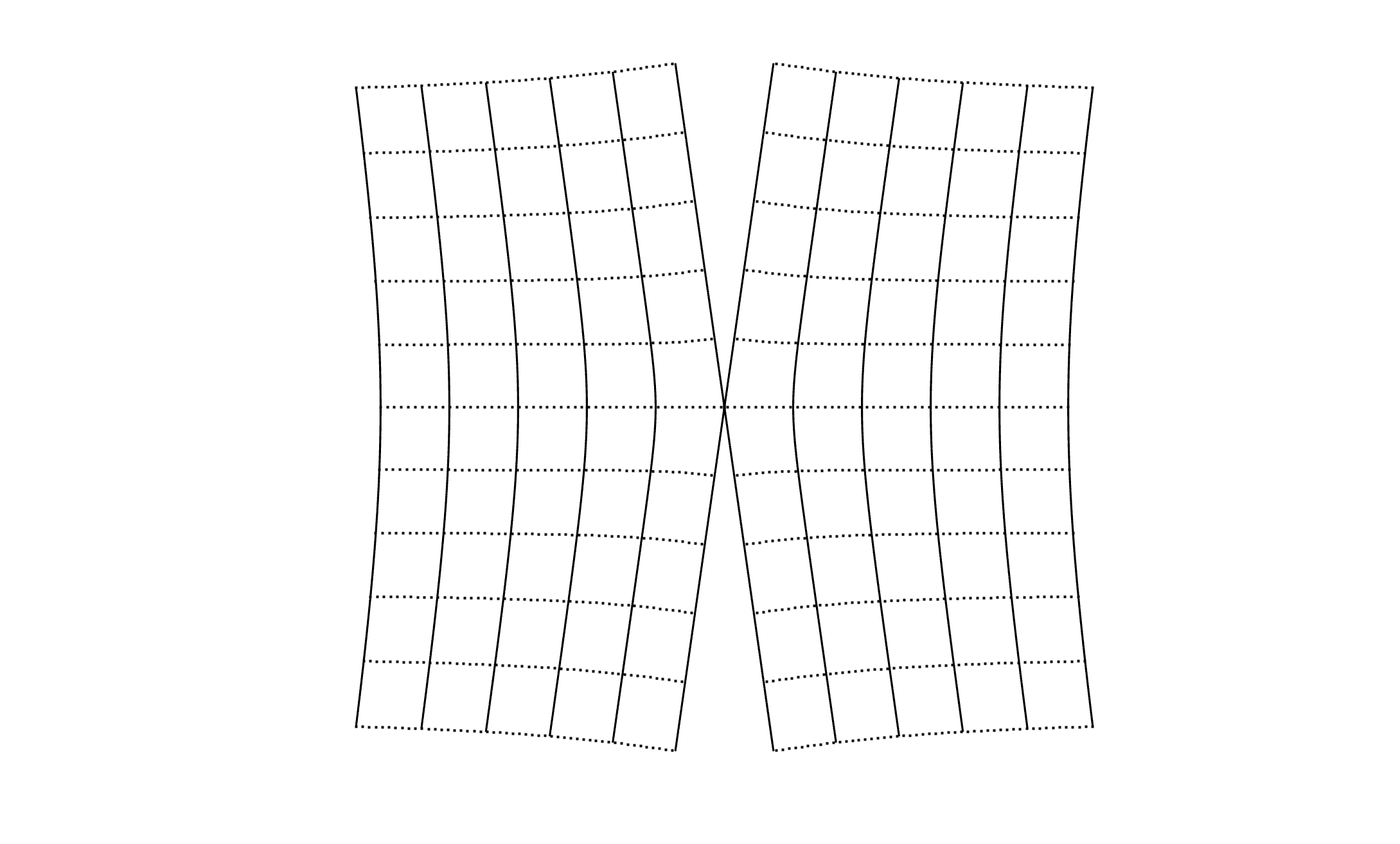}
 \caption{A perfect 2D grid  represented on the left panel is transformed  by the kinematically incompatible deformation strain field  (\ref{tracfree3}) into the discontinuous deformed grid shown on the right panel.  }
 \label{fig:discl}
\end{figure}
We interpret the inhomogeneity  $\BGve_0(\Bx)$ satisfying   \eqref{Wdisloc1} as  configurational prestress.  To have a spherically symmetric example, we choose
\begin{equation}
  \label{tracfree3}
  \BGve_0(\Bx)=a\tns{\Hat{\Bx}},\quad\Hat{\Bx}=\frac{\Bx}{|\Bx|}.
  \end{equation}    
  It is then straightforward to check that 
 \begin{equation}
  \label{tracfree31}  
  {\rm Ink}[\BGve_0]_{ijkl}(\Bx)=\frac{2a}{r^{2}}(\hat{x}^{l}\hat{x}^{k}\Gd_{ij}-
  \hat{x}^{i}\hat{x}^{k}\Gd_{lj}-\hat{x}^{l}\hat{x}^{j}\Gd_{ik}+\hat{x}^{i}\hat{x}^{j}\Gd_{lk}
  +\Gd_{ik}\Gd_{jl}-\Gd_{ij}\Gd_{kl})\neq 0,
\end{equation}
except when $n=2$, which case ${\rm Ink}[\BGve_0]_{1122}(\Bx)=2\pi\Gd(\Bx)$. This means that there is a local displacement field in $\bb{R}^{2}\setminus\{0\}$, corresponding to the strain field \eqref{tracfree3}, but no global one. This is illustrated in Fig.~\ref{fig:discl} by picturing the images of the Cartesian grid lines under the corresponding discontinous deformation $\By(\Bx)=\Bx+\Ge\Bu(\Bx)$, where
\[
  u_{1}(x_{1},x_{2})=x_{1}+x_{2}\arctan\left(\frac{x_{2}}{x_{1}}\right),\quad
  u_{2}(x_{1},x_{2})=x_{2}-x_{1}\arctan\left(\frac{x_{2}}{x_{1}}\right).
\]

As the   trivial candidate $e(\Bu)=\BGve_{0}$  is not admissible,   the minimizer of the energy  \eqref{Erad} with free \bc s  \eqref{Erad1}  will be a stationary extremal possessing positive energy as it will be unstressed ``physically'' but still prestressed ``configurationally'' in the sense that it will carry nonzero residual stresses.

 The minimizer solves the following traction \bvp:
  \[
    \begin{cases}
    \Grad(\Div\Bu)+\GD\Bu=2\Div\BGve_{0}(\Bx),&\Bx\in\GO,\\
    \BGs\Bn=(e(\Bu)-\BGve_{0})\Bn=0,&\Bx\in\dOm.
    \end{cases}
 \]
To obtain such a minimizer explicitly,  when $\GO=B(0,R)$---the ball in $\bb{R}^{n}$, centered at 0 with radius R, and $\BGve_{0}(\Bx)$ given by (\ref{tracfree3}) we first  rescale both $\Bx$ and $\Bu$ variables and set, \WLOG, $a=1$, $R=1$.
We will look for a radial solution
\begin{equation}
  \label{urad}
  \Bu(\Bx)=\eta(r)\Hat{\Bx}. 
\end{equation}
Since $\Grad\Bu$ is symmetric, the 
 function $\eta(r)$  should  be chosen such that
\begin{equation} \label{r01}
\GD\Bu=\Div(\tns{\Bx}),
\end{equation}
or in terms of $\eta(r)$, such that
\begin{equation} \label{r0}
\eta''(r)+(n-1)\left(\frac{\eta(r)}{r}\right)'=\frac{(n-1)}{r}.
\end{equation}
A solution of \eqref{r0} that does not blow up at $r=0$ must have the form   
\begin{equation}
  \label{eta}
    \eta(r)=
  br+\frac{n-1}{n}r\ln r.
\end{equation}
Finally, the zero traction boundary condition  \eqref{Erad1} gives
$b=1/n$, giving
\begin{equation}
  \label{P}
  \BP=\left(1+\frac{n-1}{n}\ln r\right)\BI_{n}-\tns{\Hat{\Bx}}, 
\end{equation}
\begin{equation}
  \label{P*}
    \BP^{*}=\frac{(n+2)(n-1)^{2}\ln r -1}{n^{2}}\tns{\Hat{\Bx}}+
  \frac{(n+2)(n-1)^{2}(\ln r)^{2}+2(n^{2}-1)\ln r +n+1)}{2n^{2}}\BI_{n}. 
\end{equation}

\begin{figure}[t]
  \centering
  \includegraphics[scale=0.25]{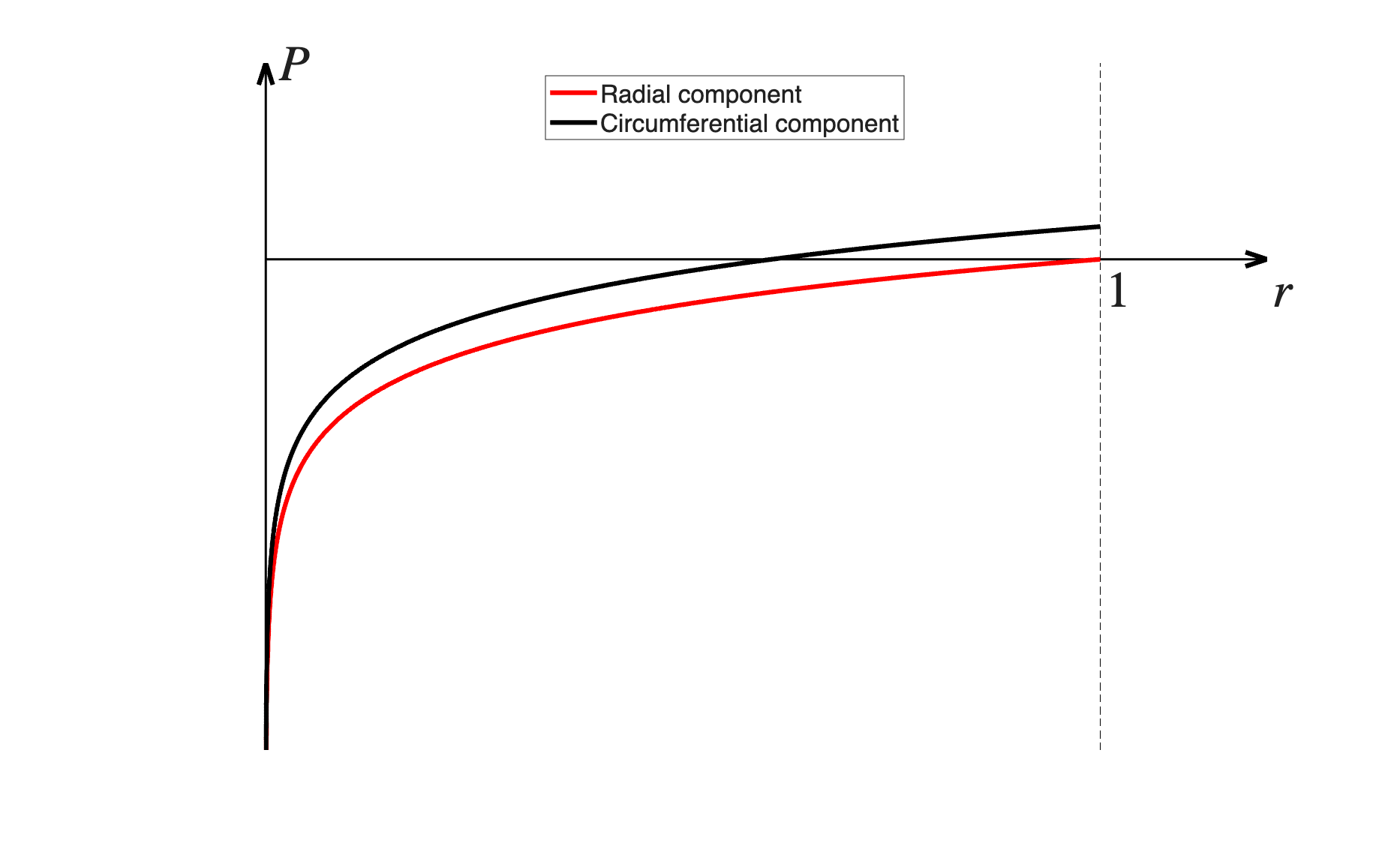}
  \includegraphics[scale=0.25]{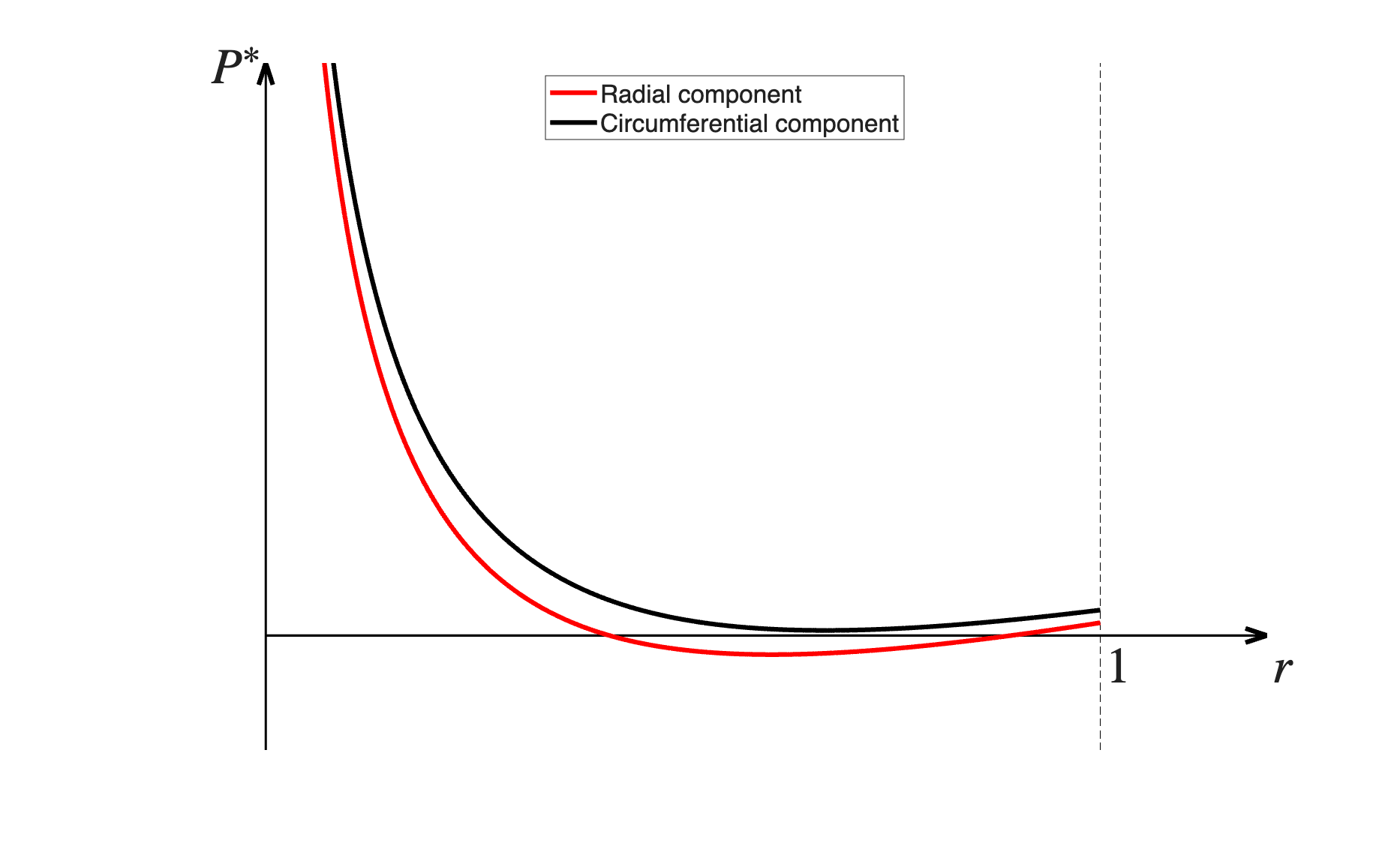}
  \caption{Radial and circumferential components of $\BP$ (left panel)
    and $\BP^{*}$ (right panel).}
  \label{fig:radial}
\end{figure}

We can now compute the energy of our configurationally prestressed body by direct substitution, obtaining
\begin{equation}
  \label{Epos1}
E[\Bu]= \frac{|B(0,1)|}{2}\frac{n-1}{n^{2}}>0,
\end{equation}
where $|B(0,1)|$ denotes the $n$-dimensional volume of the unit ball in $\bb{R}^{n}$.
As we have already mentioned, the positive value of the energy should   be attributed to the
nonzero external configurational loading
\[
\BP^{*}(\Bx)\Bn(\Bx) =\frac{n-1}{2n^{2}}\Hat{\Bx},\quad|\Bx|=R=1.
\]
When $n=3$, we obtain that Ink$[\BGve_{0}]$ has only the radial component equal to
$-2/r^{2}$ . The plots of the radial and circumferential components of
$\BP$ and $\BP^{*}$, given by (\ref{P}) and (\ref{P*}), respectively, are shown in Fig.~\ref{fig:radial} 
The nonzero radial component of $\BP^{*}$ on the boundary can be thought as a factor ensuring the   embedding of the incompatibility Ink$[\BGve_{0}]$ into the body during the process of its imaginary growth through surface deposition.

Note that  the configurational prestress can be  attributed to the
incompatibility of the embdedded inhomogeneity of 
$\BGve_{0}(\Bx)$ in general. Indeed, if we recall the  orhogonal decompostion of the
Hilbert space $\CH=L^{2}(\GO;\Sym(\bb{R}^{n}))$
\begin{equation}
  \label{decom}
  \CH=\CE\oplus\CJ,\quad \CE=\{e(\Bu): \Bu\in
H^{1}(\GO;\bb{R}^{n})\},\ \CJ=\{\BGs\in\CH:\Div\BGs=0,\ \BGs|_{\dOm}\Bn=0\}.
\end{equation}
Therefore we can always write any field in $\CH$ as a sum of a field in $\CE$ and a field in $\CJ$. Specifically, for $\BGve_{0}\in\CH$ we can find $\Bu_{0}\in H^{1}(\GO;\bb{R}^{n})$ and $\BGs_{0}\in\CJ$, such that
\[
\BGve_{0}(\Bx)=e(\Bu_{0})-\BGs_{0}.
\]
Then the energy minimizer will be $\Bu=\Bu_{0}$, the stress will be $\BGs_{0}=e(\Bu_{0})-\BGve_{0}(\Bx)$ and the minimum value of the energy will be
\[
E[\Bu_{0}]=\hf\|\BGs_{0}\|^{2}_{L^{2}(\GO)}.
\]
The incompatible component $\BGs_{0}$ can be found by solving the system
\begin{equation}
  \label{Inkeq}
  \begin{cases}
    {\rm Ink}[\BGs_{0}]=-{\rm Ink}[\BGve_{0}],&\Bx\in\GO\\
    \Div\BGs_{0}=0,&\Bx\in\GO,\\
    \BGs_{0}\Bn=0,&\Bx\in\dOm.
  \end{cases}
\end{equation}
The decomposition (\ref{decom}) implies that the \bvp\ (\ref{Inkeq}) has a unique solution on simply connected
domains, since Ink$[\BGve]=0$ implies $\BGve=e(\Bv)$ for some $\Bv\in
H^{1}$. Thus, $\BGs_{0}$, and therefore the energy $E[\Bu_{0}]$ depends only on ${\rm Ink}[\BGve_{0}]$. In fact,
\[
\BP^{*}\Bn=W\Bn=\hf|\BGs_{0}|^{2}\Bn,\text{ on }\dOm,
\]
which allows us to   conclude that our ``physically unloaded'' body  was ultimately ``loaded'' by configurational forces.

\subsection{Generalized Noether formula}
We  now summarize our general  discussion   by formulating the theorem: 
\begin{theorem}[Generalized Noether formula]~
  \label{th:genoether}
    If $\By$ is a piecewise $C^{2}$ extremal in the sense of Definition~\ref{def:equil}, whose only singularity is a smooth surface $\GS$ of jump discontinuity\footnote{The surface $\GS$ does not need to be connected, and hence the case of several surfaces of jump discontinuity is included.} of $\Grad\By(\Bx)$, then the first variation $\Gd E$ of (\ref{non-param}) under the action of a $C^{1}$ transformation (\ref{trgrp}) is given by
   \begin{equation}
  \label{incremS}
  \Gd E=\int_{\dOm}\{\BP\Bn\cdot\Gd\By+\BP^{*}\Bn\cdot\Gd\Bx\}dS
-\int_{\GS}p^{*}_{\GS}\Bn\cdot\Gd\Bx\,dS,
\end{equation}
where $p^{*}_{\GS}$ is given by (\ref{pstar}). If the extremal $\By$ is stationary, then $p^{*}_{\GS}=0$ in (\ref{incremS}).
\end{theorem}
\begin{remark}
  \label{rem:reg}
  One can also say that formula
  \begin{equation}
  \label{increm}
  \Gd E=\int_{\dOm}\{\BP\Bn\cdot\Gd\By+\BP^{*}\Bn\cdot\Gd\Bx\}dS(\Bx),
\end{equation}
holds for any Lipschitz stationary configuration. However, the \rhs\ in (\ref{increm}) can no longer be understood as an integral of a functions over a surface, since $\BP$ and $\BP^{*}$ are not defined on $\dOm$, when $\By$ is merely Lipschitz. However, if we additionally require that $\By$ be of class $C^{2}$ on a \nbh\ of $\dOm$, then the \rhs\ in formula (\ref{increm}) is understood in the classical pointwise sense. We have used this remark in \cite{grtrhard}.
\end{remark}

Theorem~\ref{th:genoether} reminds us that the classical condition of vanishing of the first variation
is not satisfied  in general for stationary configurations, but only for the variations that vanish on the boundary. It also points to the fact that in the case of extremals that are not stationary, as for instance occurs in dynamics, the vanishing of the first variation requires additionally that the defects (the discontinuities of $\Grad\By$) are not perturbed ($\Gd\Bx=0$ on $\GS$). 

Finally, we mention that one can  rewrite (\ref{increm}) in a form that makes it clear that only the part of the virtual graph-displacement $(\Gd\Bx,\Gd\By)$ that moves the graph of $\By(\Bx)$ as a geometric object contributes to the integral on the \rhs\ of (\ref{increm}). Indeed,
we observe that
\[
T_{(\Bx,\By(\Bx))}\Md\GG=\{(\BGt,\Grad\By(\Bx)\BGt):\BGt\in T_{\Bx}\dOm\}
\]
is the tangent space to the boundary of the graph of $\By(\Bx)$ at the point $(\Bx,\By(\Bx))$, $\Bx\in\dOm$. We then compute, using formula (\ref{Esh-tensor}),
\[
(\BP^{*}\Bn,\BP\Bn)\cdot(\BGt,\BF\BGt)=\BP^{*}\Bn\cdot\BGt+\BP\Bn\cdot\BF\BGt=W\Bn\cdot\BGt-\BF^{T}\BP\Bn\cdot\BGt+\BP\Bn\cdot \BF\BGt=0,
\]
where we have used the shorthand $\BF$ to denote $\Grad\By(\Bx)$, and where $\BP$ and $\BP^{*}$ are evaluated at $(\Bx,\By(\Bx),\Grad\By(\Bx))$.
This shows that the vector $(\BP^{*}\Bn,\BP\Bn)\in\bb{R}^{m+n}$ is orthogonal to $\Md\GG$ at every point of $\Md\GG$. Therefore, we can write
\begin{equation}
  \label{grincrgeom}
  \Gd E=\int_{\dOm}\{(\BP^{*}\Bn,\BP\Bn)\cdot\CP_{(\Md\GG)^{\perp}}(\Bx)(\Gd\Bx,\Gd\By)\}dS,
\end{equation}
where $\CP_{(\Md\GG)^{\perp}}(\Bx)$ denotes the orthogonal projection from $\bb{R}^{m+n}$ onto $(T_{(\Bx,\By(\Bx))}\Md\GG)^{\perp}$.

\subsection{Case study: dynamics}
 In elastodynamics the deformation is time dependent $\By=\By(t,\Bx)$ and the functional $E[\By]$ is Lagrange's action integral
\begin{equation}
  \label{dyn}  A[\By]=\int_{0}^{T}\left[\int_{\GO(t)}\left\{\frac{|\dot{\By}|^{2}}{2}-U(\Grad\By)\right\}d\Bx\right]dt,
\end{equation}
where $\GO(t)$---an arbitrary moving volume.  In this section we assume that the deformation $\By$ satisfies all assumptions of Theorem~\ref{th:genoether}.

Using our general framework we can   rewrite (\ref{dyn}) as a space-time integral
\begin{equation}
  \label{dyn4}  A[\By]=\int_{\GO^{n+1}}\left\{\frac{|\dot{\By}|^{2}}{2}-U(\Grad\By)\right\}d\Bq,
\end{equation}
where $q^{0}=t$, $q^{\Ga}=x^{\Ga}$, $\Ga=1,\ldots,n$ and
\[
\GO^{n+1}=\{\Bq=(t,\Bx):t\in[0,T],\ \Bx\in\GO(t)\}.
\]
The dynamic extremals are solutions of 
\begin{equation}
  \label{Newtlaw}
\mathfrak{E}_{L}(\Bq)=-  ^{n+1}\Div\BCP=0, 
\end{equation}
\cite{trusk87}, where the dynamic Piola stress   $\BCP\in\bb{R}^{m\times (n+1)}$ is given by 
\begin{equation}
  \label{DynP}
  \BCP=[\Bv,-\BP], 
\end{equation}
where 
 $\Bv=\dot{\By},$ 
and, as before,
 $\BP=U_{\BF}(\Grad\By).$ 
In  \eqref{Newtlaw} we also introduced the notation
\[
^{n+1}\Grad=\left(\dif{}{t},\Grad \right)=\left(\dif{}{t},\dif{}{x^{1}},\ldots, \dif{}{x^{n}}\right).
\]
Equation (\ref{Newtlaw}) can be of course rewritten as the classical Newton's equations of motion (balance of linear momentum)
\begin{equation}
  \label{Newton}
    \dif{\Bv}{t}=\Div \BP. 
\end{equation}
At those points $\Bq\in\bb{R}^{n+1}$, at which $\By$ is of class $C^{2}$, identity (\ref{Noether0}) applies and we also have
\begin{equation}
  \label{DynEsh}
  \mathfrak{E}^{*}_{L}(\Bq)=-^{n+1}\Div\BCP^{*}=0, 
\end{equation}
where the dynamic Eshelby tensor is 
\begin{equation}
  \label{DynEshTens}
  \BCP^{*}=\mat{-e}{\BP^{T}\Bv}{-\BF^{T}\Bv}{\hf|\Bv|^{2}\BI_{n}-\BP^{*}},
\end{equation}
where
\[e=\hf|\dot{\By}|^{2}+U(\Grad\By)\]
is the total energy density and, again, as before, 
  $\BP^{*} =U(\Grad\By)\BI_{n}-\BF^{T}\BP.$

Finally, observe  that in (\ref{DynEsh}) we refer to the space-time Lagrangian $$L(\BCF)=|\Bv|^{2}/2-U(\BF),$$ where the analog of deformation gradient is $$\BCF=[\Bv,\BF].$$

The system of $n+1$ equation (\ref{DynEsh}) can be represented  by a sub system of $n$ equations
\begin{equation}
  \label{config}
  \dif{(\BF^{T}\Bv)}{t}-\hf\Grad(|\Bv|^{2})+\Div\BP^{*}=0,
\end{equation}
which represent balance of linear momentum and can also be  written as
\begin{equation}
  \label{config0}
\BF^{T}\dif{\Bv}{t}+\Div\BP^{*}=0,
\end{equation}
plus an additional equation
\begin{equation}
  \label{energy}
  -\dif{e}{t}+\Div(\BP^{T}\Bv)=0, 
\end{equation}
representing  balance of energy. Note that (\ref{config}), as well as    its analog   \eqref{config0}, can be viewed as a configurational version of balance of linear momentum (\ref{Newton}).

However, as is well-known, in dynamics the smoothness of $\By(\Bq)$ can be violated even for smooth initial data due, for instance, to the formation of shock waves \cite{dafHCL}.
  In this case equation (\ref{DynEsh}) is no longer valid, and the function $\mathfrak{E}^{*}_{L}$ becomes  a distribution   supported on a singular set of $^{n+1}\Grad\By(\Bq)$. When the singular set is a jump discontinuity surface $\GS$, we obtain \cite{trusk87}, according to (\ref{EWstar}) that 
\begin{equation}
  \label{dynpst}
  \mathfrak{E}^{*}_{L}(\Bq)=-{\mathcal P}^{*}_{\GS}\BN^{\rm sh}\Gd_{\GS}(\Bq), 
\end{equation}
where
\begin{equation}
  \label{dynpstar}
{\mathcal P}^{*}_{\GS}=\jump{L}-\av{\BCP_{\pm},\jump{\BCF}}
\end{equation}
and $\BN^{\rm sh}$ denotes unit normal to an $n$-dimensional   surface $\GS\subset\bb{R}^{n+1}$. The choice of sign in (\ref{dynpstar}) is immaterial in view of (\ref{jumpP}) that in our dynamics context takes  the form
\begin{equation}
  \label{dynjumpP}
  \jump{\BCP}\BN^{\rm sh}=0.
\end{equation}
Equations (\ref{dynjumpP}) are known as the Rankine-Hugoniot conditions.

It is instructive to rewrite  (\ref{dynpstar})  in terms of the potential $U$ and block-components $\Bv$ and $\BF$ of $\BCF$. Suppose that the surface $\GS\subset\bb{R}^{n+1}$ represents a moving discontinuity surface  $S(t)\subset\bb{R}^{n}$ so that
\[
\GS=\{(t,S(t))\subset\bb{R}^{n+1}:t\in[0,T]\}.
\]
Let $V^{\rm sh}(t,\Bx)$ denote the normal velocity of $S(t)$ and $\Bn(\Bx,t)$ the unit normal to $S(t)$. Then
the vector
\[
\Hat{\BN}^{\rm sh}(\Bq)=[-V^{\rm sh}(t,\Bx),\Bn^{\rm sh}(t,\Bx)],
\]
is normal to $\GS$. Therefore, the Rankine-Hugoniot conditions (\ref{dynjumpP}) can be written as
\begin{equation}
  \label{RH}
  \jump{\Bv}V^{\rm sh}+\jump{\BP}\Bn^{\rm sh}=0. 
\end{equation}
Observe that the Hadamard relations (\ref{Hadamard}) in dynamical case take the form
\begin{equation}
  \label{Had}
  \jump{\Bv}=-V^{\rm sh}\Ba,\quad\jump{\BF}=\Ba\otimes\Bn^{\rm sh}, 
\end{equation}
where $\Ba(t,\Bx): S(t)\to\bb{R}^{m}$ is a smooth vector field. Eliminating $\Ba$ from (\ref{Had})
we obtain
\begin{equation}
  \label{dynHad}
  V^{\rm sh}\jump{\BF}=-\jump{\Bv}\otimes\Bn^{\rm sh}.
\end{equation}
Using (\ref{RH}) and (\ref{dynHad}) we can rewrite (\ref{dynpstar}) as
\begin{equation}
  \label{dynpstfin}
  {\mathcal P}^{*}_{\GS}=\av{\lump{\BP},\jump{\BF}}-\jump{U},
\end{equation}
where
\[
\lump{\BP}=\hf(\BP_{+}+\BP_{-}).
\]
Indeed, expanding the \rhs\ of (\ref{dynpstar}) we have
\[
{\mathcal P}^{*}_{\GS}=\lump{\Bv}\cdot\jump{\Bv}-\jump{U}-\Bv_{\pm}\cdot\jump{\Bv}+\av{\BP_{\pm},\jump{\BF}}=\av{\BP_{\pm},\jump{\BF}}-\jump{U}\mp\hf|\jump{\Bv}|^{2}.
\]
Taking a dot product of (\ref{Had})$_{1}$ with $\jump{\Bv}$ we obtain
\[
|\jump{\Bv}|^{2}=-V^{\rm sh}\jump{\Bv}\cdot\Ba.
\]
Now using (\ref{RH}) we obtain
\[
|\jump{\Bv}|^{2}=\jump{\BP}\Bn^{\rm sh}\cdot\Ba=\av{\jump{\BP},\jump{\BF}},
\]
where we have used (\ref{Had})$_{2}$ in the last equality above. Now,
observing that
\[
\BP_{\pm}\mp\hf\jump{\BP}=\lump{\BP},
\]
we obtain (\ref{dynpstfin}).
 Comparing (\ref{dynpstfin}) and (\ref{pstar}) we introduce a ``spatial'' version $p^{*}_{S(t)}$ of the space-time quantity ${\mathcal P}^{*}_{\GS}$:
  \begin{equation}
    \label{spatialp*}
    p^{*}_{S(t)}=-{\mathcal P}^{*}_{\GS}=\jump{U}-\av{\lump{\BP},\jump{\BF}}.
  \end{equation}

Elastodynamics provides   an interesting  case of application of  Theorem \ref{th:genoether}  because  one can take advantage of the specific quadratic dependence of the Lagrangian on
$\Bv=\dot{\By}$. Therefore,   it is natural to consider the
family of deformations  
\begin{equation}
  \label{dyntdef}
  \Tld{\Bx}=\Bx,\quad\Tld{\By}=\By,\quad\Tld{t}=e^{\Ge}t. 
\end{equation}
So we have   $\Gd\By=0$, $\Gd\Bq=t\Be_{0}$ and 
the \rhs\ of   (\ref{incremS}) becomes
\[
\int_{\dOm^{n+1}}t[-e,\BP^{T}\Bv]\cdot\BN
d\CS-\int_{\GS}t{\mathcal P}^{*}_{\GS}\BN^{\rm sh}\cdot\Be_{0}d\CS
\]
where we used the fact that  $(\BCP^{*})^{T}\Be_{0}=[-e,\BP^{T}\Bv]$. We also observe that on
$\dOm^{n+1}$ we have
\begin{equation}
  \label{NdS}
  \BN d\CS=[-V_{n},\Bn]dSdt, 
\end{equation}
where $\BN(\Bq)$ is the outward unit normal to $\dOm^{n+1}$, $\Bn(t,\Bx)$ is the
outward unit normal to $\dOm(t)$, and $V_{n}(t,\Bx)$ is the normal
velocity of $\dOm(t)$. Similarly,
\[
\BN^{\rm sh} d\CS=[-V^{\rm sh},\Bn^{\rm sh}]dS dt,
\]
Hence, the \rhs\ of (\ref{incremS}) can be written as
\[
\int_{0}^{T}\int_{\dOm(t)}t(V_{n}e+\BP^{T}\Bv\cdot\Bn)dSdt-\left.t\int_{\GO(t)}ed\Bx\right|_{0}^{T}
+\int_{0}^{T}\int_{S(t)}t V^{\rm sh}{\mathcal P}^{*}_{\GS}dS(t)dt.
\]
Let us now compute the \lhs\ in (\ref{incremS}). Given that 
\[
  \Tld{\By}_{\Ge}(\Tld{t},\Bx)=\By(e^{-\Ge}\Tld{t},\Bx),
\]
we can write 
\[
A[\Tld{\By}_{\Ge}]=\int_{0}^{e^{\Ge}T}\int_{\GO(t)}\left\{\frac{e^{-2\Ge}}{2}
|\dot{\By}(e^{-\Ge}\Tld{t},\Bx)|^{2}-U(\Grad\By)\right\}d\Bx d\Tld{t}.
\]
Changing variables $\Tld{t}=e^{\Ge}t$ we obtain
\[
A[\Tld{\By}_{\Ge}]=\int_{0}^{T}\int_{\GO(t)}\left\{\frac{e^{-\Ge}}{2}
|\dot{\By}(t,\Bx)|^{2}-e^{\Ge}U(\Grad\By)\right\}d\Bx dt.
\]
Hence, the \lhs\ in (\ref{incremS}) is
\[
\Gd A=\left.\frac{dA[\Tld{\By}_{\Ge}]}{d\Ge}\right|_{\Ge=0}=-\int_{0}^{T}\int_{\GO(t)}ed\Bx dt.
\]
We conclude that for the family of deformations (\ref{dyntdef}) the equation 
 (\ref{incremS}) can be rewritten in  the form
\[
  \int_{0}^{T}\int_{\GO(t)}ed\Bx dt=\left.t\int_{\GO(t)}ed\Bx\right|_{0}^{T}-
  \int_{0}^{T}\int_{\dOm(t)}t\{eV_{n}+\BP\Bn\cdot\Bv\}dSdt
  -\int_{0}^{T}\int_{S(t)}V^{\rm sh}t{\mathcal P}^{*}_{\GS}dSdt.
\]
After differentiation in $T$ and other simplifications we obtain 
\begin{equation}
  \label{kinetic}
\frac{d}{dt}\int_{\GO(t)}ed\Bx=\int_{\dOm(t)}\BP\Bn\cdot\Bv\,dS+\int_{\dOm(t)}eV_{n}dS+
\int_{S(t)}V^{\rm sh}{\mathcal P}^{*}_{\GS}dS.
\end{equation}
This is nothing else but the energy balance relation, showing that the   change
of the total energy contained in a moving volume takes place because  of the work 
performed by external surface tractions,  the energy brought   
inside the body due to addition/removal of mass through the external boundary, and finally the energy produced or absorbed  due to the mass exchange on the moving  strain discontinuity.

Note that,  due to thermodynamic reasons, on  physically admissible discontinuities in passive media the energy can only dissipate and therefore   
we must necessarily have 
\begin{equation}
  \label{dissipation1}
  V^{\rm sh}{\mathcal P}^{*}_{\GS}\le 0.
  \end{equation}
 The meaning of the inequality  \eqref{dissipation1} is that  extremals in the case of dynamics cannot be expected to be
  stationary  and the the  condition $^{n+1}\Div\BCP^{*}=0$ does not
  have to hold.  More specifically, it means  that
  $^{n+1}\Div\BCP^{*}$  can be  represented by  a distribution
  supported on a singular surface of the solution and subjected only to an
  inequality constraint (\ref{dissipation1}).  The conventional form of this constraint known from the theory of shock waves  \cite{trusk87} can be obtained if we recall that by 
convention the direction of $\Bn^{\rm sh}$ is always chosen to point
in the direction of the advance of the discontinuity, $V^{\rm sh}\ge 0$. Therefore, on
a physically admissible strain discontinuity we must always have ${\mathcal P}^{*}_{\GS}\le 0$, or
\begin{equation}
  \label{dissipation}
  p^{*}_{S(t)}=\jump{U}-\av{\lump{\BP},\jump{\BF}}\ge 0,
\end{equation}
where  $p^{*}_{S(t)}$ is defined in (\ref{spatialp*}), and the ``$\pm$'' convention in (\ref{dissipation}) is such that the
``$-$'' side is behind the shock wave and the ``$+$'' side is in front.

\section{Partial invariance} 
\setcounter{equation}{0}
\label{sec:parinv}
Throughout this section we assume that both $W$ and the extremal are of class $C^{2}$.
The main significance of the developed formalism  lies in its usefulness for converting special scaling and invariance properties of variational problems  into useful  \emph{formulas} satisfied by stationary configurations or almost regular extremals (as in Theorem  \ref{th:genoether}(b)). The idea is to use particular  deformations $\BGF_{\Ge}$ of the $(\Bx,\By)$ space under which, due to the  special   properties of the Lagrangian $W(\Bx,\By,\BF)$, the functional $E[T\GG]$ in   \eqref{firstvar2} transforms \emph{in a known way}.

In this section  we collected  several  examples illustrating  the  idea that if $E[T\GG_{\Ge}]$ can be simplified or tranformed by means of specific properties of the Lagrangian, then $\Gd E$, computed as (\ref{firstvar}) will express this ``partial symmetry'', while equations (\ref{increm}) or (\ref{incremS}) will represent the corresponding ``partial conservation laws''. In all these cases we show how    such partial   invariance  of a problem  can be used to obtain  useful  information about the solutions.

\subsection{Parametric problems}
\label{sec:param}
An important set of examples concerns  parametric variational integrals
\begin{equation}
  \label{Egraph1}
E[\GS]=\int_{D}W(\Bz(\Bt),\Grad\Bz(\Bt))d\Bt.
\end{equation}
The goal is to find an optimal  $n$-dimensional surface $\GS=\{\Bz(\Bt):\Bt\in D\}\subset\bb{R}^{m}$, where $D\subset\bb{R}^{n}$ is open, bounded and smooth.  The solution is not supposed to depend  on a particular choice of parametrization.

\begin{definition}
  \label{def:param}
  Suppose that $\Bt(\Bs)$ is any diffeomorphism between the
  closures of the open, bounded and smooth domains $\Tld{D}$ and $D$ and
  $\Tld{\Bz}(\Bs)=\Bz(\Bt(\Bs))$. The Lagrangian $W$ is called
  \emph{parametric} if
\begin{equation}
  \label{rep}
\int_{D}W(\Bz(\Bt),\Grad\Bz(\Bt))d\Bt= \int_{\Tld{D}}W(\Tld{\Bz}(\Bs),\Grad\Tld{\Bz}(\Bs))d\Bs.
\end{equation}
for any smooth function $\Bz(\Bt)$.
\end{definition}
The notion of the parametric Lagrangian can be reformulated as a
a variational symmetry.
Let $\BGth:D\to\bb{R}^{n}$, be an arbitrary smooth function, and let
$\Bs=\Bt+\Ge\BGth(\Bt)$. Then $\Tld{\Bz}(\Bs)$ is a reparametrization of $\GS$ for all sufficiently small in absolute value $\Ge$, if
\[
\Bz(\Bt)=\Tld{\Bz}(\Bs)=\Tld{\Bz}(\Bt+\Ge\BGth(\Bt)).
\]
By definition of the parametric variartional integral equation (\ref{rep}) must hold. But then, formula (\ref{firstvar1}) applied to the transformation
\[
\BGF_{\Ge}(\Bt,\Bz)=(\Bt+\Ge\BGth(\Bt),\Bz),\quad \Gd\Bt=\BGth,\quad \Gd\Bz=0,
\]
yields taking into account that $W$ is independent of $\Bt$,
\[
  \int_{D}\av{\BP^{*},\Grad\BGth}d\Bt=0.
\]
Since, both the domain $D$ and the vector field $\BGth$ were
arbitrary, we conclude that
\begin{equation}
  \label{Pst0}
  \BP^{*}(\Bz,\BF)=0 
\end{equation}
identically. Equation  \eqref{Pst0} is then a universal relation to be
satisfied for any solution of a parametric problem.

It is now easy to see what it takes for the problem to be parametric. If we   fix $\Bz\in\bb{R}^{m}$ and $\BF\in\bb{R}^{m\times n}$, and differentiate the function
\[
GL(n)\ni\BA\mapsto\Phi(\BA)=\frac{W(\Bz,\BF\BA)}{\det\BA}
\]
we obtain
\[
  \Grad\Phi(\BA)=\frac{\BF^{T}W_{\BF}(\Bz,\BF\BA)-W(\Bz,\BF\BA)\BA^{-T}}{\det\BA}=
  \frac{\BA^{-T}}{\det\BA}\BP^{*}(\Bz,\BF\BA)=0.
\]
We can then conclude that if $W(\Bz,\BF)$ is a parametric Lagrangian, then it must have the property
\begin{equation}
  \label{vechom}
  W(\Bz,\BF\BA)=W(\Bz,\BF)\det\BA
\end{equation}
for any $\BA\in GL^{+}(n)$. Conversely, if (\ref{vechom}) holds for
all $\BA\in GL^{+}(n)$, then, differentiating (\ref{vechom}) in $\BA$
at $\BA=\BI_{n}$, we obtain $\BP^{*}(\Bz,\BF)=0$. Therefore,  the
algebraic property (\ref{vechom}) is equivalent to (\ref{Pst0}) and
gives a complete characterization of parametric Lagrangians. Note that a well known  example of parametric invariance in the context of nonlinear elasticity is the   invariance with respect to ``material remodeling", see for instance \cite{epel07}.

 To show that parametric problems are  ubiquitous,  consider again the transformations $
\BGF_{\Ge}(\Bx,\By)=(\BX(\Bx,\By,\Ge)$, $\BY(\Bx,\By,\Ge)) 
$ from (\ref{trgrp})
which treat the $(\Bx,\By)$-space $\mathfrak{X}$ as a single entity. We recall that from this standpoint 
 the field $\By(\Bx)$ is determined by its graph
$\GG\subset\bb{R}^{n}\times\bb{R}^{m}$. The crucial observation is that  the \emph{non-parametric} variational integral
(\ref{non-param}) can be written as a \emph{parametric} integral by making an
arbitrary change of variables 
\begin{equation}
  \label{graphu1}
  \Bx=\BGn(\Bt),
  \end{equation}
where $\BGn:\bra{D}\to\bra{\GO}$ is a
diffeomorphism, and $D$ is an arbitrary domain, such that $\bra{D}$ is
diffeomorphic to $\bra{\GO}$.

Indeed, since all    transformations  \eqref{graphu1}  do not change the graph of $\By$, i.e.
\begin{equation}
  \label{graphu}
  \GG=\{(\Bx,\By(\Bx)):\Bx\in\GO\}=\{(\BGn(\Bt),\By(\BGn(\Bt))):\Bt\in D\},
\end{equation}
it is natural to introduce the notation $$\Bz(\Bt)=[\BGn(\Bt);\By(\BGn(\Bt))],$$
where $[a;b]$ denotes the
column-vector $\vect{a}{b}$.  We can now write 
\[
\Grad\Bz(\Bt)=[\Grad\BGn;\Grad\By\Grad\BGn]=[(\Grad\Bz)_{1};(\Grad\Bz)_{2}],
\]
and therefore,
\[
E[\By]=\int_{D}W(\Bz(\Bt),(\Grad\Bz)_{2}(\Grad\Bz)_{1}^{-1})\det((\Grad\Bz)_{1})d\Bt=:\CE[\Bz;D].
\]
Thus, the original non-parametric Lagrangian $W(\Bx,\By,\BF)$ has been replaced by the
 parametric Lagrangian operating in the extended space 
\[
\CW(\Bz,\BCF)=W(\Bz,\BF_{2}\BF_{1}^{-1})\det\BF_{1},
\]
where $\Bz\in\bb{R}^{n+m}$, $\BCF=[\BF_{1};\BF_{2}]\in\bb{R}^{(m+n)\times n}$, and
$\BF_{1}$ is the upper $n\times n$ submatrix of $\BCF$. Note that even if the
original Lagrangian was smooth on
$\bb{R}^{n}\times\bb{R}^{m}\times\bb{R}^{m\times n}$, the new Lagrangian is
smooth only on a subset $\CO$ of $\bb{R}^{n+m}\times\bb{R}^{(m+n)\times n}$,
on which $\BF_{1}$ is invertible. 

By direct differentiation we obtain
\begin{equation}
  \label{Phat12}
\BCP(\Bz,\BCF)=\CW_{\BCF}=\left[\begin{array}{c}
\BP^{*}(\Bz,\BF_{2}\BF_{1}^{-1})\cof\BF_{1}\\[1ex]
\BP(\Bz,\BF_{2}\BF_{1}^{-1})\cof\BF_{1}
\end{array} \right].
\end{equation}
Note that such  a higher dimensional  Piola stress $\BCP$ mixes the classical low dimensional Piola and Eshelby stresses. We also have, according to    \eqref{Pst0}  , that
\[
\BCP^{*}(\Bz,\BCF)=0
\]
where
$$\BCP^{*}(\Bz,\BCF)=\CW (\Bz,\BCF)\BI_{n+m}-\BCF^{T}\BCP(\Bz,\BCF). $$
We also see that in view of the parametric nature of the problem 
\begin{equation}
  \label{Qhom}
  \CW(\Bz,\BCF\BQ)=\CW(\Bz,\BCF)\det\BQ
\end{equation}
for any invertible $n\times n$ matrix $\BQ$.

Finally, we  show that
for every stationary configuration $\By(\Bx)$ of the original functional
(\ref{non-param}) and any smooth diffeomorphism $\BGn:D\to\GO$ the function
$\Bz(\Bt)=(\BGn(\Bt),\By(\BGn(\Bt)))$ is an equilibrium configuration of
$\CE[\Bz;D]$. Intuitively, the statement is apparent, since any weak outer
variation of $\Bz=(\Bx,\By)$ is a combination of an inner and outer variations
for $\By(\Bx)$ and it can also be verified by a direct calculation.
Indeed, the extended Euler-Lagrange equation
$$\Grad_{\Bt}\cdot\CW_{\BCF}(\Bz,\BCF)=\CW_{\Bz}$$ can be written in terms of
$\By(\Bx)$ as a system
\begin{equation}
  \label{ELt1}
  \begin{cases}
  \Grad_{\Bt}\cdot\left[\BP^{*}(\BGn(\Bt),\By(\BGn(\Bt)),\Grad_{\Bx}\By(\BGn(\Bt)))
\cof\Grad_{\Bt}\BGn(\Bt)\right]=W_{\Bx}\det\Grad_{\Bt}\BGn(\Bt),\\
\Grad_{\Bt}\cdot(\BP(\BGn(\Bt),\By(\BGn(\Bt)),\Grad_{\Bx}\By(\BGn(\Bt)))
\cof\Grad_{\Bt}\BGn(\Bt))=W_{\By}\det\Grad_{\Bt}\BGn(\Bt).
\end{cases}
\end{equation}
It is then clear  that a smooth change of variables $\Bx=\BGn(\Bt)$ converts  the 
 system  (\ref{ELt1}) into an equivalent  already familiar system $$\mathfrak{E}^{*}(W)=0,\,\,\ \mathfrak{E}(W)=0.$$

\subsection{Invariant integrals}

Suppose that the energy density
$W$ is homogeneous, i.e., does not depend explicitly on $\Bx$. Therefore, we apply the transformation (\ref{trgrp}),
given by
\begin{equation}
  \label{translation}
  \BGF_{\Ge}(\Bx,\By)=(\Bx+\Ge\Ba,\By).
\end{equation}
 Then  according to (\ref{Noether0}), we can write 
\[
-\Div\BP^{*}=\mathfrak{E}^{*}_{W}(\Bx)=-\Grad\By^{T}\mathfrak{E}_{W}(\Bx)=0,
\]
and  for any piecewise smooth subdomain $V\subset\GO$, we must have
\begin{equation}
  \label{Jreason}
  \int_{\Md V}\BP^{*}\Bn\,dS=0.
\end{equation}
To show  that such  statement can be highly  nontrivial, it sufficient to refer  to the J-integral in fracture mechancis which  in two space dimensions $\Bn
dl=(dx_{2},-dx_{1})$, is defined as a vector with two components
\[
J_{1}=\int_{C}\left\{Wdx_{2}-\Bt\cdot\dif{\By}{x_{1}}dl\right\},\qquad
J_{2}=-\int_{C}\left\{Wdx_{1}+\Bt\cdot\dif{\By}{x_{2}}dl\right\}.
\]
Here we use the notation $\Bt=\BP\Bn$ for traction forces
 and $C$ is any  path in $\GO$. It is then clear that the general definition of  J-integral is  
\begin{equation}
  \label{Jreason1}
\BJ=\int_{\Md U}\BP^{*}\Bn\,dS,
\end{equation}
where the subdomain $U\subset\GO_{0}$. Due to (\ref{Jreason}), the value of the integral $\BJ$ does not change if we
deform the set $U$ smoothly. If the set $U$ contains a connected component of the
singular set $S$ in its interior, the value of  $\BJ$ does not have to
vanish and it is invariant as long as  the boundary of the deformed set
does not intersects the singular set $S$. More formally, one can say that   the J-integral \eqref{Jreason1} is
either loop homotopy invariant or fixed end-point homotopy invariant in
$\GO=\GO_{0}\setminus S$. The most  important example of a singular set  $S$ of this type is provided by the  mechanical fracture or crack  set
$S\subset\GO_{0}$, usually represented by one or more smooth surfaces of displacement discontinuity  with   elastic
deformation $\By$  remaining  smooth (at least $C^{2}$) on
$\GO=\bra{\GO}_{0}\setminus S$.  Note that the homotopy
invariance of the J-integral is shown here under more general assumptions
than it is usually done in fracture mechanics  \cite{rice68,buri73,chsh77,delph82,cher89}
in that we permit the energy density to depend on both
$\By$ and $\BF$.

Another example of this type comes naturally if we replace translational invariance by rotational invariance. In   mechanics the corresponding  homotopy-invariant quantity, called the L-integral in \cite{buri73},  was introduced  in the context of geometrically linear elasticity that does not
make a distinction between a reference and a deformed configuration.  To give a variational interpretation  of the now classical   L-integral, we will have to define isotropy in a way that is not quite
natural from the point of view of finite elasticity.

\begin{definition}
  We say that the energy density $W(\Bx,\By,\BF)$ is isotropic if 
  \begin{equation}
    \label{gliniso}
    W(\BR\Bx,\BR\By,\BR\BF\BR^{T})=W(\Bx,\By,\BF),\quad\forall\BR\in SO(3).
  \end{equation}
\end{definition}
It is easy to check that the transformation group (\ref{trgrp}), given by
\[
\BGF_{\Ge}(\Bx,\By)=(e^{\Ge\star\BGo}\Bx,e^{\Ge\star\BGo}\By),
\]
is a variational symmetry. Here the $3\times 3$ antisymmetric matrix
$\star\BGo$ is a special case of the Hodge-star operator\footnote{The
  Hodge-star operator maps $p$-forms $\Go$ on an $n$-dimensional
  Euclidean vector space $V$ into twisted $n-p$ forms $\star\Go$,
  which requires a choice of an orientation on $V$. It is defined
  by to the rule $\Go[v_{1},\ldots,v_{p}]=\star\Go[v_{p+1},\ldots,v_{n}]$
  for any positively oriented orthonormal basis $[v_{1},\ldots,v_{n}]$ of $V$. Here we use
  the convention $\Go[v_{1},\ldots,v_{p}]=\Go_{i_{1}\ldots i_{p}}
  v_{1}^{i_{1}}\cdots v_{p}^{i_{p}}$, where the summation over
  repeated indices is assumed. The object $\star\Go$ is a twisted $n-p$
  form because it changes sign, when the orientation of $V$ changes,
  while ordinary forms do not, see \cite[Section~4.5]{boss98}.}
defined in our specific case of $\bb{R}^{3}$ by
$(\star\BGo)\Bv=\Bv\times\BGo$. 
It is clear that the \lhs\ in
(\ref{increm}) is zero and that $\Gd\Bx=\Bx\times\BGo$ and 
$\Gd\By=\By\times\BGo$. Taking into account that $\BGo\in\bb{R}^{3}$ is arbitrary,  we obtain 
\begin{equation}
  \label{Lreason}
  \int_{\Md V}\{\BP\Bn\times\By+\BP^{*}\Bn\times\Bx\}dS=0,
\end{equation}
 where $V\subset\GO$ is any piecewise smooth subdomain. Equation  \eqref{Lreason}  suggests that 
\[
\BL=\int_{\GS}\{\BP\Bn\times\By+\BP^{*}\Bn\times\Bx\}dS,
\]
where $\GS$ is any closed surface in $\GO$ is a homotopy-invariant integral. The nonzero values of $\BL$ come
from noncontractible closed surfaces that are usually the boundaries of smooth
domains in $\GO_{0}$ containg one or more components of the crack set $S$.
We point out again that the homotopy
invariance of the L-integral is proved here under more general assumptions
than in \cite{buri73}, in that we permit the energy density to depend (smoothly) on $\Bx$,
$\By$, and $\BF$, as long as the isotropy property (\ref{gliniso}) holds.

Yet another example of an invariant integral in linear elasticity emerges if we note that the corresponding energy density is a quadratic function of
the displacement gradient. The corresponding scaling invariance  property can be also formulated as a
classical variational symmetry. Suppose for simplicity that  $W=W(\BF)$ is
independent of either  $\Bx$  or  $\By$. We will also assume, more generally, that
\begin{equation}
  \label{scaling}
  W(\Gl\BF)=\Gl^{p}W(\BF),\quad\forall\Gl>0,\ \forall\BF\in\bb{R}^{m\times n},
\end{equation}
for some $p\not=0$. It is easy to find a scaling law that
$W(\Bx,\By,\BF)$ would have to satisfy to possess a variational
symmetry:
\[W(\Gl^{-\frac{n}{p}}\Bx,\Gl^{1-\frac{n}{p}}\By,\Gl\BF)=\Gl^{p}W(\Bx,\By,\BF).\]
It is straightforward  to verify that the corresponding transformation group
\[
\BGF_{\Ge}(\Bx,\By)=(e^{\Ge}\Bx,e^{\frac{p-n}{p}\Ge}\By)
\]
is a variational symmetry. Then applying (\ref{increm}) with $\Gd\Bx=\Bx$,
$\Gd\By=\frac{p-n}{p}\By$ we obtain that 
\begin{equation}
  \label{Mreason}
  \int_{\Md V}\left\{\BP^{*}\Bn\cdot\Bx+\frac{p-n}{p}\BP\Bn\cdot\By\right\}dS=0.
\end{equation}
In the special case $p=2$, $n=2$, \eqref{Mreason} reduces to the  classical  M-integral which we write in the form
\[
M_{2D}=\int_{\GS}\BP^{*}\Bn\cdot\Bx\,dS,
\]
If  $p=2$ and $n=3$, we obtain a more complex expression 
\[
M_{3D}=\int_{\GS}\left\{\BP^{*}\Bn\cdot\Bx-\hf\BP\Bn\cdot\By\right\}dS,
\]
which can be viewed as a generalization of the relation   first obtained  in \cite{buri73}. For more information about path independent integrals in elastostatics and elastodynamics  and their relation to Noether theorems, see \cite{flet76,herr82,delph82,haol98,mark06,luma07,ngkp24}.

\subsection{Case study: Pohozaev-type identities}

In this section we show how to use Noether's formula (\ref{increm}) to 
generalize the so called Pohozaev's identity which is often used in the proof of various uniqueness results for nonlinear PDEs \cite{poho65,wagn02}.
Consider a \bvp\ for a nonlinear PDE of the form
\begin{equation}
  \label{nlpde}
  \begin{cases}
    \Div W_{\BF}(\Grad\By(\Bx))+\Phi_{\By}(\By(\Bx))=0,&\Bx\in\GO,\\
    \By(\Bx)=0,&\Bx\in\dOm,
  \end{cases}
\end{equation}
where the potentials $\Phi$ and $W$ are of class $C^{1}$, $\Phi(0)=0$,
$W(\BF)$ is positively homogeneous of degree $p\ge1$, and
$W(\Bu\otimes\Bv)\ge 0$ for all $\Bu\in\bb{R}^{n}$ and $\Bv\in\bb{R}^{m}$.
We also assume that $\GO\subset\bb{R}^{n}$ is star-shaped with
respect to the origin\footnote{We are free to choose the origin arbitrarily because
  the PDE in (\ref{nlpde}) does not depend on $\Bx$ explicitly.}. The
key observation is that any solution $\By\in
C^{2}(\bra{\GO};\bb{R}^{m})$ must be a stationary extremal of the functional
\begin{equation}
  \label{PDEvar}
  E[\By]=\int_{\GO}\{W(\Grad\By)-\Phi(\By)\}d\Bx.
\end{equation}
To obtain the generalized Pohozaev's identities we apply
Noether's formula (\ref{increm}) to the functional (\ref{PDEvar})
for the family of transformations
\begin{equation}
  \label{Pohdef}
  \Tld{\Bx}=e^{\Ge}\Bx,\quad \Tld{\By}=e^{A\Ge}\By,
\end{equation}
where $A\in\bb{R}$ is arbitrary. In this case, $\Gd\Bx=\Bx$,
$\Gd\By=A\By$, and the Noether's formula (\ref{increm}) becomes
\begin{equation}
  \label{genPoho0}
  \Gd E=\int_{\dOm}\BP^{*}\Bn\cdot\Bx\,dS, 
\end{equation}
since $\Gd\By(\Bx)=A\By(\Bx)=0$ on $\dOm$, in view of the \bc s. Here 
\[
\BP^{*}(\Bx)=W(\Grad\By)\BI_{n}-(\Grad\By)^{T}W_{\BF}(\Grad\By),\quad\Bx\in\dOm,
\]
since $\Phi(0)=0$. We also compute
\[
\Gd E=
\int_{\GO}\{n(W(\Grad\By)-\Phi(\By))+(A-1)\av{\Grad\By,W_{\BF}(\Grad\By)}-A\Phi_{\By}(\By)\cdot\By\}d\Bx.
\]
Since, $A$ is arbitrary, Noether's formula (\ref{genPoho0}) implies
the two generalized Pohozaev's identities for $C^{2}$ solutions of the
\bvp\ (\ref{nlpde}):
\begin{equation}
  \label{genPoho1}
  \int_{\GO}\{\av{\Grad\By,W_{\BF}(\Grad\By)}-\Phi_{\By}(\By)\cdot\By\}d\Bx=0
\end{equation}
and
\begin{multline}
  \label{genPoho2}
  \int_{\GO}\{n(W(\Grad\By)-\Phi(\By))-\av{W_{\BF}(\Grad\By),\Grad\By}\}d\Bx=\\
  \int_{\dOm}\{(W(\Ba\otimes\Bn)-\av{W_{\BF}(\Ba\otimes\Bn),\Ba\otimes\Bn})(\Bn\cdot\Bx)\}dS,
\end{multline}
where we have used the \bc s in (\ref{nlpde}), and denoted
$\Ba(\Bx)=\Md\By/\Md\Bn\in C^{1}(\dOm;\bb{R}^{m})$. The positive
$p$-homogeneity of $W(\BF)$ allows us to simplify Pohozaev's identities:
\begin{equation}
  \label{Pohoid1}
  \int_{\GO}\{pW(\Grad\By)-\Phi_{\By}(\By)\cdot\By\}d\Bx=0,
\end{equation}
\begin{equation}
  \label{Pohoid2}
  \int_{\GO}\{(n-p)W(\Grad\By)-n\Phi(\By)\}d\Bx=(1-p)\int_{\dOm}W(\Ba\otimes\Bn)(\Bn\cdot\Bx)dS.
\end{equation}
We can use (\ref{Pohoid1}) to eliminate $\int_{\GO}W(\Grad\By)$ from
(\ref{Pohoid2}):
\[
\int_{\GO}\{(n-p) \Phi_{\By}(\By)\cdot\By-np\Phi(\By)\}d\Bx=p(1-p)\int_{\dOm}W(\Ba\otimes\Bn)(\Bn\cdot\Bx)dS.
\]
By assumption, $W(\Ba(\Bx)\otimes\Bn(\Bx))\ge 0$ for all $\Bx\in\dOm$,
while on a star-shaped domain we have $\Bn\cdot\Bx\ge 0$ for all
$\Bx\in\dOm$. Therefore,
\begin{equation}
  \label{signdE}
\int_{\GO}\{np\Phi(\By(\Bx))-(n-p)\Phi_{\By}(\By(\Bx))\cdot\By(\Bx)\}d\Bx\ge 0.
\end{equation}
If the potential $\Phi(\By)$ is such that the integrand in
(\ref{signdE}) is negative, except at $\By=0$, then $\By=0$ is the
only smooth solution of the \bvp\ (\ref{nlpde}). Hence, we have proved
the following theorem.
\begin{theorem}
  \label{th:genPoho}
  Suppose that $W\in C^{1}(\bb{R}^{m\times n})$ is a positively
  $p$-homogeneous function, $p\ge 1$, that is nonnegative on the rank-one cone.
  Let $\Phi\in C^{1}(\bb{R}^{m})$ be such that $\Phi(0)=0$ and
\begin{equation}
  \label{signdef}
  np\Phi(\By)\le(n-p)\Phi_{\By}(\By)\cdot\By
\end{equation}
for all $\By\in\bb{R}^{m}$. Then $\By(\Bx)=0$ is the only
$C^{2}$ solution of the \bvp\ (\ref{nlpde}), provided $\{0\}$ is the connected component of
$
  Z=\{\By\in\bb{R}^{m}: np\Phi(\By)=(n-p)\Phi_{\By}(\By)\cdot\By\}
$
that contains $\{0\}$.
\end{theorem}
We note that if $\Phi(\By)$ is positively $k$-homogeneous, and $\Phi(\By)>0$ for all $\By\not=0$, then $\Phi(\By)$ satisfies (\ref{signdef}) whenever
\[
\nth{n}+\nth{k}\le\nth{p}.
\]
Indeed,
\[
np\Phi(\By)-(n-p)\Phi_{\By}(\By)\cdot\By=\Phi(\By)(np-(n-p)k)\le 0.
\]
Therefore, the  conclusion of Theorem~\ref{th:genPoho} applies when
\[
\nth{n}+\nth{k}<\nth{p}.
\]
Note that in the critical case of equality in (\ref{signdef}),
$Z=\bb{R}^{m}$, and we must have
\[(1-p)W\left(\dif{\By}{\Bn}\otimes\Bn\right)(\Bn\cdot\Bx)=0,\quad\Bx\in\dOm.\]
If $p>1$, the domain
is strictly star-shaped and $W$ is positive on all the rank-one matrices, then
we also have $\Md\By/\Md\Bn=0$ on $\dOm$. In may cases, as in the
original Pohozaev's PDE $$\GD u+a^{2}u^{\frac{n+2}{n-2}}=0,$$ this is
sufficient to assert that $\By=0$ is the only solution.

\section{Generalized Clapeyron theorem}
\setcounter{equation}{0}
\label{sec:clapeyron}
In the previous section we presented several  examples  of
variational problems with  partial invariance and  the   correspondence
between partial symmetries and partial conservation laws is exemplified by the classical 
Clapeyron's Theorem of linear elasticity.  Here we show that it  is but one member of a broader class of formulas expressing the stored elastic energy of the body in stable equilibrium in terms of physical and configurational tractions on its boundary. In what follows  we show how different integral relations of this kind  emerge as we vary our  assumptions regarding the partial symmetries of the Lagrangian. 
We refer to a particularly  important formula from this class  associated with scale invarince of an associated variational problem as the Generalized Clapeyron Theorem (GCT).

Assume that $\By(\Bx)$ satisfies all assumptions of Theorem~\ref{th:genoether},
and consider a special class of scaling transformations of the form 
\begin{equation}
  \label{scalingxy}
  \Tld{\Bx}=e^{\Ge}\Bx,\quad \Tld{\By}=e^{\Ge}\By. 
\end{equation}
In this case  $\Gd\Bx=\Bx$, $\Gd\By=\By$, and $\Tld{\By}(\Tld{\Bx})=e^{\Ge}\By(\Bx)=e^{\Ge}\By(e^{-\Ge}\Tld{\Bx})$. We compute
\begin{equation}
  \label{deltaE}
    \Gd  E=\left.\frac{d}{d\Ge}
    \int_{\GO}W(e^{\Ge}\Bx,e^{\Ge}\By(\Bx),\Grad\By(\Bx))e^{n\Ge}d\Bx\right|_{\Ge=0}=
  n E[\By]+\int_{\GO}\{W_{\Bx}\cdot\Bx+W_{\By}\cdot\By(\Bx)\}d\Bx. 
\end{equation}
By our assumptions,  $\mathfrak{E}_{W}=0$ and $
  \mathfrak{E}^{*}_{W}=-p^{*}_{\GS}\Bn\Gd_{\GS}, 
$. Therefore, the  generalized Noether's formula \eqref{incremS} gives the following representation of the elastic energy
\begin{equation}
  \label{genClap}
  nE[\By]=-\int_{\GO}\{W_{\Bx}\cdot\Bx+W_{\By}\cdot\By(\Bx)\}d\Bx
  +\int_{\dOm}\{\BP\Bn\cdot\By+\BP^{*}\Bn\cdot\Bx\}dS-\int_{\GS}p^{*}_{\GS}\Bn\cdot\Bx\,dS,
\end{equation}
where  $-W_{\By}$ has the meaning of externally applied bulk forces, while $-W_{\Bx}$ describes external configurational bulk forces.  Equation  \eqref{genClap} represents the most  general Clapeyron-type  relation which we use below to derive our main results.

\subsection{Scale-free theories}

Suppose further that the simultaneous change of scale $\Bx\mapsto\Gl\Bx$,
$\By\mapsto\Gl\By$, with $\Gl>0$,   does not affect
the energy density.
In other words, we assume that there  exists a partial symmetry  of the form 
\begin{equation}
  \label{scalefree}
  W(\Gl\Bx,\Gl\By,\BF)=W(\Bx,\By,\BF),\quad\forall\Gl>0.
\end{equation}
If $W$ is continuous on $\bb{R}^{n}\times\bb{R}^{m}\times\bb{R}^{m\times n}$, property (\ref{scalefree})
simply means that $W=W(\BF)$ and does not depend on $\Bx$ and $\By$
explicitly. If $W$ is not required to be continuous at $\Bx=0$ or
$\By=0$, then the class of functions satisfying (\ref{scalefree}) is
much broader. In this case assumption (\ref{scalefree}) implies that
\[
W_{\Bx}(\Bx,\By,\BF)\cdot\Bx+W_{\By}(\Bx,\By,\BF)\cdot\By=0
\]
identically. Hence, we have proved the following theorem

\begin{theorem}
  \label{th:clapeyron}
  Suppose that $\GS\subset\GO$ is a smooth surface of jump
  discontinuity of a Lipschitz extremal $\By\in
  C^{2}(\bra{\GO}\setminus\GS)$ of the functional
  (\ref{non-param}). If the energy density function satisfies
  (\ref{scalefree}), then
 \begin{equation}
  \label{clapeyron}
  nE[\By]=\int_{\dOm}\{\BP\Bn\cdot\By+\BP^{*}\Bn\cdot\Bx\}dS -\int_{\GS}p^{*}_{\GS}\Bn\cdot\Bx\,dS.
\end{equation}
If $\By$ is also stationary, i.e. $p^{*}_{\GS}=0$, then 
\begin{equation}
  \label{Clapeyron}
E[\By]=\nth{n}\int_{\dOm}\{\BP\Bn\cdot\By+\BP^{*}\Bn\cdot\Bx\}dS(\Bx).
\end{equation}
\end{theorem}

Note that (\ref{Clapeyron}) is a generalization  to not necessarily smooth extremals and arbitrary values of $m$ and $n$ of the  known results in elasticity theory ($m=n=3$) first obtained in \cite[Eq.~(16)]{eshe70} and independently in \cite[formula~(2.10)]{green73} (see also \cite[Eq.~(2.5)]{knst84} and \cite[Eq.~(38)]{hill86}). 
 
We now turn to the interpretation of the obtained result. Observe   first that  due to the
  presence of the configurational stress $\BP^{*}$ and the prefactor $1/n$,  the relation (\ref{Clapeyron})
  cannot be regarded as a simple  ``energy = work'' relation, even
  though the first term seems to be very similar to the work of tractions
  $\BP\Bn$ over displacements $\Bu=\By-\Bx$ on the boundary of the
  domain, since we can rewrite (\ref{Clapeyron}) as
\[
E[\By]=\nth{n}\int_{\dOm}\{\BP\Bn\cdot\Bu+\BP_{\rm disp}^{*}\Bn\cdot\Bx\}dS(\Bx),
\]
where $$\BP_{\rm disp}^{*}=W\BI_{n}-(\Grad\Bu)^{T}\BP.$$

 Note also that (\ref{Clapeyron}) is clearly fundamentally different
 from the classical Clapeyron's theorem (\ref{Clap}), where the additive term  on the \rhs\  with the Eshelby
  (Hamilton, energy-momentum) tensor   $\BP^{*}$  is absent and even in 3D space instead of $1/n$ the multiple  in front of the surface integral  is $1/2$. As we are going to see in what follows the coefficient 1/2 emerges   when  we
  impose an  additional assumption that  the  energy is a quadratic
  function of deformation gradient. The origin of the   Eshelby stress in (\ref{Clapeyron}) is
more interesting.  It reflects a subtle  fact that an elastic body may
have a nonzero energy even if it is not loaded either internally by ``physical" 
bulk forces or externally by ``physical"  surface  tractions.  

We recall that in the previous
  section we have already discussed    a situation of this type in our Example~\ref{ex1}. Using Theorem~\ref{th:clapeyron} we can generalize it to shed more light on what the loading by configurational tractions $\BP^{*}\Bn$ is all about.
   Suppose that there exists a ``transformation strain'' $\BGve_{0}(\Hat{\Bx})$, such that the energy density $W(\Hat{\Bx},\BGve)$ satisfies, for any $\BGve\in\Sym(\bb{R}^{n})$ and any unit vector $\Hat{\Bx}$, an inequality
  \begin{equation}
    \label{Wlb}
    W(\Hat{\Bx},\BGve)\ge c|\BGve-\BGve_{0}(\Hat{\Bx})|^{p},
  \end{equation}
where $p>0$ and $c>0$. Note that the energy density (\ref{Wdisloc}) with (\ref{tracfree3}) in our Example~\ref{ex1} is of this type.

Now, if there exists a Lipschitz stationary equilibrium $\Bu_{0}$ in $\GO$ satisfying both $\BP\Bn=0$ and $\BP^{*}\Bn=0$ on $\dOm$ in the sense of traces, then, according to  (\ref{Clapeyron}) we must have
\[
n\int_{\GO}W(\Hat{\Bx},e(\Bu_{0}))d\Bx=0.
\]
But then
\[
\int_{\GO}|e(\Bu_{0})-\BGve_{0}(\Hat{\Bx})|^{p}d\Bx\le\nth{c}\int_{\GO}W(\Hat{\Bx},e(\Bu_{0}))d\Bx=0.
\]
It would then mean that $\BGve_{0}(\Hat{\Bx})=e(\Bu_{0})$ is a compatible strain. 

To make clear the connection between the incompatibility of $\BGve_{0}(\Hat{\Bx})$ and the configurational tractions on the boundary we now  rephrase the problem. Suppose that $\BGve_{0}(\Hat{\Bx})$ is an incompatible  strain field, and suppose that $\Bu_{0}$ is the minimizer in
\begin{equation}
  \label{m}
 \mathcal M=\min_{\Bu\in H^{1}(\GO;\bb{R}^{n})}\int_{\GO}W(\Hat{\Bx},e(\Bu))d\Bx.
\end{equation}
Observe that due to (\ref{Wlb}), we can regard $ \mathcal M >0$ as a measure of incompatibility of $\BGve_{0}(\Hat{\Bx})$. At the same time, the minimizer $\Bu_{0}$ in (\ref{m}) must be a stationary extremal, satisfying $\BP\Bn=0$. But then, Theorem~\ref{th:clapeyron} gives
\[
 \mathcal M=\nth{n}\int_{\dOm}\BP^{*}\Bn\cdot\Bx\,dS,
\]
which ultimately links the ``applied'' configurational tractions to the  incompatibility of the ``transformation strain'' $\BGve_{0}(\Hat{\Bx})$.

  Since in  our Example~\ref{ex1} we could  obtain the solution of the  energy minimization problem explicitly, we could also directly  
  compute  the stored elastic energy  which turned out to be
  different from zero despite the absence of ``physical'' loading. Here we independently 
 compute the same energy by applying  
Theorem~\ref{th:clapeyron}; it  is obviously 
  applicable because the energy denisity (\ref{Wdisloc}) with
  (\ref{tracfree3}) satisfies (\ref{scalefree}). We obtain
  \begin{equation}
  \label{Clapeyron1}
  E[\Bu]=\nth{n}\int_{\dOm}\BP^{*}\Bn\cdot\Bx\,dS=
  \nth{n}\int_{\dOm}W(\Bx,\By,\BF)\Bn\cdot\Bx\,dS.
\end{equation}
where we used the fact that in our example ``physical''  tractions were absent and therefore 
$
  \BP^{*}\Bn=W(\Bx,\By,\BF)\Bn.
 $ Furthermore, in our special case
formula (\ref{Clapeyron1}) gives
\[
  E[\Bu]=\frac{R}{n}\int_{\Md B(0,R)}W(\Bx,\Grad\By)dS=
  \frac{R|\Md B(0,R)|}{2n}\left((\eta'(R)-a)^{2}+(n-1)\frac{\eta(R)^{2}}{R^{2}}\right).
\]
Using the explicit solution
\[
\eta(r)=\frac{a}{n}\left(r+(n-1)r\ln\left(\frac{r}{R}\right)\right),
\]
we can then obtain the expression for the energy 
\[
 E[\Bu]=\frac{a^{2}(n-1)}{2n^{2}}|B(0,R)|,
\]
which expectedly  agrees with the rescaled energy (\ref{Epos1}).

We have therefore illustrated the fact  that even if a body  is  unloaded externally
(in view of the absence of  physical forces) it can be still loaded
externally  by ``applied" configurational forces. Such loading would then reveal itself 
through a residual prestress generated by the  acquired
incompatibility.
We reiterate that  such an incompatibility should be viewed as being effectively embedded into the
body during the process of its  creation, which can be in turn  viewed as an outcome of the the work  of
configurational forces   $ \BP^{*}\Bn $ applied on the boundary of the
body.  Specifically, in our  Example~\ref{ex1} a  nonzero configurational loading was represented by the configurational tractions 
\[
\BP^{*}\Bn=\frac{a^{2}(n-1)}{2n^{2}}\Bn\not=0.
\]
We reiterate  that the classical Clapeyron's Theorem does not account for such ``work of creation'' since it  cannot accommodate any transformation strain as it  breaks the homogeneity of degree 2, required for classical Clapeyron to hold. Instead, our GCT remains suitable because of the scale-invariant nature of the transformation strain in our Example~\ref{ex1}.

We finish this section with a remark that will prove  useful in what follows. 
\begin{remark}
   \label{rem:shift}
   When $W$ does not depend explicitly on $\Bx$ and $\By$, the
   stationary extremals satisfy
\[
\int_{\dOm}\BP\Bn\,dS(\Bx)=\Bzr,\qquad\int_{\dOm}\BP^{*}\Bn\,dS(\Bx)=\Bzr.
\]
These relations allow one to re-write the Clapeyron
Theorem~\ref{th:clapeyron} as
\begin{equation}
  \label{purple}
  \int_{\GO}W(\Grad\By)d\Bx=
\nth{n}\int_{\Md\GO}\{(\BP^{*}\Bn,\Bx-\Ba)+(\BP\Bn,\By(\Bx)-\Bb)\}dS(\Bx),
\end{equation}
where $\Ba\in\bb{R}^{n}$ and $\Bb\in\bb{R}^{m}$ are arbitrary constant
vectors.
\end{remark}

\subsection{Problems with constrants}
We now show that the  relation  \eqref{Clapeyron} can be easily generalized to  cover  the case of scale-free energy with scale-free pointwise constraints. A typical example is incompressible elasticity, where every deformation $\By(\Bx)$ must satisfy the incompressibility constraint
\begin{equation}
  \label{incstr}
\det\Grad\By(\Bx)=1
\end{equation}
for all $\Bx\in\GO$. Specifically, one can prove the following theorem:
\begin{theorem}
  \label{th:constr}
  Suppose that the deformation $\Grad\By(\Bx)$ satisfies $i=1,2,3 ...k$ pointwise constraints
  \begin{equation}
    \label{constr}
      C^{i}(\Bx,\By(\Bx),\Grad\By(\Bx))=0, 
  \end{equation}
 where both  functions $W(\Bx,\By,\BF)$ and $C^{i}(\Bx,\By,\BF)$ are  scale-free in the sense of (\ref{scalefree}). Then, if $\By(\Bx)$ is a $C^{1}$ energy minimizer that satisfies the constraints, the following relation holds
  \begin{equation}
  \label{constrClap}
  \int_{\GO}nW(\Bx,\By,\Grad\By)d\Bx= \int_{\dOm}\{(\BP_{W}-\GL_{i}(\Bx)\BP_{C^{i}})\Bn\cdot\By
  +(\BP^{*}_{W}-\GL_{i}(\Bx)\BP^{*}_{C^{i}})\Bn\cdot\Bx\}dS,
\end{equation}
where $\GL_{i}(\Bx)$ are the Lagrange multipliers corresponding to the constraints (\ref{constr}).
\end{theorem}
\begin{proof}
From the general theory of constrained variational problems we know that the $C^{1}$ energy minimizers 
can be obtained as stationary extremals of the augmented Lagrangian
\[
L(\Bx,\hat{\By},\hat{\BF})=W(\Bx,\By,\BF)-\BGL(\Bx)\cdot\BC(\Bx,\By,\BF),
\]
Here the functions $\BGL(\Bx)$ are Lagrange multipliers, $\hat{\By}=[\By;\BGL]\in\bb{R}^{m+k}$,  $\Hat{\BF}=[\BF;\BG]\in\bb{R}^{(m+k)\times n}$ and $\BG=\Grad\BGL$. 
Note first that the augmented Lagrangian is affine in $\BGL$ and does not depend explicitly on $\BG$. 

In the regular case, when $\By(\Bx)$ is the energy minimizer satisfying the constraints, there are smooth Lagrange multipliers $\GL_{i}(\Bx)$, such that $\hat{\By}=[\By;\BGL]$   solves the Euler-Lagrange equations $\mathfrak{E}_{L}(\Bq)=0$ for the augmented Lagrangian, and hence also the stationarity equations $\mathfrak{E}^{*}_{L}(\Bq)=0$.

Note next that by assumption both the Lagrangian and the constraints are scale-free.  Therefore we can again consider the action of the following relevant deformation
\[
\BX=e^{\Ge}\Bx,\quad\Hat{\BY}=[e^{\Ge}\By;\GL].
\]
The deformed functional
\[
E_{\Ge}[\Hat{\BY}]=\int_{e^{\Ge}\GO}L(\BX,\Hat{\BY}_{\Ge}(\BX),\Grad\Hat{\BY}_{\Ge}(\BX))d\BX
\]
becomes, after the change of  variables $\Tld{\Bx}=e^{\Ge}\Bx$
\[
E_{\Ge}=\Ge^{n\Ge}\int_{\GO}\{W(e^{\Ge}\Bx,e^{\Ge}\By(\Bx),\Grad\By(\Bx))-
  \BGL(\Bx)\cdot\BC(e^{\Ge}\Bx,e^{\Ge}\By(\Bx),\Grad\By(\Bx))\}d\Bx.
\]
and since   both $W$ and $\BC$ satisfy (\ref{scalefree}), we have
\[
E_{\Ge}=\Ge^{n\Ge}\int_{\GO}\{W(\Bx,\By(\Bx),\Grad\By(\Bx))-
  \BGL(\Bx)\cdot\BC(\Bx,\By(\Bx),\Grad\By(\Bx))\}d\Bx=\Ge^{n\Ge}\int_{\GO}Wd\Bx,
\]
where we took into account that $\By(\Bx)$ satisfies the constraints $\BC(\Bx,\By(\Bx),\Grad\By(\Bx))=0$.
Therefore,
\[
\Gd E=n\int_{\GO}W(\Bx,\By(\Bx),\Grad\By(\Bx))d\Bx,
\]
and Noether identity (\ref{increm}) becomes (\ref{constrClap}),
since
$\Gd\Bx=\Bx$, $\Gd\Hat{\By}=[\By;0]$,
\[
  \BP_{L}=[\BP_{W}-\GL_{i}(\Bx)\BP_{C^{i}};0],\quad
  \BP_{L}^{*}=\BP^{*}_{W}-\GL_{i}(\Bx)\BP^{*}_{C^{i}}.
\]
\end{proof}

In particular, when we deal with  incompressible elasticity and the equilibrium configuration $\By(\Bx)$ has to
satisfy the constraint (\ref{incstr}), the GCT reads
\begin{equation}
  \label{incgrn}
  \int_{\GO}W(\Grad\By(\Bx))d\Bx=\nth{n}\int_{\dOm}\{(\BP-p(\Bx)\cof\BF)\Bn\cdot\By(\Bx)
  +(\BP^{*}+p(\Bx)\BI_{n})\Bn\cdot\Bx\}dS(\Bx).
\end{equation}
Here, following convention,  we have denoted the Lagrange multiplier $\GL(\Bx)$ by $p(\Bx)$ to indicate its physical meaning as pressure, and took into account that
\[
\BP^{*}_{C}=(\det\BF-1)\BI_{n}-\BF^{T}\cof(\BF)=-\BI_{n}.
\]
Note that when $\By(\Bx)$ is an orientation preserving diffeomorphism, we can change variables $\By=\By(\Bx)$ in the first term on the \rhs\ of   \eqref{incgrn} and obtain
\begin{equation}
  \label{mixedgreen}
  \int_{\GO}W(\Grad\By)d\Bx=\nth{n}\int_{\dOm^{*}}(\BGs-p(\By)\BI_{n})\BN\cdot\By\,dS^{*}(\By)
  +\nth{n}\int_{\dOm}(\BP^{*}+p(\Bx)\BI_{n})\Bn\cdot\Bx\,dS(\Bx),
\end{equation}
where $\BGs(\By)$ is the Cauchy stress tensor, $\BN(\By)$ is the outward unit normal to $\dOm^{*}=y(\dOm)$, and $p(\By)=p(\Bx)$, with $\By=\By(\Bx)$.

\subsection{$p$-homogeneous energies}
\label{sub:phom}
Now assume that the Lagrangian $W$ instead of scale invariance has the following property of
$p$-homogeneity
\begin{equation}
  \label{phom}
  W(\Bx,\Gl\By,\Gl\BF)=\Gl^{p}W(\Bx,\By,\BF).
\end{equation}
In this case the deformation 
\begin{equation}
  \label{uscale}
  \Tld{\Bx}=\Bx,\qquad\Tld{\By}=e^{\Ge}\By
\end{equation}
results in $\Tld{\By}(\Tld{\Bx})=e^{\Ge}\By(\Bx)=e^{\Ge}\By(\Tld{\Bx})$, so that
\[
E[\Tld{\By}](\Ge)=\int_{\GO}W(\Tld{\Bx},e^{\Ge}\By(\Tld{\Bx}),e^{\Ge}\Grad\By(\Tld{\Bx}))d\Tld{\Bx}=
e^{p\Ge}E[\By].
\]
We have now  $\Gd\Bx=0$, $\Gd\By=\By$ which  leads   to a new Clapeyron-type relation  satisfied by
equilibrium energy 
\begin{equation}
  \label{semiC}
  \int_{\GO}W(\Bx,\By(\Bx),\Grad\By(\Bx))d\Bx=\nth{p}\int_{\dOm}\BP\Bn\cdot\By\,dS. 
\end{equation}
We stress that the relation  (\ref{semiC}) is an expression of the $p$-homogeneity of the energy
density that manifests itself through a rescaling of the displacements
$\By$  without the concommittant rescaling of the Lagrangian
  coordinates. Therefore  the validity of (\ref{semiC}) does not require that $\By$ be stationary, i.e. solving (\ref{ELT}). Also, since $\Gd\Bx=0$  the relation
(\ref{semiC}) remains  valid even when the Lagrangian $W(\Bx,\By,\BF)$ as a function of $\Bx$ is only
 measurable. We note that  while (\ref{semiC}) can be  interpreted as a direct generalization of the classical Clapeyron's theorem  because the latter is clearly a special   case of the former, we preserve the name GCT for (\ref{Clapeyron}) as it marks a more radical generalization of this classical result linking together the work of physical and configurational forces. 
 
 Thus, below we show that a whole family of  GCT type results emerge when  scale invariance coexist with $p$-homogeneity. In particular we show that  if the energy gets rescaled when we stretch  the Lagrangian
  coordinates, without the concommittant stretching of the displacements,
  \begin{equation}
    \label{otsoh}
    W(\Gl^{-1}\Bx,\By,\Gl\BF)=\Gl^{p}W(\Bx,\By,\BF),
  \end{equation}
then the energy can be expressed as a surface integral involving only the
conjugate momenta $(\BP^{*})^{T}\Bx$.

Indeed, in this case the transformation 
\begin{equation}
  \label{xscale}
  \Tld{\Bx}=e^{\Ge}\Bx,\qquad\Tld{\By}=\By
\end{equation}
results in $\Tld{\By}(\Tld{\Bx})=\By(\Bx)=\By(e^{-\Ge}\Tld{\Bx})$, so that
\begin{multline*}
  E[\Tld{\By}](\Ge)=\int_{e^{\Ge}\GO}W(\Tld{\Bx},\By(e^{-\Ge}\Tld{\Bx}),e^{-\Ge}\Grad\By(e^{-\Ge}\Tld{\Bx}))d\Tld{\Bx}=\\
e^{-p\Ge}\int_{e^{\Ge}\GO}W(e^{-\Ge}\Tld{\Bx},\By(e^{-\Ge}\Tld{\Bx}),\Grad\By(e^{-\Ge}\Tld{\Bx}))d\Tld{\Bx}=
e^{(n-p)\Ge}E[\By].
\end{multline*}
Now we have  $\Gd\Bx=\Bx$, $\Gd\By=0$, and therefore
\begin{equation}
  \label{semiL}
  E[\By]=\int_{\GO}W(\Bx,\By(\Bx),\Grad\By(\Bx))d\Bx=\nth{n-p}\int_{\dOm}\BP^{*}\Bn\cdot\Bx\,dS, 
\end{equation}
 if $n\not=p$, and
\begin{equation}
  \label{n=p}
  \int_{\dOm}\BP^{*}\Bn\cdot\Bx\,dS=0,
\end{equation}
  if $n=p$.

More generally, it is clear that in the case when $W(\Bx,\By,\BF)$ is both scale-free and $p$-homogeneous, in
the sense of (\ref{scalefree}) and (\ref{phom}), respectively, then (\ref{otsoh}) holds as a consequence:
\[
  W(\Gl^{-1}\Bx,\By,\Gl\BF)=W(\Gl^{-1}\Bx,\Gl\Gl^{-1}\By,\Gl\BF)=
  \Gl^{p}W(\Gl^{-1}\Bx,\Gl^{-1}\By,\BF)=\Gl^{p}W(\Bx,\By,\BF).
\]
Thus both formulas (\ref{semiC}) and (\ref{semiL}) (or (\ref{n=p}), if $n=p$) hold.
Combining (\ref{semiC}) and (\ref{semiL}), when $n\not=p$, we also obtain the relation between the two terms on the \rhs\ of (\ref{Clapeyron}):
\begin{equation}
  \label{PPst}
  \nth{p}\int_{\dOm}\BP\Bn\cdot\By\,dS=\nth{n-p}\int_{\dOm}\BP^{*}\Bn\cdot\Bx\,dS.
\end{equation}
Note that in view of \eqref{PPst}  equation (\ref{Clapeyron}) which does not
contain $p$ can be obtained as a linear combination of (\ref{semiC}) and \eqref{semiL}.
We remark that formula (\ref{PPst}) is valid for any subdomain of $\GO$ and can therefore be expressed as an apparently new  concervation law
\begin{equation}
  \label{pi}
 \Div\left(p(\BP^{*})^{T}\Bx-(n-p)\BP^{T}\By\right)=0,
\end{equation}
 corresponding to the  full variational symmetry unifying  scale-invariance  and $p$-homogeneity  and described by the  transformation group $(\Bx,\By)\mapsto(e^{p\Ge}\Bx,e^{(p-n)\Ge}\By)$, $\Ge\in\bb{R}$.

It is instructive to present explicit   examples  of Lagrangians that are both scale-free and $p$-homogeneous.
If $W$ depends on all of its variables, it would have to satisfy
\[
W(\Bx,\By,\BF)=\left(\frac{|\By|}{|\Bx|}\right)^{p}W\left(\hat{\Bx},\hat{\By},\frac{|\Bx|}{|\By|}\BF\right),\quad\hat{\Bx}=\frac{\Bx}{|\Bx|},\ \hat{\By}=\frac{\By}{|\By|}.
\]
If $W$ does not depend explicitly either on $\Bx$ or on $\By$, then
\[
W(\Bx,\By,\BF)=W(\Hat{\Bw},\BF),\qquad W(\Hat{\Bw},\Gl\BF)=\Gl^{p}W(\Hat{\Bw},\BF),
\]
for every $\Gl>0$, where $\Bw$ is either $\Bx$ or $\By$.

Let us now demonstrate that in some sense the $p$-homogeneous Clapeyron-type Theorem
(\ref{semiC}) implies the GCT represented by (\ref{Clapeyron}). This becomes clear if we  realize that at their source, the transformations (\ref{trgrp})
treat the $(\Bx,\By)$-space $\mathfrak{X}$ as a single entity. Therefore, to see the implied relation  we need to go back to the description of the problem in the extended space introduced in Section~\ref{sec:param}.

We recall that in this space  the original Lagrangian $W(\Bx,\By,\BF)$ is  replaced by the
\emph{graph}-Lagrangian 
\[
\CW(\Bz,\BCF)=W(\Bz,\BF_{2}\BF_{1}^{-1})\det\BF_{1},
\]
where $\Bz\in\bb{R}^{n+m}$, $\BCF=[\BF_{1};\BF_{2}]\in\bb{R}^{(m+n)\times n}$, and
$\BF_{1}$ is the upper $n\times n$ submatrix of $\BCF$.

Now, if $\CW$ does not depend explicitly on $\Bz$, i.e., $W=W(\BF)$, then
$\CW(\BCF)$ satisfies (\ref{phom}), with $p=n$. Therefore,

\[
E[\By]=\CE[\Bz,D]=\nth{n}\int_{\Md D}\BCP\Bn\cdot\Bz(\Bt)dS(\Bt).
\]
Using formula (\ref{Phat12}) for $\BCP$ we obtain
\begin{multline*}
E[\By]=\nth{n}\int_{\Md D}\{
\BP^{*}(\Grad_{\Bx}\By(\BGn(\Bt)))\cof(\Grad_{\Bt}\BGn(\Bt))\BGv(\Bt)\cdot\BGn(\Bt)+\\
\BP(\Grad_{\Bx}\By(\BGn(\Bt)))\cof(\Grad_{\Bt}\BGn(\Bt))\BGv(\Bt)\cdot\By(\BGn(\Bt))
\}dS(\Bt),
\end{multline*}
where $\BGv(\Bt)$ is the outward unit normal to $\Md D$. Changing variables in
the above surface integral, using the relation
\[
\cof(\Grad\BGn(\Bt))\BGv(\Bt)dS(\Bt)=\Bn(\Bx)dS(\Bx),
\]
we obtain the scale-invariant GCT represented by (\ref{Clapeyron}).

\subsection{Case study: shock waves}
\label{sub:dyn}
The goal of  this section is to apply Clapeyron's Theorem~\ref{th:clapeyron}
(\ref{Clapeyron}) to dynamics where shock waves (dissipative jump
discontinuities)  emerge generically, and where the
Lagrangian $L(\BCF)$ in the action functional $A[\By]$, given by
(\ref{dyn}) is scale-free because $L$ does not depend explicitly on
$\Bq=(t,\Bx)$ and $\By$. By the Clapeyron Theorem~\ref{th:clapeyron},
  \begin{equation}
  \label{shockClapeyron}
  \int_{\GO^{n+1}}L(\BCF)d\Bq=
  \nth{n+1}\int_{\dOm^{n+1}}\{\BCP\BN\cdot\By+\BCP^{*}\BN\cdot\Bq\}d\CS
  -\int_{\GS}{\mathcal P}^{*}_{\GS}(\Bq)\BN\cdot\Bq d\CS.
\end{equation}

Our goal is to rewrite the ``space-time'' relation \eqref{shockClapeyron} using    the conventional ``space+time'' decomposition. In view of 
 (\ref{DynP}), (\ref{DynEshTens}), and (\ref{NdS}) we can write 
\[
\int_{\dOm^{n+1}}\BCP\BN\cdot\By
d\CS=-\int_{0}^{T}\int_{\dOm(t)}(V_{n}\Bv\cdot\By+\BP\Bn\cdot\By)dS dt
+\left.\int_{\GO(t)}\Bv\cdot\By d\Bx\right|_{0}^{T}.
\]
\[
\int_{\dOm^{n+1}}\BCP^{*}\BN\cdot\Bq d\CS=\int_{0}^{T}\int_{\dOm(t)}\{
t(V_{n}e+\BP\Bn\cdot\Bv)+V_{n}\Bv\cdot\BF\Bx+\BQ_{L}^{*}\Bn\cdot\Bx\}dSdt
-\left.\int_{\GO(t)}\{te+\BF\Bx\cdot\Bv\}d\Bx\right|_{0}^{T}
\]
where
\[
\BQ_{L}^{*}=L\BI_{n}+\BF^{T}\BP=\hf|\Bv|^{2}\BI_{n}-\BP^{*}.
\]
We also have
\[
  \int_{\GS}{\mathcal P}^{*}_{\GS}(\Bq)\BN\cdot\Bq d\CS=
  \int_{0}^{T}\int_{S(t)}(\Bx\cdot\Bn^{\rm sh}-V^{\rm sh}t){\mathcal P}^{*}_{\GS}dSdt.
\]
Combining all these relations  we obtain
\begin{multline}
  \label{dynClap}
  (n+1)\int_{0}^{T}\int_{\GO(t)}Ld\Bx dt=
  \left.\int_{\GO(t)}\left\{\Bv\cdot(\By-\BF\Bx)-te\right\} d\Bx\right|_{0}^{T}+\\
  \int_{0}^{T}\int_{\dOm(t)}\{t(V_{n}e+\BP\Bn\cdot\Bv)
  +V_{n}\Bv\cdot(\BF\Bx-\By)
  -\BP\Bn\cdot\By+\BQ^{*}_{L}\Bn\cdot\Bx\}dSdt\\
  -\int_{0}^{T}\int_{S(t)}(\Bx\cdot\Bn^{\rm sh}-V^{\rm sh}t){\mathcal P}^{*}_{\GS}dSdt.
\end{multline}
Formula (\ref{dynClap}) is valid for any $T>0$. We can therefore differentiate it with respect to $T$:
\begin{multline}
  \label{dynClapdif}
  (n+1)\int_{\GO(t)}Ld\Bx=
  \frac{d}{dt}\int_{\GO(t)}\left\{\Bv\cdot(\By-\BF\Bx)-te \right\} d\Bx\\
  -\int_{\dOm(t)}\{\BP\Bn\cdot(\By-t\Bv)-\BQ^{*}_{L}\Bn\cdot\Bx-tV_{n}e-V_{n}\Bv\cdot(\BF\Bx-\By)+\}dS\\
  -\int_{S(t)}(\Bx\cdot\Bn^{\rm sh}-V^{\rm sh}t){\mathcal P}^{*}_{\GS}dS.
\end{multline}
We now use the previosuly obtained energy balance relation
(\ref{kinetic}) to simplify (\ref{dynClapdif})
\begin{multline}
  \label{predynclap}
  \int_{\GO(t)}\{(n+1)L+e\}d\Bx=
  \frac{d}{dt}\int_{\GO(t)}\Bv\cdot(\By-\BF\Bx)d\Bx\\
  -\int_{\dOm(t)}\{\BP\Bn\cdot\By-\BQ^{*}_{L}\Bn\cdot\Bx-V_{n}\Bv\cdot(\BF\Bx-\By)\}dS(\Bx)\\
  -\int_{S(t)}\Bx\cdot\Bn^{\rm sh}{\mathcal P}^{*}_{\GS}dS(\Bx).
\end{multline}
  
 Using the formula for the derivative of an integral over a moving
volume and taking into account that the integrand has a jump
discontinuity over a moving shock surface $S(t)$,  we compute
\begin{multline}
  \label{inert}
  \frac{d}{dt}\int_{\GO(t)}\Bv\cdot(\By-\BF\Bx)d\Bx=\int_{\dOm(t)}V_{n}\Bv\cdot(\By-\BF\Bx)dS+
\int_{\GO(t)}\left(\dif{}{t}\{\Bv\cdot(\By-\BF\Bx)\}\right)_{\rm reg}d\Bx\\
-\int_{S(t)}V^{\rm sh}\jump{\Bv\cdot(\By-\BF\Bx)}dS.
\end{multline}
We also have
\begin{multline*}
    \left(\dif{}{t}\{\Bv\cdot(\By-\BF\Bx)\}\right)_{\rm reg}=\left(\dif{\Bv}{t}\right)_{\rm reg}\cdot(\By-\BF\Bx)+\Bv\cdot(\Bv-(\Grad\Bv)_{\rm reg}\Bx)=\\
  \left(\dif{\Bv}{t}\right)_{\rm reg} \cdot(\By-\BF\Bx)+|\Bv|^{2}-\hf(\Grad|\Bv|^{2}) _{\rm reg}\cdot\Bx 
\end{multline*}
 Integrating $(\Grad|\Bv|^{2}) _{\rm reg}\cdot\Bx$ by parts, we obtain 
\begin{multline*}
  \int_{\GO(t)}\left(\dif{}{t}\{\Bv\cdot(\By-\BF\Bx)\}\right)_{\rm reg} d\Bx=
  \int_{\GO(t)}\left\{\left(\dif{\Bv}{t}\right)_{\rm reg}\cdot(\By-\BF\Bx)
    +\frac{n+2}{2}|\Bv|^{2}\right\}d\Bx\\
  -\int_{\dOm(t)}\frac{1}{2}|\Bv|^{2}\Bx\cdot\Bn\,dS
  +\hf\int_{S(t)}\jump{|\Bv|^{2}}\Bx\cdot\Bn^{\rm sh}dS.
\end{multline*}
 Substituting this into (\ref{inert}) and back into (\ref{predynclap}),
we obtain the desired dynamical version of the Clapeyron Theorem 
\begin{multline*}
   n\int_{\GO(t)}Ud\Bx=\int_{\GO(t)}\left(\dif{\Bv}{t}\right)_{\rm reg}\cdot(\BF\Bx-\By)d\Bx
  +\int_{\dOm(t)}\{\BP\Bn\cdot\By+\BP^{*}\Bn\cdot\Bx\}dS(\Bx)\\
  +\int_{S(t)}\{V^{\rm sh}\jump{\Bv\cdot(\By-\BF\Bx)}-\jump{\Bv}\cdot\lump{\Bv}(\Bx\cdot\Bn^{\rm sh}) -p^{*}_{S(t)}\Bx\cdot\Bn^{\rm sh}\}dS(\Bx),
\end{multline*}
 where we have passed from the space-time ${\mathcal P}^{*}_{\GS}$ to its spatial version via (\ref{spatialp*}).
Next, using Hadamard's kinematic compatibility relations
  (\ref{Had}) on the shock in space-time we   simplify the \rhs\
  of the relation above, by observing that
\[ V^{\rm sh}\jump{\Bv\cdot(\By-\BF\Bx)}-\jump{\Bv}\cdot\lump{\Bv}(\Bx\cdot\Bn^{\rm sh})=V^{\rm sh}\jump{\Bv}\cdot(\By-\lump{\BF}\Bx),
\]
  and obtain our preliminary version of the dynamic Clapeyron Theorem  
\begin{multline}
  \label{dynClapfin}
   n\int_{\GO(t)}Ud\Bx=\int_{\GO(t)}\left(\dif{\Bv}{t}\right)_{\rm reg}\cdot(\BF\Bx-\By)d\Bx
  +\int_{\dOm(t)}\{\BP\Bn\cdot\By+\BP^{*}\Bn\cdot\Bx\}dS(\Bx)\\
  +\int_{S(t)}\{V^{\rm sh}\jump{\Bv}\cdot(\By-\lump{\BF}\Bx)
  -p^{*}_{S(t)}\Bx\cdot\Bn^{\rm sh}\}dS(\Bx),
\end{multline}

To associate physical meaning with  different terms in the \rhs\ of (\ref{dynClapfin}), we first introduce several notations.  Thus, we observe that 
\[
  \Bb_{i}=-\left(\dif{\Bv}{t}\right)_{\rm reg}
\]
is  the regular component of the inertial (d'Alambert) body force
density.  The Noether identity (\ref{Noether0}) assigns to 
\[
  \Bb_{i}^{*}=\BF^{T}\left(\dif{\Bv}{t}\right)_{\rm reg}=-\BF^{T}\Bb_{i}
\]
the   meaning of the the regular component of the inertial configurational body force
density.

The  distributional part of the   inertial  force density
$ -\Ba= -\Md\Bv/\Md t$ is supported on the shock surface and  can be also   interperted as
inertial traction  on the shock; it takes the form
\[
  \Bt_{i}=V^{\rm sh}\jump{\Bv}.
\]
Similar reasoning suggests that
\[
\Bt_{i}^{*}=-V^{\rm sh}\lump{\BF}^{T}\jump{\Bv}=-\lump{\BF}^{T}\Bt_{i}
\]
can be  identified with the configurational inertial tractions on the
shock.

Using these notations, we can now present  the dynamic GCT in the final form
\begin{equation}
  \label{dynClapin0}
  \int_{\GO(t)}Ud\Bx=\frac{\mathfrak{S}+\mathfrak{I}}{n},
\end{equation}
where we separated  in the \rhs\ of (\ref{dynClapin0}) the ``static'' term $\mathfrak{S}$ from the ``dynamic'', or  inertial term  $\mathfrak{I}$.  The former 
\begin{equation}
  \label{staticClap} \mathfrak{S}=\int_{\dOm(t)}\{\BP\Bn\cdot\By+\BP^{*}\Bn\cdot\Bx\}dS(\Bx)-\int_{S(t)} p^{*}_{S(t)}\Bx\cdot\Bn^{\rm sh}dS(\Bx).
\end{equation}
represents the work of the physical and configurational forces on the
boundary of the domain, minus the energy dissipated on the shock
surface  (see formula (\ref{clapeyron}) in the  statics case).
The latter 
\begin{equation}
  \label{inertClap}
  \mathfrak{I}=\int_{\GO(t)}\{\Bb_{i}\cdot\By+\Bb_{i}^{*}\cdot\Bx\}d\Bx
  +\int_{S(t)}\{\Bt_{i}\cdot\By+\Bt_{i}^{*}\cdot\Bx\}dS(\Bx)
\end{equation}
represents  the work of the inertial physical and  configurational
forces inside  the domain, plus the work of the corresponding inertial
tractions acting on the shock surface.

\section{Applications}
\setcounter{equation}{0}
\label{sec:app}
In this section we show the   utility of the Clapeyron Theorem in various problems of mechanics and physics rooted in the Calculus of Variations.  Throughout this section we assume that $\By(\Bx)$ is a Lipschitz stationary extremal in $\GO$ that is class $C^{2}$ on a \nbh\ of $\dOm$, as in Remark~\ref{rem:reg}.
 
\subsection{Necessary condition of metastability}
\label{sec:dbc}
The problem of metastability is central for nonlinear elasticity theory and, more generally, for vectorial variational problems with nonconvex energy density.  In our special setting, it concerns the existence of strong local minimia  that are not automatically global.
To analyze this problem we can take advantage of the fact that GCT gives a formula for the energy   in terms of the boundary values of the solution and its gradients. The idea is to  consider the difference of two such expressions and rearrange it advantageously.
As our subsequent analysis shows a  helpful hint is  to  isolate in the boundary integral an  expression involving a Weierstrass excess function 
\begin{equation}
  \label{Weierstrass}
  \CE(\BF,\BG)=W(\BG)-W(\BF)-(W_{\BF}(\BF),\BG-\BF), 
\end{equation}
and then  compare the energies of two different configurations with  the same Dirchlet data. The result of the implied rearrangement can be formulated in the form of a theorem which relies essentially on the GCT. 
\begin{theorem}
  \label{cor:aff}
  Assume that the Lipschitz deformations $\By_{1}(\Bx)$ and $\By_{2}(\Bx)$ are
  stationary, of class $C^{2}$ near $\dOm$ and that
\begin{equation}
  \label{bc}
\By_{1}(\Bx)=\By_{2}(\Bx)=\By_{0}(\Bx),\quad \Bx\in\dOm.
\end{equation}
Then
\begin{equation}
\int_{\GO}\{W(\BF_{2})-W(\BF_{1})\}d\Bx=\nth{n}\int_{\Md\GO}
\{\CE(\BF_{1},\BF_{2})(\Bx,\Bn)+(\BP_{1}-\BP_{2})\Bn\cdot(\BF_{2}\Bx-\By_{0})\}dS(\Bx),
\label{naff1}
\end{equation}
where
\[
\BF_{j}=\Grad\By_{j}(\Bx),\qquad\BP_{j}=W_{\BF}(\Grad\By_{j}(\Bx)),\qquad
j=1,2.
\]
\end{theorem}
\begin{proof}
  Since the energy is scale-free, we can use  the nonlinear Clapeyron Theorem (\ref{Clapeyron})  and write 
  \[
\int_{\GO}W(\BF_{i})d\Bx=\nth{n}\int_{\dOm}\{\BP_{i}\Bn\cdot\By_{0}+\BP_{i}^{*}\Bn\cdot\Bx\}dS,
\]
Subtracting these two equations and expanding $\BP^{*}_{j}$ via (\ref{Esh-tensor}) we obtain
\begin{multline*}
  \int_{\GO}\{W(\BF_{2})-W(\BF_{1})\}d\Bx=\nth{n}\int_{\Md\GO}\{
  \CE(\BF_{1},\BF_{2})(\Bn\cdot\Bx)+
  (\BP_{1}-\BP_{2})\Bn\cdot(\BF_{2}\Bx-\By_{0})+\\
    \av{\BP_{1},\BF_{2}-\BF_{1}}(\Bn\cdot\Bx)
  +(\BP_{1}\Bn,(\BF_{1}-\BF_{2})\Bx)\}dS.
\end{multline*}
Since $\By_{1}$ and $\By_{2}$ are of class $C^{2}$ near $\dOm$ and $\By_{1}-\By_{2}=0$ on $\dOm$, we conclude that there exists a $C^{1}$ field $\Ba:\dOm\to\bb{R}^{m}$, such that $\BF_{1}-\BF_{2}=\Ba\otimes\Bn$. Therefore,
\[
(\BP_{1}\Bn,(\BF_{1}-\BF_{2})\Bx)=(\BP_{1}\Bn\cdot\Ba)(\Bn\cdot\Bx)=\av{\BP_{1},\BF_{1}-\BF_{2}}(\Bn\cdot\Bx),
\]
and  (\ref{naff1})   follows.
\end{proof}
\begin{remark}
Switching $\By_{1}$ and $\By_{2}$ we also have
  \begin{equation}
    \label{naff2}
    \int_{\GO}\{W(\BF_{2})-W(\BF_{1})\}d\Bx=\nth{n}\int_{\Md\GO}
\{-\CE(\BF_{2},\BF_{1})(\Bx,\Bn)+(\BP_{1}-\BP_{2})\Bn\cdot(\BF_{1}\Bx-\By_{0})\}dS(\Bx).
  \end{equation}
  Taking the average of (\ref{naff1}) and (\ref{naff2}) we obtain a more symmetric version of the energy increment formula
   \begin{multline*}
      \label{symd}
  \int_{\GO}\{W(\BF_{2})-W(\BF_{1})\}d\Bx= \\
  \nth{n}\int_{\Md\GO}\left\{
\hf(\CE(\BF_{1},\BF_{2})-\CE(\BF_{2},\BF_{1}))(\Bx,\Bn)+(\BP_{1}-\BP_{2},\lump{\BF}\Bx-\By_{0})
\right\}dS(\Bx),
   \end{multline*}
where $\lump{\BF}=(\BF_{1}+\BF_{2})/2$. Subtracting (\ref{naff1}) and (\ref{naff2}) we also obtain
    \begin{equation}
      \label{syms}
      \int_{\dOm}(\CE(\BF_{1},\BF_{2})+\CE(\BF_{2},\BF_{1}))(\Bx,\Bn)dS(\Bx)=
\int_{\dOm}((\BP_{2}-\BP_{1})\Bn,(\BF_{2}-\BF_{1})\Bx)dS(\Bx)
    \end{equation}
  \end{remark}

Observe further that if $\By_{1}$ and $\By_{2}$, satisfying (\ref{bc}), are smooth near $\dOm$, then $\BF_{2}-\BF_{1}=\Ba(\Bx)\otimes\Bn(\Bx)$ for all $\Bx\in\dOm$, where $\Bn(\Bx)$ is an outward unit normal to $\dOm$. Observe also that if $W(\BF)$ is rank-one convex, we have  for all $\Bx\in\dOm$
\[
\CE(\BF_{1},\BF_{2})=W(\BF_{2})-W(\BF_{1})-\av{\BP_{1}, \BF_{2}-\BF_{1}}\ge 0.
\]
In view of these observations and the fact that, if the origin is a star point in $\GO$, which can always be arranged
by a translation due to Remark~\ref{rem:shift},
then $(\Bx,\Bn)\ge 0$ for all $\Bx\in\GO$, we can conclude that  whenever $\By_{1}$ and $\By_{2}$ are two smooth extremals satisfying (\ref{bc}),  the following inequality holds
\begin{equation}
  \label{R1ineq}
  \int_{\GO}\{W(\BF_{2})-W(\BF_{1})\}d\Bx\ge\nth{n}\int_{\Md\GO}
\{(\BP_{1}-\BP_{2})\Bn\cdot(\BF_{2}\Bx-\By_{0})\}dS(\Bx).
\end{equation}
Now, if $\By_{2}(\Bx)$ is a global minimizer and $\By_{1}(\Bx)$ is a stationary extremal satisfying
\begin{equation}
  \label{nslmcrit}
 \int_{\Md\GO}(\BP_{1}-\BP_{2})\Bn\cdot(\BF_{2}\Bx-\By_{0})dS(\Bx)\ge 0,
\end{equation}
then
\[
0\ge\int_{\GO}\{W(\BF_{2})-W(\BF_{1})\}d\Bx\ge \int_{\Md\GO}(\BP_{1}-\BP_{2})\Bn\cdot(\BF_{2}\Bx-\By_{0})dS(\Bx)\ge 0,
\]
and hence $\By_{1}(\Bx)$ must also be a global minimizer. 
One can see that to allow for metastaility condition
(\ref{nslmcrit}) should be violated and since  the \lhs\ of
(\ref{nslmcrit}) can be regarded as a measure of non-affinity of
$\By_{2}(\Bx)$, such non-affinity can be then viewed as a necessary
condition of metastability \cite{grtrhard}.

\subsection{Quasiconvex envelope}

To formulate the problem considered in this section  we would need first to recall several standard definitions. The most important one is the definition of the quasiconvex envelope of the energy density $W(\BF)$---the largest quasiconvex function that does not exceed $W(\BF)$. It is denoted $QW$ and is often referred to  as quasiconvexification of $W$. There is a formula for $QW$ \cite{Dak08}
\begin{equation}
  \label{QW2}
  QW(\BF)|D|=\inf_{\myatop{\By|_{\Md D}=\BF\Bx}{\By\in W^{1,\infty}(D;\bb{R}^{m})}}\int_{D}W(\Grad\By)d\Bx.
\end{equation}
We would also need  the  notion of  rank-one convexity, i.e. convexity along all straight lines connecting two points that differ by a rank-one matrix. Similarly to the quasiconvex envelope of $W$ we define the rank-one convex envelope $RW$ to be the largest rank-one convex function that does not exceed $W$. It was shown in \cite{morr52} that every quasiconvex function is rank-one convex and therefore, $RW(\BF)\ge QW(\BF)$. The  question  whether $QW(\BF)=RW(\BF)$  has been answered in the negative in \cite{sver92}, except in the case $m=2$. Nonetheless examples where $QW(\BF)\not=RW(\BF)$ are extremely rare \cite{grab18}.

While, in many cases the variational problem (\ref{QW2}) has no solutions, it is known that the
quasiconvexification $QW$ of $W$ does not depend on the choice of the domain $D$ in (\ref{QW2}), and there are cases when one is be able to choose a particular domain $D$, where Lipschitz solutions do exist \cite{grtrhard}. In those  cases, the minimizers must be stationary extremals for which GCT is applicable.

 In this context one can regard such minimizers as known and pose the question whether   GCT   provides information about the function $QW$ in general, rather than about a specific solution of equilibrium equations. One  representative  result of this type, allowing one  to transfer the knowledge about  quasiconvexity from the boundary to the interior of the domain, is  formulated below in the form of a theorem:
\begin{theorem}
  \label{th:ELA}
Let $\GO\subset{\bb{R}^{d}}$ be a star-shaped domain with Lipschitz boundary. Assume that
\begin{itemize}
\item[(BVP)] For a given $\BF_{0}\in\bb{R}^{m\times n}$  the function  $\bra{\By}\in W^{1,\infty}(\GO;\bb{R}^{m})$ solves 
\begin{equation}
  \label{ELN1}
  \Div\BP(\Grad\bra{\By})=\Bzr,\quad\Div\BP^{*}(\Grad\bra{\By})=\Bzr,
\end{equation}
in the sense of distributions, and satisfies $\bra{\By}(\Bx)=\BF_{0}\Bx$ for all $\Bx\in\dOm$;
\item[(REG)] $\bra{\By}$ is of class $C^{2}$ near $\dOm$;
\item[(RCX)] $W(\Grad\bra{\By}(\Bx))=RW(\Grad\bra{\By}(\Bx))$ for a.e. $\Bx\in\dOm$;
\item[(Q=R)] $RW(\BF_{0})=QW(\BF_{0})$.
\end{itemize}
Then
\begin{enumerate}
\item[(i)] $RW(\BF_{0})=QW(\BF_{0})=\dashint_{\GO}W(\Grad\bra{\By})d\Bx$, i.e. $\bra{\By}(\Bx)$ is the global minimizer in (\ref{QW2}).
\item[(ii)] $QW(\Grad\bra{\By}(\Bx))=W(\Grad\bra{\By}(\Bx))$ for a.e. $\Bx\in\GO$.
\end{enumerate}
\end{theorem}
\begin{proof}
  Assumption (BVP) implies that the Clapeyron theorem is applicable to the stationary equilibrium $\bra{\By}(\Bx)$ of (\ref{QW2}) we obtain, by virture of (REG),
\begin{equation}
  \label{ClapRW}
  \int_{\GO}W(\Grad\bra{\By})d\Bx=\nth{n}\int_{\dOm}(W(\BF_{0}+\Ba\otimes\Bn)-
  \BP(\BF_{0}+\Ba\otimes\Bn)\Bn\cdot\Ba)(\Bn\cdot\Bx)dS. 
\end{equation}

Next, we appeal to \cite[Lemma~4.2]{grtrnc}.
We quote the part of the lemma needed for the proof.
\begin{lemma}
  \label{lem:r1cx}
Let $V(\BF_{0})$ be a rank-one convex function such that $V(\BF_{0})\le W(\BF_{0})$. Let
\[
\CA_{V}=\{\BF_{0}\in\bb{R}^{m\times n}:W(\BF_{0})=V(\BF_{0})\}.
\]
Then for every $\BF_{*}\in\CA_{V}$, $\Bb\in\bb{R}^{m}$, and $\Bm\in\bb{R}^{n}$
\begin{equation}
  \label{coolineq}
  V(\BF_{*}+\Bb\otimes\Bm)\ge W(\BF_{*})+\BP(\BF_{*})\Bm\cdot\Bb.
\end{equation}
\end{lemma}
We now apply Lemma~\ref{lem:r1cx} by choosing $V(\BF_{0})=RW(\BF_{0})$,
$\BF_{*}=\Grad\bra{\By}(\Bx)=\BF_{0}+\Ba\otimes\Bn$, and $\Bb\otimes\Bm=-\Ba\otimes\Bn$,
for each $\Bx\in\dOm$. The
assumption (RCX) implies that Lemma~\ref{lem:r1cx} is applicable, and therefore,
\begin{equation}
  \label{QWineq}
RW(\BF_{0})\ge W(\BF_{0}+\Ba\otimes\Bn)-\BP(\BF_{0}+\Ba\otimes\Bn)\Bn\cdot\Ba.
\end{equation}
Finally, we use the assumption that $\GO$ is star-shaped. If we choose the
origin at the star point, then the function $\Bn(\Bx)\cdot\Bx$ is always
non-negative at all points on $\dOm$. Therefore, inequality (\ref{QWineq})
together with the identity (\ref{ClapRW}) delivered by the Clapeyron theorem,
implies
\[
\int_{\GO}W(\Grad\bra{\By})d\Bx\le \frac{RW(\BF_{0})}{n}\int_{\dOm}(\Bn\cdot\Bx)dS=|\GO|RW(\BF_{0}).
\]
Since $|\GO|QW(\BF_{0})$ is the minimal value of the functional in (\ref{QW2}) we
obtain
\begin{equation}
  \label{RWQW}
  |\GO|RW(\BF_{0})\ge\int_{\GO}W(\Grad\bra{\By})d\Bx\ge|\GO|QW(\BF_{0}). 
\end{equation}
By the assumption (Q=R) we must have equality in both inequalities in (\ref{RWQW}) and conclude that $\bra{\By}(\Bx)$ is the global minimizer in (\ref{QW2}), proving (i). Property (ii) is a necessary condition for any strong local minimizer \cite{tahe02}, and hence must hold as well.
\end{proof}

\subsection{Probing the binodal}

It is well-known \cite{ball7677} that if $\By(\Bx)$ is a strong local minimizer of (\ref{non-param}), then the function $\BF\mapsto W(\Bx,\By(\Bx),\BF)$ must be quasiconvex at $\BF$, i.e. satisfy the inequality
\begin{equation}
  \label{qcxF}
  \int_{D}\{W(\BF+\Grad\BGf(\Bx))-W(\BF)\}d\Bx\ge 0\quad\forall\BGf\in W_{0}^{1,\infty}(D;\bb{R}^{m}), 
\end{equation}
where $D$ is any domain in $\bb{R}^{n}$. Thus, the gradient of any metastable configuration must avoid the binodal region of the phase space
\[
\CB=\left\{\BF\in\bb{R}^{m\times n}:\exists\BGf\in W_{0}^{1,\infty}(D;\bb{R}^{m}): \int_{D}\{W(\BF+\Grad\BGf(\Bx))-W(\BF)\}d\Bx<0\right\}.
\]
Its boundary $\Md\CB$ is called the \emph{binodal}.

 Since the function is quasiconvex, \IFF its binodal region is empty,  
the binodal region can then be  characterized by the inequality \cite{huss95}
$$
\CB=\{\BF\in\bb{R}^{m\times n}:QW(\BF)<W(\BF)\}.
$$
 Note  also that if $W(\BF_{0})$ is strictly quasiconvex at $\BF_{0}\in\bb{R}^{m\times n}$, then $\bra{\By}(\Bx)=\BF_{0}\Bx$ is the unique solution of (\ref{QW2}) \cite{tahe03}. If $\BF_{0}$ lies on the binodal, we expect the emergence of nontrivial equilibria in unbounded domains, where the volume fraction of the region where $\Grad\bra{\By}$
is a finite distance from $\BF_{0}$ to be zero, even though  in finite domains we still expect that $\bra{\By}(\Bx)=\BF_{0}\Bx$ be the unique solution. To follow this path  one can, for instance,  try to  find radial solutions of (\ref{ELN1}) in the entire space. As we show below, in case of the success,  Clapeyron's Theorem can   deliver a formula for $QW(\BF)$ for a range of $\BF$.

Suppose $m=n$, and $W(\BF)$ is objective and isotropic $$W(\BF)=w(v_{1},\ldots,v_{n}),$$ where $v_{j}$ are
the sigular values of $\BF$ and the function $w$ does not change after a permutation of its arguments.
We look for a radial solution $\By(\Bx)=\eta(r)\Hat{\Bx}$ of (\ref{ELN1}).
We first observe that
\[
\Grad\By=\eta'\tns{\Hat{\Bx}}+\frac{\eta}{r}(\BI_{n}-\tns{\Hat{\Bx}}).
\]
Hence, $\BF=\Grad\By$ has a singular value $|\eta'(r)|$ and $n-1$ singular values $|\eta(r)/r|$.
We will assume that $\eta(r)\ge 0$ and $\eta'(r)\ge 0$. In that case we define two functions
\[
  w_{1}(r)=\dif{w}{v_{1}}(\eta'(r),\eta(r)/r,\ldots,\eta(r)/r),\quad
  w_{2}(r)=\dif{w}{v_{2}}(\eta'(r),\eta(r)/r,\ldots,\eta(r)/r).
\]
The Piola stress tensor is \cite[p.~564]{ball82}
\[
\BP=w_{1}\tns{\Hat{\Bx}}+w_{2}(\BI_{n}-\tns{\Hat{\Bx}}).
\]
Therefore, $\eta(r)$ must solve the nonlinear second order ODE
\begin{equation}
  \label{rODE}
w_{11}\eta''+(n-1)w_{12}\left(\frac{\eta}{r}\right)'+\frac{n-1}{r}(w_{1}-w_{2})=0,
\end{equation}
where
\[
  w_{11}(r)=\hess{w}{v_{1}}(\eta'(r),\eta(r)/r,\ldots,\eta(r)/r),\quad
  w_{12}(r)=\mix{w}{v_{1}}{v_{2}}(\eta'(r),\eta(r)/r,\ldots,\eta(r)/r).
\]
If $\eta'(r)$ suffers a jump discontinuity at $r=r_{0}$, then for (\ref{ELN1}) to hold we must have
\[\jump{w_{1}}(r_{0})=0,\quad\jump{w}(r_{0})=w_{1}(r_{0})\jump{\eta'}.\]
Typically, we would look for a solution of the form 
\begin{equation}
  \label{radsol}
  \By(\Bx)=
  \begin{cases}
    f_{0}r\Hat{\Bx}, &|\Bx|<1,\\
    \eta(r)\Hat{\Bx},&|\Bx|>1
  \end{cases}
\end{equation}
In this case we must have
\[
  w_{1}(\eta'(1^{+}),f_{0},\ldots,f_{0})=w_{1}(f_{0},f_{0},\ldots,f_{0})
\]
\[
  w(\eta'(1^{+}),f_{0},\ldots,f_{0})=w(f_{0},f_{0},\ldots,f_{0})+w_{1}(f_{0},f_{0},\ldots,f_{0})
  (\eta'(1^{+})-f_{0}).
\]
These can be regarded as equations determining the values of $f_{0}$ and $\eta'(1^{+})$.
Ultimately, the radial solution $\eta(r)$ will be fixed by
\[
f_{\infty}=\lim_{r\to\infty}\frac{\eta(r)}{r}=\lim_{r\to\infty}\eta'(r).
\]
In Appendix~\ref{app:A2} we prove that
\begin{equation}
  \label{etainf}
  \eta(r)=f_{\infty}r+\frac{A}{r^{n-1}}+O(r^{-n}),\text{ as }r\to\infty.
\end{equation}

Next we show  how Clapeyron's Theorem can use  the knowledge of quasiconvexity at infinity (the boundary of all-space),  transfers it to the interior and  finally   delivers an explicit formula for $QW(f\BI_{n})$ for all $f\in[f_{0},f_{\infty}]$.

\begin{theorem}
  \label{th:QWrad}
Suppose that $\eta\in C^{2}([1,+\infty))$ is such that
\begin{itemize}
\item[(i)] $\displaystyle\bra{\By}(\Bx)=
  \begin{cases}
    f_{0}\Bx, &|\Bx|<1,\\
    \eta(|\Bx|)\Hat{\Bx},&|\Bx|>1
  \end{cases}
$ solves (\ref{ELN1}) in $\bb{R}^{n}$ in the sense of distributions.
\item[(ii)] $\displaystyle f_{\infty}=\lim_{r\to\infty}\frac{\eta(r)}{r}=\lim_{r\to\infty}\eta'(r)$ is such that
$W(f_{\infty}\BI_{n})=QW(f_{\infty}\BI_{n})$.
\end{itemize}
Then 
\begin{itemize}
\item[(a)] $W(f_{0}\BI_{n})=QW(f_{0}\BI_{n})$;
\item[(b)] for all $r\ge 1$
\begin{equation}
  \label{QWR}
    W\left(\eta'(r)\tns{\Hat{\Bx}}+\frac{\eta(r)}{r}(\BI_{n}-\tns{\Hat{\Bx}})\right)=
  QW\left(\eta'(r)\tns{\Hat{\Bx}}+\frac{\eta(r)}{r}(\BI_{n}-\tns{\Hat{\Bx}})\right). 
\end{equation}
\item[(c)] for every $r\ge 1$
  \begin{equation}
    \label{QWfI}
    QW\left(\frac{\eta(r)}{r}\BI_{n}\right)=w_{1}(r)\frac{\eta(r)}{r}+w(r)-\eta'(r)w_{1}(r).
  \end{equation}
\end{itemize}
\end{theorem}
\begin{proof}
According to (\ref{qcxF}), quasiconvexity is defined in terms of a bounded domain. In \cite{grtrmms} we have linked quasiconvexity to test functions in all-space. 
According to \cite[Theorem~4.6]{grtrmms} the existence of the radial solution $\bra{\By}(\Bx)$ of (\ref{ELN1}) as in assumption (i) implies that
\[
\int_{\bb{R}^{n}}\CE(f_{\infty}\BI_{n},\Grad\bra{\By})d\Bx=0,
\]
as
\[
  \Grad\bra{\By}-f_{\infty}\BI_{n}\in\CS=\{\BGf\in W^{1,\infty}(\bb{R}^{n};\bb{R}^{n}):
  \Grad\BGf\in L^{2}(\bb{R}^{n}; \bb{R}^{m\times n})\},
\]
according to (\ref{etainf}). Then,
by \cite[Theorem~3.10]{grtrmms}, the assumption (ii) implies
\[
\int_{\bb{R}^{n}}\CE(f_{\infty}\BI_{n},\Grad\By)d\Bx\ge 0
\]
for all $\By(\Bx)$, such that $\Grad\By(\Bx)-f_{\infty}\BI_{n}\in\CS$.
In that case conclusions (a) and (b) follow from  \cite[Theorem~4.3]{grtrmms}.

To prove part (c) we observe that radial solutions have the property that
$\bra{\By}(\Bx)=f_{R}\Bx$ for all $\Bx\in\Md B(0,R)$, where $f_{R}=\eta(R)/R$. Hence, by the already eastablished conclusion (b), Theorem~\ref{th:ELA}  applies and says that
\begin{equation}
  \label{QWrad}
  |B(0,R)|QW\left(\frac{\eta(R)}{R}\BI_{n}\right)=\int_{B(0,R)}W(\Grad\bra{\By})d\Bx.
\end{equation}
We complete the proof by the application of the Clapeyron theorem:
\[
  \BP\Bn=w_{1}\Hat{\Bx},\quad
  \BP^{*}=w\BI_{n}-\eta'w_{1}\tns{\Hat{\Bx}}+w_{2}\frac{\eta}{r}(\BI_{n}-\tns{\Hat{\Bx}})
\]
Thus, according to (\ref{Clapeyron}),
\[
  \int_{B(0,R)}W(\Grad\bra{\By})d\Bx=\nth{n}\int_{\Md B(0,R)}\left\{w_{1}\Hat{\Bx}\cdot\eta\Hat{\Bx}+
    (w-\eta'w_{1})\Hat{\Bx}\cdot R\Hat{\Bx}\right\}dS.
\]
We conclude that
\[
QW\left(\frac{\eta(R)}{R}\BI_{n}\right)=w_{1}(R)\frac{\eta(R)}{R}+w(R)-\eta'(R)w_{1}(R).
\]
The theorem is now proved.

\end{proof}

\subsection{Formation of a void}
\label{sub:void}

In this section we show how GCT can be  
used for solving   specific problems coming from engineering
applications. Specifically, we consider here a classical problem of
small void nucleation in a preloaded body which plays important role
in a broad range of mechanical theories \cite{mana87,ligu52b,cahn62,bafi93,mco96,ngtr02,trva04,gagne2005simulations,kko13,bks11,trci19,horg92,leal11,need87,ridr67,buri73}.  In this subsection we assume that both $W$ and all extremals are of class $C^{2}$.

Our starting point is  a classical hyperelastic body  with nonlinear energy density $W(\BF)$ which occupies a domain $\GO$ in its reference state and is loaded at the boundary by
tractions $\Bt(\Bx)$. The corresponding energy functional  is 
\begin{equation}
  \label{Eundam}
  E=\int_{\GO}W(\Grad\By)d\Bx-\int_{\dOm}\Bt(\Bx)\cdot\By\,dS(\Bx),
\end{equation} 
 Suppose   that the applied 
tractions   are  fully equilibrated   and therefore
\[
\int_{\dOm}\Bt\,dS=0,\quad\int_{\dOm}\{\Bt\otimes\Bx-\Bx\otimes\Bt\}dS=0.
\]
The resulting deformation
$\By(\Bx)$ can be found by solving the system of equations 
\begin{equation}
  \label{tracbc}
  \begin{cases}
    \Div W_{\BF}(\Grad\By)=0,&\Bx\in\GO,\\
  W_{\BF}(\Grad\By(\Bx))\Bn=\Bt(\Bx),&\Bx\in\dOm.
  \end{cases}
\end{equation}
and we assume that the solution is of class $C^{2}$.

To describe the process of void formation, we introduce  $\Go\subset\bb{R}^{n}$,   a bounded domain, such that $0\in\Go$.
Suppose that during the implied point ablation,  at  $0\in\GO$ a small stress-free void $\Ge\Go$ is
formed inside $\GO$ through the removal of mass. Suppose further that  $\Ge>0$ is so small that $\Ge\Go\subset\GO$, and $\GO\setminus\Ge\Go$ is connected. 

The  equilibrium deformation $\By_{\Ge}$, resulting from the described mass removal, is assumed to be of
class $C^{2}(\bra{\GO}\setminus\Ge\Go)$.  It must then solve the system of equations
\begin{equation}
  \label{Peps}
  \left\{
  \begin{array}[c]{cl}
\Div\BP(\Grad\By_{\Ge})=\Bzr,\quad &\Bx\in\GO\setminus\Ge\Go\\
\BP(\Grad\By_{\Ge})\Bn=\Bzr,\quad &\Bx\in\Md(\Ge\Go)\\
\BP(\Grad\By_{\Ge})\Bn=\Bt(\Bx).\quad &|\Bx|\in\dOm
  \end{array}\right.
\end{equation}
The potential energy of the resulting  configuration  is (up to a constant) 
\begin{equation}
  \label{Edam}
  E_{\Ge}=\int_{\GO\setminus\Ge\Go}W(\Grad\By_{\Ge})d\Bx-\int_{\dOm}\Bt(\Bx)\cdot\By_{\Ge}dS(\Bx),
\end{equation}

We can now  show how  GCT can be used to obtain  an expression for the 
energy release associated with the implied void formation and defined by the formula
\begin{equation}
  \label{DltE}
  \GD E=\lim_{\Ge\to 0}\frac{E_{\Ge}-E}{\Ge^{n}}.
\end{equation}
This apparently straightforward question has caused considerable
controversy in the past \cite{cotterell1997,kien08,mana87,kff06,sile67,spen65,eftis87,luma99,kesi96} 
and to the best of
our knowledge has never been so far resolved in the fully nonlinear
setting.

Note first that one can  expect  the  solution $\By_{\Ge}$ of the equilibrium problem
(\ref{Peps}) in the ablated medium to exhibit  an effective  boundary
layer near the cavity. Such  boundary layer would be  different from
classical boundary layers in singularly perturbed problems, where the exponential decay of the solutions within the
boundary layer sets the length scale. In the  case of a small cavity, it is the size
$\Ge$ of the defect that sets the length scale. In all other respects
the treatment of the solution with such  boundary layer can be expected  to follow  the
standard procedure of matching the inner and outer solutions.

To shortcut  the   formal mathematical study and simplify the presentation, we  make upfront the following assumptions which will be later at least partially justified by our explicit constructions.

\begin{itemize}
\item[(A1)] The \bvp s (\ref{tracbc}) and (\ref{Peps}) have unique, up
  to an additive constant solutions $\By(\Bx)$ and $\By_{\Ge}(\Bx)$,
  respectively.
\item[(A2)] There exists $\Bw\in C^{2}(\bra{\GO}\setminus\{0\})$, such
  \begin{equation}
  \label{far}
  \lim_{\Ge\to 0}\frac{\By_{\Ge}(\Bx)-\By(\Bx)}{\Ge^{n}}=\Bw(\Bx),
\end{equation}
in $C^{1}$ on all compact subsets of $\bra{\GO}\setminus\{0\}$. The
assumption \eqref{far} is motivated by the behavior of radial
solutions in the context of linear elasticity, see our
Appendix~\ref{app:A3} for details. In particular, we also have
\[
\lim_{\Ge\to 0}\frac{\Grad\By_{\Ge}(\Bx)-\Grad\By(\Bx)}{\Ge^{n}}=\Grad\Bw(\Bx)
\]
is uniform on compact subsets of $\GO\setminus\{0\}$. The function
$\By_{\rm out}(\Bx)=\By(\Bx)+\Ge^{n}\Bw(\Bx)$ will be called the outer
solution with the understanding that  the normalized
displacement increment $\Bw(\Bx)$ governs the ``far field'' induced by  the formation of a void. 
\item[(A3)] The inner solution in the immediate vicinity of the cavity is
\begin{equation}
  \label{local}
  \Grad\By_{\Ge}(\Ge\Bz)\to\Grad\By_{\rm in}(\Bz),\quad\Bz\not\in\Go
\end{equation}
uniformly, on compact subsets of $\bb{R}^{n}\setminus\Go$, where 
$\By_{\rm in}$ is the unique solution of
\begin{equation}
  \label{Pinf}
  \begin{cases}
\Div\BP(\Grad\By_{\rm in})=\Bzr, &\Bz\in\bb{R}^{n}\setminus\Go\\  
\BP(\Grad\By_{\rm in})\Bn=\Bzr, &\Bz\in\Md\Go\\  
\Grad\By_{\rm in}(\Bz)\to\Grad\By(0),\quad &|\Bz|\to\infty.
  \end{cases}
\end{equation}
This assumption  regarding  the inner solution in \eqref{local} is again motivated by the behavior of radial
solutions   of linearly elastic isotropic
   exterior problem which  in our case can be written in the form (see again our Appendix \ref{app:A3} for details):  
\begin{equation}
  \label{ulinresc}
  \Bu_{\rm in}^{\rm lin}(\Bz)=\frac{p\Bz}{n\Gk}+\frac{p\Bz}{2\mu(n-1)|\Bz|^{n}}=
  \lim_{\Ge\to 0}\frac{\Bu_{\Ge}(\Ge\Bz)}{\Ge},\quad|\Bz|>1.
\end{equation}
Since the deformation field $\By_{\rm in}(\Bz)$ characterizes the local field in the immediate vicinity of the void,  it will be interpreted as  the inner solution. 

\item[(A4)] There exists a constant $n\times m$ matrix $\BS$, such that
\[
\By_{\rm in}(\Bz)=\Grad\By(0)\Bz+\frac{\BS\Bz}{|\Bz|^{n}}+O(|\Bz|^{-n}),\text{ as }|\Bz|\to\infty,
\]
where $\By(\Bx)$ solves (\ref{tracbc}).
\item[(A5)] 
 One  can replace
  $\Grad\By_{\rm out}(\Ge\Bz)$ with $\Grad\By_{\rm in}(\Bz)$, when
  $|\Bz|$ is large, while $\Ge|\Bz|$ is small.  
\end{itemize}
 
We can now formulate our main result in the form of a theorem:
\begin{theorem}
  \label{th:nlGrif}
  Assume that assumptions (A1)--(A5) hold. Then
\begin{equation}
  \label{DeltaEfin}
\GD E=-\nth{n}\int_{\Md\Go}W(\Grad\By_{\rm in}(\Bz))(\Bn\cdot\Bz)dS(\Bz).
\end{equation}
\end{theorem}
 Note that a related result was obtained in \cite{ridr67} and its relation with Theorem~\ref{th:nlGrif} is discussed in Appendix~\ref{app:RD}.
  \begin{proof}
Observe first that if we apply  GCT (\ref{Clapeyron})  to the solution
$\By$ of (\ref{tracbc}) 
we obtain  
\[
E=\nth{n}\int_{\dOm}\{\BP\Bn\cdot\By+\BP^{*}\Bn\cdot\Bx\}dS-\int_{\dOm}\Bt\cdot\By\,dS.
\]
In view of  the \bc s $\BP\Bn=\Bt$ we also know that 
\[
\BP^{*}\Bn=W(\Grad\By)\Bn-(\Grad\By)^{T}\BP\Bn=W(\Grad\By)\Bn-(\Grad\By)^{T}\Bt.
\]
Therefore
\begin{equation}
  \label{E0}
  E=
\nth{n}\int_{\dOm}\{\Bt\cdot\By+W(\Grad\By)(\Bn\cdot\Bx)-
\Bt\cdot\Grad\By\Bx\}dS-\int_{\dOm}\Bt\cdot\By\,dS.
\end{equation}
Similarly, since the solution $\By_{\Ge}$  of (\ref{Peps}) has the same \bc s as $\By$ on $\dOm$ and
traction-free \bc $\BP(\Grad\By_{\Ge})\Bn=0$ on $\Md\Go$,  application of the  GCT (\ref{Clapeyron})  finally  gives
\begin{multline}
  \label{Eeps}
  E_{\Ge}=\nth{n}\int_{\dOm}\{\Bt\cdot\By_{\Ge}+W(\Grad\By_{\Ge})(\Bn\cdot\Bx)
-\Bt\cdot\Grad\By_{\Ge}\Bx\}dS\\
-\nth{n}\int_{\Md(\Ge\Go)}W(\Grad\By_{\Ge})(\Bn\cdot\Bx)dS
-\int_{\dOm}\Bt\cdot\By_{\Ge}dS.
\end{multline}

Subtracting (\ref{E0}) and (\ref{Eeps}), dividing by $\Ge^{n}$ and passing to the limit
using our assumptions (\ref{far}) and (\ref{local}), we compute
\begin{align} \label{prel}
\begin{split}
\GD E &=\nth{n}\int_{\dOm}\{\Bt\cdot\Bw+
\av{\BP(\Grad\By),\Grad_{\dOm}\Bw}(\Bn\cdot\Bx)-\Bt\cdot\Grad_{\dOm}\Bw\Bx\}dS\\ 
 &-\nth{n}\int_{\Md\Go}\{W(\Grad\By_{\rm in})(\Bn\cdot\Bz)dS(\Bz)-\int_{\dOm}\Bt\cdot\Bw\,dS.
 \end{split}
\end{align}
This is our preliminary formula for the  desired  energy increment and now we  show that this formula can be simplified and expressed
entirely in terms of the inner  solution. Performing in \eqref{prel}  integration by parts in the terms containing $\Grad_{\dOm}\Bw$ and 
using the formula
\begin{equation}
\int_{\dOm}\Bf\cdot\Grad_{\dOm}\phi\,dS(\Bx)=
\int_{\dOm}\{\phi(\Bf\cdot\Bn)\Grad_{\dOm}\cdot\Bn-\phi\Grad_{\dOm}\cdot\Bf\}dS,
\label{ibp}
\end{equation}
we obtain
\begin{equation}
  \label{DeltaE}
\GD E=J-\nth{n}\int_{\Md\Go}W(\Grad\By_{\rm in})(\Bn\cdot\Bz)dS(\Bz),
\end{equation}
where
\begin{equation}
  \label{J}
J=\nth{n}\int_{\dOm}\{(\Grad_{\dOm}\BP)[\Bx,\Bn,\Bw]-(\Bn\cdot\Bx)(\Grad_{\dOm}\cdot\BP)\cdot\Bw\}dS,
\end{equation}
and
\[
(\Grad_{\dOm}\BP)[\Bx,\Bn,\Bw]=(\Grad_{\dOm}P_{i}^{\Ga}\cdot\Bx)n_{\Ga}w^{i}; 
\]
note that we assumed summation over repeated indices. To finalize the proof of the
theorem we need to show that $J=0$ which is done in our Appendix~\ref{app:A5}.
\end{proof}
\begin{remark} 
We emphasize  that   Theorem~\ref{th:nlGrif} represents the energy release due to the  formation of a traction-free void in terms of the negative of the work of configurational forces as the formula we have actually  used is 
\begin{equation}
  \label{E1}
\GD E=-\nth{n}\int_{\Md\Go}\BP^{*}(\Grad\By_{\in})\Bn\cdot\Bz\,dS.
\end{equation}
This illustrates the physical meaning of the second term in GCT (\ref{Clapeyron}), as the energy associated with creation (in this case destruction) of a body.
\end{remark} 

The derivation presented  above is not the only way how  GCT can be used in this problem  and below we  provide an alternative  proof of Theorem~\ref{th:nlGrif}.
Here we  take advantage of the known  properties of the Weierstrass excess function
$\CE(\BF,\BG)$, defined in (\ref{Weierstrass}).
\begin{proof}[Alternative proof of Theorem~\ref{th:nlGrif}]
We begin by rewriting
the work of the loading device in the problems  (\ref{Eundam}) and (\ref{Edam}) in the form of
volume integral, respectively, as 
$$
  E=\int_{\GO}\{W(\Grad\By)-\av{\BP(\Grad\By),\Grad\By}\}d\Bx,$$
  and 
$$E_{\Ge}=\int_{\GO\setminus(\Ge\Go)}\{W(\Grad\By_{\Ge})-\av{\BP(\Grad\By_{\Ge}),\Grad\By_{\Ge}}\}d\Bx.
$$
Rearranging the terms we can write
$E_{\Ge}-E=I_{1}+I_{2}+I_{3}$, where
\[ I_{1}=\int_{\GO\setminus(\Ge\Go)}\CE(\BF,\BF_{\Ge})d\Bx+
  \int_{\GO\setminus(\Ge\Go)}\av{\BP(\Grad\By)-\BP(\Grad\By_{\Ge}),\Grad\By_{\Ge}-\Grad\By}\}d\Bx,
\]
\[
  I_{2}=\int_{\GO\setminus(\Ge\Go)}\av{\BP(\Grad\By)-\BP(\Grad\By_{\Ge}),\Grad\By}\}d\Bx,
\]
\[
  I_{3}=-\int_{\Ge\Go}\{W(\Grad\By)-\av{\BP(\Grad\By),\Grad\By}\}d\Bx.
\]
Integration by parts in $I_{2}$, while taking the \bc s from
(\ref{tracbc}) and (\ref{Peps}) into account, gives
\[
I_{2}=-\int_{\Md(\Ge\Go)}\BP(\Grad\By)\Bn\cdot\By\,dS(\Bx)=-\int_{\Ge\Go}\av{\BP(\Grad\By),\Grad\By}d\Bx,
\]
where a divergence theorem on $\Ge\Go$ was used in the second equality above.
Thus,
\[
I_{2}+I_{3}=-\int_{\Ge\Go}W(\Grad\By)d\Bx.
\]
Changing variables $\Bz=\Ge\Bx$ and passing to the limit, we obtain
\[
\lim_{\Ge\to 0}\Ge^{-n}(I_{2}+I_{3})=-|\Go|(W(\BF_{0}).
\]
We also have
\[
  \Ge^{-n}I_{1}=\int_{(\Ge^{-1}\GO\setminus\Go)}\{\CE(\Grad\By(\Ge\Bz),\Grad\By_{\Ge}(\Ge\Bz))+\av{\BP(\Grad\By(\Ge\Bz))-\BP(\Grad\By_{\Ge}(\Ge\Bz)),\Grad\By_{\Ge}(\Ge\Bz)-\Grad\By(\Ge\Bz)}\}d\Bz,
\]
To pass to the limit, we observe that $\Grad\By_{\Ge}(\Bx)\to\Grad\By(\Bx)$, as $\Ge\to 0$,
when $|\Bz|\to\infty$, such that $\Bx=\Ge\Bz$ is finite. 

We note that the main   reason for using the Weierstrass excess function $\CE(\BF,\BG)$ in this calculation is the estimate
\begin{equation}
  \label{Wdecay}
|\CE(\BF,\BG)|\le C|\BF-\BG|^{2},
\end{equation}
provided both $\BF$ and $\BG$ belong to a fixed compact subset of $\bb{R}^{m\times n}$, determining the value of the constant $C$ in (\ref{Wdecay}). 

Note next that the second term in $I_{1}$ clearly satisfies the same estimate, and that is what allows us to replace $\Grad\By_{\Ge}(\Ge\Bz)$ by $\Grad\By_{\rm in}(\Bz)$, $\Grad\By(\Ge\Bz)$ by $\BF_{0}=\Grad\By(0)$, and $\Ge^{-1}\GO$ by $\bb{R}^{n}$ in the expression for $\Ge^{-n}I_{1}$. Hence,
\begin{equation}
  \label{I1}
    \lim_{\Ge\to 0}\Ge^{-n}I_{1}=\int_{\bb{R}^{n}\setminus\Go}\{\CE(\BF_{0},\Grad\By_{\rm in}(\Bz))+
  \av{\BP_{0}-\BP(\Grad\By_{\rm in}(\Bz)),\Grad\By_{\rm in}(\Bz)-\BF_{0}}\}d\Bz, 
\end{equation}
where $\BF_{0}=\Grad\By(0)$, and $\BP_{0}=\BP(\Grad\By(0))$.

Since $\Grad\By_{\rm in}$ is a smooth extremal of $\Hat{W}(\BF)=\CE(\BF_{0},\BF)$, and because of the estimate (\ref{Wdecay}), we can use GCT (\ref{Clapeyron}) in the first term in (\ref{I1}) and integrate by parts in the second one, while we remain in the \emph{ infinite} domain $\bb{R}^{n}\setminus\Go$. We obtain
\begin{multline*}
  \int_{\bb{R}^{n}\setminus\Go}\CE(\BF_{0},\Grad\By_{\rm in}(\Bz))d\Bz=-\nth{n}\int_{\Md\Go}
\{(\BP(\Grad\By_{\rm in})-\BP_{0})\Bn\cdot\By_{\rm in}+\CE(\BF_{0},\Grad\By_{\rm in})(\Bn,\Bz)-\\
(\BP(\Grad\By_{\rm in})-\BP_{0})\Bn\cdot\Grad\By_{\rm in}(\Bz)\Bz\}dS=
-\nth{n}\int_{\Md\Go}\{\BP_{0}\Bn\cdot(\Grad\By_{\rm in}\Bz-\By_{\rm in})+
\CE(\BF_{0},\Grad\By_{\rm in})(\Bn,\Bz)\}dS,
\end{multline*}
where we took into account the \bc\ $\BP(\Grad\By_{\rm in})\Bn=0$ on $\Md\Go$.
Similarly,
\begin{multline*}
  \int_{\bb{R}^{n}\setminus\Go}\av{\BP_{0}-\BP(\Grad\By_{\rm in}(\Bz)),\Grad\By_{\rm in}(\Bz)-\BF_{0}}d\Bz=
  -\int_{\Md\Go}\BP_{0}\Bn\cdot(\By_{\rm in}(\Bz)-\BF_{0}\Bz)dS=\\
  |\Go|\av{\BP_{0},\BF_{0}}-\int_{\Md\Go}\BP_{0}\Bn\cdot\By_{\rm in}(\Bz)dS.
\end{multline*}
Expanding $\CE(\BF_{0},\Grad\By_{\rm in})$, we obtain
\begin{multline*}
    -\nth{n}\int_{\Md\Go}\CE(\BF_{0},\Grad\By_{\rm in}(\Bz))(\Bn,\Bz)dS=
  -\nth{n}\int_{\Md\Go}W(\Grad\By_{\rm in}(\Bz))(\Bn,\Bz)d\BS+\\(W(\BF_{0})-\av{\BP_{0},\BF_{0}})|\Go|+
  \nth{n}\int_{\Md\Go}\av{\BP_{0},\Grad\By_{\rm in}}(\Bn,\Bz)dS 
\end{multline*}
Thus,
\[
\GD E=-\nth{n}\int_{\Md\Go}W(\Grad\By_{\rm in})(\Bn\cdot\Bz)dS(\Bz)+J^*,
\]
where
\begin{equation}
  \label{Jst}
  J^*=-\nth{n}\int_{\Md\Go}\left\{\BP_{0}\Bn\cdot(\Grad\By_{\rm in}\Bz-\By_{\rm in})-\av{\BP_{0},\Grad\By_{\rm in}}\Bz\cdot\Bn +n\BP_{0}\Bn\cdot\By_{\rm in}\right\}dS.
\end{equation}
To complete the proof it remains to  show that $J^*=0$, which is  shown by direct computation  in our Appendix~\ref{app:A6}. 
\end{proof}

\begin{remark}
The obtained results allow us to elucidate  a delicate point, which has led to a famous ``Griffith's error" in the analogous calculation conducted  in the context of linear elasticity \cite{grif21,kesi96}. Indeed, if  instead of our way of computing the limit of $\Ge^{-n}I_{2}$ above, we  would have   passed to the limit ``directly'' and replaced $\Grad\By_{\Ge}(\Ge\Bz)$ by $\Grad\By_{\rm in}(\Bz)$, $\Grad\By(\Ge\Bz)$ by $\BF_{0}=\Grad\By(0)$, and $\Ge^{-1}\GO$ by $\bb{R}^{n}$, then  in the re-scaled expression $\Ge^{-n}I_{2}$,  we would have obtained
\begin{multline*}
  \Tld{I}_{2}=\lim_{R\to\infty}\int_{B(0,R)\setminus\Go}\av{\BP_{0}-\BP(\Grad\By_{\rm in}),\BF_{0}}\}d\Bz=
  -\int_{\Md\Go}\BP_{0}\Bn\cdot\BF_{0}\Bz\,dS+\\
  \lim_{R\to\infty}\int_{\Md B(0,R)}(\BP_{0}-\BP(\Grad\By_{\rm in}))\Bn\cdot\BF_{0}\Bz\,dS=
  -|\Go|\av{\BP_{0},\BF_{0}}+\\
  \lim_{R\to\infty}\int_{\Md B(0,R)}(\BP_{0}-\BP(\Grad\By_{\rm in}))\Bn\cdot\BF_{0}\Bz\,dS, 
\end{multline*}
which suggests that   the discrepancy $G$ between the correct answer $-|\Go|\av{\BP_{0},\BF_{0}}$ and the result obtained by Griffith, who originally followed  the ``direct" approach presented above, was ``hiding at infinity'':
 \begin{equation}
  \label{GrerrG}
  G=\lim_{R\to\infty}\int_{\Md B(0,R)}(\BP_{0}-\BP(\Grad\By_{\rm in}))\Bn\cdot\BF_{0}\Bz\,dS.
 \end{equation}
The value of $G$ in the special case  of linear elasticity studied by Griffith in \cite{grif21}  is computed explicitly  in Appendix~\ref{app:A8}.  Griffith himself corrected the formula for $\GD E$ in the follow up paper  \cite{grif24}, however, since he did not provide any details, the correct result was independently reconstructed by many other authors \cite{sned46,eshe51,irwin57,buec58,sand60,spen65,ridr67,sile67,cher67,bies68}.
\end{remark}

\subsection{Case study: linear elasticity}
\label{sub:linel}

 As we have already mentioned, the defining commonality of all nonlinear versions of  Clapeyron's Theorem is that they express the elastic energy of a stable configuration in terms of boundary tractions, which implies both physical and configurational tractions. In this section we will juxtapose the classical Clapeyron's Theorem (\ref{Clap}) against the  formulas (\ref{Clapeyron}), (\ref{semiC}) and (\ref{semiL}) which are all relevant,  given that   in classical  linear elasticity the energy density function is both scale-invariant and 2-homogeneous. Indeed, consider the energy density
\begin{equation}
  \label{linel}
  W(\Bx,\BF)=\hf\av{\SFC(\Bx)\BF_{\rm sym},\BF_{\rm sym}},
\end{equation}
where 
\[
  \BF_{\rm sym}=\frac{1}{2}( \BF+\BF^{T}),
\]
 and where the tensor of elastic moduli
$\SFC(\Bx)$ takes values in the space of positive definite quadratic forms on $\Sym(\bb{R}^{n})$---the space of $n\times n$ symmetric real matrices, and can otherwise be an arbitrary bounded and measurable function of $\Bx$. Note first that according to  (\ref{linel}) the function 
$W(\Bx,\BF)$ is $p$-homogenous in the sense of (\ref{phom}) with $p=2$.   This  means, in particular,  that, given a   solution of
\[
\Div(\SFC(\Bx)e(\Bu))=0,
\]
understood   in the sense of distributions, we can apply formula (\ref{semiC})
and obtain the classical Clapeyron theorem (\ref{Clap})
\begin{equation}\label{classclap}
E[\Bu]:=\hf\int_{\GO}\av{\SFC(\Bx)e(\Bu),e(\Bu)}d\Bx=\hf\int_{\dOm}\BGs(\Bx)\Bn\cdot\Bu(\Bx)dS(\Bx),
\end{equation}

Note next that when the  body is homogeneous and the tensor of elastic moduli $\SFC$ is constant, the energy density function (\ref{linel}) is both scale-free and 2-homogeneous in the sense of (\ref{scalefree}) and (\ref{phom}), with $p=2$, respectively. In this case  one can also show that all equilibrium solutions are smooth and therefore, by the Noether formula (\ref{Noether0}), they are also stationary. Hence, in addition to (\ref{semiC}) that has the form (\ref{classclap}) in the linearly elastic context, formulas (\ref{semiL}), (\ref{PPst}), and (\ref{Clapeyron}) are also applicable. However, in contrast to (\ref{classclap}), these formulas involve the skew-symmetric part $w(\Bu)$ of $\Grad\Bu$,
\[
w(\Bu)=\frac{1}{2} (\Grad\Bu-(\Grad\Bu)^{T}),
\]
representing the infinitesimal rotations. Indeed, formula (\ref{semiL}) takes the form
\begin{equation}
  \label{Clapeyron223}
  \hf\int_{\GO}\av{\SFC e(\Bu),e(\Bu)}d\Bx=\nth{n-2}\int_{\dOm}\left\{
\hf\av{\BGs,e(\Bu)}\Bn\cdot\Bx-\BGs\Bn\cdot (e(\Bu)+w(\Bu))\Bx\right\}dS(\Bx).
\end{equation}
Formula (\ref{PPst}) relating the \rhs s of (\ref{classclap}) and (\ref{Clapeyron223}) can be now  rewritten as an (apparently new)  integral relation satisfied on the boundary by the antisymmetric part of the gradient of an equilibrium solution:
\begin{equation}
  \label{neweq}
  \int_{\dOm}\BGs\Bn\cdot w(\Bu)\Bx\,dS=\int_{\dOm}\left\{\hf\av{\BGs,e(\Bu)}(\Bn\cdot\Bx)-
  \BGs\Bn\cdot e(\Bu)\Bx-\frac{n-2}{2}\BGs\Bn\cdot\Bu\right\}dS.
\end{equation}
In the special case of $n=2$ formula (\ref{neweq}) simplifies
\[
\int_{\dOm}\BGs\Bn\cdot w(\Bu)\Bx\,dS=\int_{\dOm}\left\{\hf\av{\BGs,e(\Bu)}(\Bn\cdot\Bx)-
  \BGs\Bn\cdot e(\Bu)\Bx\right\}dS.
\]
Finally, formula (\ref{Clapeyron}), representing the  linear version of GCT, takes the   form
\begin{equation}
  \label{LinClapew}
  \hf\int_{\GO}\av{\SFC(\Bx)e(\Bu),e(\Bu)}d\Bx=\nth{n}\int_{\dOm}\left\{\BGs\Bn\cdot(\Bu-w(\Bu)\Bx)-
    \BGs\Bn\cdot e(\Bu)\Bx+\hf\av{\BGs,e(\Bu)}(\Bn\cdot\Bx)\right\}dS. 
\end{equation}
which is clearly different from the classical Clapeyron's Theorem.

Note also that the conservation law (\ref{pi}) can be now written as a differential equation for $w(\Bu)$ in terms of $e(\Bu)$:
\[
\Div(\BGs w(\Bu)\Bx)=\Div\left(\hf\av{\BGs,e(\Bu)}\Bx-\BGs e(\Bu)\Bx-\frac{n-2}{2}\BGs\Bu\right).
\]
Interestingly, we can also rewrite it as a differential equation for $\Grad\Bu$
\begin{equation}
  \label{newdeq}
  \Div(\BGs\Grad\Bu\Bx)=\hf\Div(\av{\BGs,e(\Bu)}\Bx-(n-2)\BGs\Bu),
\end{equation}
which in two space dimension simplifies as a representation in terms $e(\Bu)$ only:
\[
\Div(\BGs\Grad\Bu\Bx)=\hf\Div(\av{\BGs,e(\Bu)}\Bx).
\]
To the best of our knowledge,   formulas (\ref{Clapeyron223})--(\ref{newdeq}) are new. It is likely that they did not appear in the literature on linear elasticity due to the rather striking appearance in them  of the skew-symmetric part of the gradient $w(\Bu)$. Nonetheless, formulas (\ref{Clapeyron223})-- (\ref{neweq})  may be useful in applications, as we show, for instance,  in our Appendix~\ref{app:A7}.

Despite their fundamentally  variational origin,  formulas (\ref{Clapeyron223})-- (\ref{neweq})  could have been be also discovered  by direct computation.  For instance, the direct differentiation in (\ref{newdeq}), taking into account the equilibrium equations $\Div\BGs=0$ and $\BGs^{T}=\BGs$ gives
\[
\av{\BGs,\Grad\Grad\Bu\Bx}=\hf\Grad\av{\BGs,e(\Bu)}\cdot\Bx.
\]
This equation can be now  verified directly using a version of Betti reciprocity which still relies essentially on the fact that  the  elastic tensor is constant
\[
\av{\dif{\BGs}{x^{i}},\BGve}=\av{\BGs, \dif{\BGve}{x^{i}}}.
\]

If we turn to applications, i is  instructive to first of all reconsider  in the context of linear elasticity the problem of void formation which we have  discussed in the general nonlinear setting in Section~\ref{sub:void}. 

Recall that our goal in this setting  is  to compute the difference between the potential energies
\[
E=\hf\int_{\GO}\av{\SFC e(\Bu),e(\Bu)}d\Bx-\int_{\dOm}\Bt\cdot\Bu\,dS,
\]
and
\[
E_{\Ge}=\hf\int_{\GO_{\Ge}}\av{\SFC e(\Bu_{\Ge}),e(\Bu_{\Ge})}d\Bx-\int_{\dOm}\Bt\cdot\Bu_{\Ge}dS
\]
of the undamaged and damaged bodies loaded by boundary tractions
\begin{equation}
  \label{uueps}
  \left\{
  \begin{array}[c]{cl}
\Div\SFC e(\Bu)=\Bzr,\quad &\Bx\in\GO\\
\SFC e(\Bu)\Bn=\Bt(\Bx),\quad &\Bx\in\dOm
  \end{array}\right.,\qquad
  \left\{
  \begin{array}[c]{cl}
\Div\SFC e(\Bu_{\Ge})=\Bzr,\quad &\Bx\in\GO\setminus\Ge\Go\\
\SFC e(\Bu_{\Ge})\Bn=\Bzr,\quad &\Bx\in\Md(\Ge\Go)\\
\SFC e(\Bu_{\Ge})\Bn=\Bt(\Bx),\quad &\Bx\in\dOm.
  \end{array}\right.
\end{equation}
Specifically, we would like  to compute
\begin{equation}
  \label{DeltaElin}
  \GD E=\lim_{\Ge\to 0}\frac{E_{\Ge}-E}{\Ge^{n}}. 
\end{equation}

Instead of adapting to the case of linear elasticity our general Theorem~\ref{th:nlGrif} which relies on GCT, we  use here the classical Clapeyron's Theorem (\ref{classclap}) to prove  an apparently different formula for $\GD E$: 
\begin{theorem}
  \label{th:mazya}
  \begin{equation}
    \label{mazya}
\GD E=-\hf\int_{\Md\Go}\BGs(0)\Bn(\Bz)\cdot\Bu_{\rm in}(\Bz)dS(\Bz),
  \end{equation}
  where $\Bu_{\rm in}$ solves an exterior \bvp
\begin{equation}
  \label{Linf}
  \begin{cases}
\Div\SFC e(\Bu_{\rm in})=\Bzr, &\Bz\in\bb{R}^{n}\setminus\Go\\
(\SFC e(\Bu_{\rm in}))\Bn=\Bzr, &\Bz\in\Md\Go\\
e(\Bu_{\rm in})\to\BGve_{0}, &|\Bz|\to\infty,
  \end{cases}
\end{equation}
  where $\BGve_{0}=e(\Bu)(0)$, and where $\BGs(\Bx)=\SFC e(\Bu(\Bx))$.
\end{theorem}
For equivalent results in different specialized settings, see  \cite{sned46,eshe51,irwin57,buec58,sand60,spen65,ridr67,cher67,bies68}.
\begin{proof}
We first note that by the classical Clapyeron theorem (\ref{Clap})
\begin{equation}
  \label{linE}
  E=-\hf\int_{\GO}\av{\SFC e(\Bu),e(\Bu)}d\Bx,
\end{equation}
and
\begin{equation}
  \label{linEeps}
 E_{\Ge}=-\hf\int_{\GO_{\Ge}}\av{\SFC e(\Bu_{\Ge}),e(\Bu_{\Ge})}d\Bx,
\end{equation}
where $\Bu_{\Ge}$ and $\Bu(\Bx)$ are the displacement vectors in the
damaged and the undamaged domains, respectively. 

Let $\Bw_{\Ge}(\Bx)=\Bu_{\Ge}(\Bx)-\Bu(\Bx)-\Ba_{\Ge}$, where
$\Ba_{\Ge}=\Bu_{\Ge}(\Ge\Bz_{0})-\Bu(\Ge\Bz_{0})$, and $\Bz_{0}\in\Md\Go$
is an arbitrarily selected, and henceforth fixed, point.
Then, the vector field $\Bw_{\Ge}$ satisfies
\begin{equation}
  \label{weps}
\left\{\begin{array}[c]{cl}
\Div\SFC e(\Bw_{\Ge})=\Bzr,\quad &\Bx\in \GO\setminus\Ge\Go\\
\SFC e(\Bw_{\Ge})\Bn=-\BGs(\Bx)\Bn,\quad &\Bx\in\Md(\Ge\Go)\\
\SFC e(\Bw_{\Ge})\Bn=\Bzr,\quad &\Bx\in\Md\GO.
  \end{array}\right.
\end{equation}
Observe that
\[
\hf\int_{\GO_{\Ge}}\av{\SFC e(\Bw_{\Ge}),e(\Bw_{\Ge})}d\Bx= \hf\int_{\GO_{\Ge}}\av{\SFC
e(\Bu_{\Ge}),e(\Bu_{\Ge})}d\Bx+ \hf\int_{\GO_{\Ge}}\av{\SFC e(\Bu),e(\Bu)}d\Bx-
\int_{\GO_{\Ge}}\av{\SFC e(\Bu_{\Ge}),e(\Bu)}d\Bx
\]
Integrating the last term by parts, and using $\SFC e(\Bu_{\Ge})\Bn=0$ on $\Md(\Ge\Go)$, we obtain
\[
\int_{\GO_{\Ge}}\av{\SFC e(\Bu_{\Ge}),e(\Bu)}d\Bx=
\int_{\dOm}\BGs(\Bx)\Bn\cdot\Bu(\Bx)dS(\Bx)=\int_{\GO}\av{\SFC
e(\Bu),e(\Bu)}d\Bx.
\]
Therefore,
\[
\hf\int_{\GO_{\Ge}}\av{\SFC e(\Bw_{\Ge}),e(\Bw_{\Ge})}d\Bx=E-E_{\Ge}-\hf\int_{\Ge\Go}\av{\SFC e(\Bu),e(\Bu)}d\Bx.
\]
Then
\[
E_{\Ge}-E=-\hf\int_{\GO_{\Ge}}\av{\SFC e(\Bw_{\Ge}),e(\Bw_{\Ge})}d\Bx-
\hf\int_{\Ge\Go}\av{\SFC e(\Bu),e(\Bu)}d\Bx.
\]
Applying Clapeyron's theorem again, this time to $\Bw_{\Ge}$, we get
\[
E_{\Ge}-E=-\hf\int_{\Md\Go_{\Ge}}\BGs(\Bx)\Bn\cdot\Bw_{\Ge}dS(\Bx)-
\hf\int_{\Ge\Go}\av{\SFC e(\Bu),e(\Bu)}d\Bx,
\]
where $\Bn$ is the outer unit normal to $\Md\Go$.
Now let us change variables $\Bx=\Ge\Bz$ and obtain
\[
E_{\Ge}-E=-\frac{\Ge^{n}}{2}\left(
\int_{\Md\Go}\BGs(\Ge\Bz)\Bn(\Bz)\cdot\Tld{\Bu}_{\Ge}(\Bz)dS(\Bz)-
\int_{\Go}\av{\SFC e(\Bu)(\Ge\Bz),e(\Bu)(\Ge\Bz)}d\Bz\right),
\]
where $\Tld{\Bu}_{\Ge}(\Bz)=\Ge^{-1}\Bw_{\Ge}(\Ge\Bz)$.
We observe that
$\Grad\Tld{\Bu}_{\Ge}=\Grad\Bu_{\Ge}(\Ge\Bz)-\Grad\Bu(\Ge\Bz)\to
\Grad\Bu_{\rm in}(\Bz)-\Grad\Bu(0)$, uniformly on compact subsets of
$\bb{R}^{n}\setminus\{0\}$, while $\Tld{\Bu}_{\Ge}(\Bz_{0})=0$. It follows
that $\Tld{\Bu}_{\Ge}(\Bz)\to\Bu_{\rm in}(\Bz)-\Grad\Bu(0)\Bz-\Ba_{0}$,
where $\Ba_{0}=\Bu_{\rm in}(\Bz_{0})-\Grad\Bu(0)\Bz_{0}$.
Therefore,
\[
\GD E=\lim_{\Ge\to 0}\frac{E_{\Ge}-E}{\Ge^{n}}=
-\frac{1}{2}\int_{\Md\Go}\BGs(0)\Bn(\Bz)\cdot(\Bu_{\rm in}(\Bz)-\Grad\Bu(0)\Bz-\Ba_{0})dS(\Bz)
-\frac{|\Go|}{2}\av{\BGs(0),\BGve_{0}}.
\]
Observing that
\[
\int_{\Md\Go}\BGs(0)\Bn(\Bz)\cdot\Ba_{0}\,dS(\Bz)=0,
\]
and
\[
  \int_{\Md\Go}\BGs(0)\Bn(\Bz)\cdot\Grad\Bu(0)\Bz\,dS(\Bz)=|\Go|\av{\BGs(0),\Grad\Bu(0)}=
  |\Go|\av{\BGs(0),\BGve_{0}},
\]
we obtain
\[
\GD E=\lim_{\Ge\to 0}\frac{E_{\Ge}-E}{\Ge^{n}}=
-\frac{1}{2}\int_{\Md\Go}\BGs(0)\Bn(\Bz)\cdot\Bu_{\rm in}(\Bz) \,dS(\Bz).
\]
\end{proof}
The apparent difference between the conclusions  of a  general  Theorem~\ref{th:nlGrif} and  a special  Theorem~\ref{th:mazya} is discussed in our Appendix \ref{app:A7}. Specifically, we show  that formulas  (\ref{DeltaEfin}) and (\ref{mazya}) are in full agreement, and that their different appearance is due to the fact that in their derivations we used two very different versions of Clapeyron's theorem. Specifically, in the proof of Theorem~\ref{th:nlGrif} we relied on the version (\ref{Clapeyron}) expressing
the scale-free nature of elasticity theory. Instead, in our proof of the Theorem~\ref{th:mazya} we used  the version (\ref{semiC}) manifesting the homogeneity of degree 2 of linear elastic energy.
 The different  forms  \eqref{Clapeyron223} and (\ref{LinClapew}) taken by the nonlinear versions of the  Claperon's Theorem  in the setting of linear elasticity relate the two approaches. It is also important to mention that a different path of addressing the void formation problem in the nonlinear setting was taken in \cite{ridr67} where the authors obtained the result  more  resembling (\ref{mazya}), than (\ref{DeltaEfin}). This is discussed  vis-a-vis our own approach in Appendix~\ref{app:RD}.

\section{Conclusions}
\setcounter{equation}{0}
\label{sec:conc}
In this paper we   revisited the classical Clapeyron's Theorem of  the  linear elasticity theory
that expresses the energy of an equilibrium configuration in terms of the work of physical forces on the boundary.  We derived various  nonlinear analogs of this well known engineering result  by revealing its  link  with  the classical work of E. Noether on variational symmetries in nonlinear field theories.  We   showed that in view of this connection,  the obtained  family of formulas expressing ``energy in terms of work''  can be  interpreted as a rather general statement within Calculus of Variations reflecting partial  variational symmetries of the corresponding theories. Specifically, we showed that in the case of scale invariance  combined with $p$-homogeneity one  obtains   not only a direct generalization of the Clapeyron's theorem of linear elasticity containing the stress tensor but also its dual version expressed through the  corresponding Eshelby tensor.    In this respect an important   mechanical aspect   of our work  is that it   brings together   physical and configurational forces and allows one to  differentiate between the part of the elastic energy stored due to the action of applied loads and the ``cold work type" part of the energy emerging due to elastic incompatibility. More generally,  the proposed  broader reading of the Clapeyron's theorem allowed us to link it to various apparently unrelated classical results,   including the Green's formula in nonlinear elastostatics   and the expression of  the minimal value of the equilibrium energy through the Weierstrass excess function.  Of particular interest are the applications of our generalized Clapeyron Theorem (GCT) to elastodynamics where it naturally accounts for the inevitable  formation of  shocks.  We have also shown that our study   provides new  useful tools allowing one to advance in various  applied variational problems where explicit solutions are not readily available. In particular,  we showed that  GCT can be potentially useful in generating stability limits for elastic solids undergoing fracture  and phase transitions.

\medskip

\textbf{Acknowledgments.} YG was supported by the
National Science Foundation under Grant No. DMS-2305832. The work of LT was supported by the French   grant ANR-10-IDEX-0001-02 PSL. The authors appreciate support from  the  Oberwolfach Fellows program. 

\appendix
\section{Appendix}
\setcounter{equation}{0}
\label{sec:apps}

\subsection{The far-field asymptotics of solutions in \eqref{rODE}}
\label{app:A2}
Here we prove the asymptotics (\ref{etainf}) of the solutions of \eqref{rODE} by substituting the ansatz
\begin{equation}
  \label{etasym}
  \eta(r)=f_{\infty}r+Ar^{-\Ga}
\end{equation}
into \eqref{rODE}. The asymptotics at $r=\infty$ of various terms in \eqref{rODE} are as follows:
\[
  w_{11}\eta''=w_{11}^{\infty}\frac{\Ga(\Ga+1)A}{r^{\Ga+2}}+o(r^{-(\Ga+2)}),\quad
 w_{12}\left(\frac{\eta}{r}\right)'=-w_{12}^{\infty}\frac{(\Ga+1)A}{r^{\Ga+2}}+o(r^{-(\Ga+2)}),
\]
\[
  w_{1}=w_{1}^{\infty}-w_{11}^{\infty}\frac{\Ga A}{r^{\Ga+1}}
  +(n-1)w_{12}^{\infty}\frac{A}{r^{\Ga+1}}+o(r^{-(\Ga+1)}),
\]
\[
  w_{2}=w_{2}^{\infty}-w_{21}^{\infty}\frac{\Ga A}{r^{\Ga+1}}+w_{22}^{\infty}\frac{A}{r^{\Ga+1}}+
  (n-2)w_{23}^{\infty}\frac{A}{r^{\Ga+1}}+o(r^{-(\Ga+1)}),
\]
where
\[
w_{i}^{\infty}=\dif{w}{v_{i}}(f_{\infty},\ldots,f_{\infty}),\quad w_{ij}^{\infty}=\mix{w}{v_{i}}{v_{j}}(f_{\infty},\ldots,f_{\infty}).
\]
Thus, the indicial equation for $\Ga$ is, taking into account that
\[
  w_{1}^{\infty}=w_{2}^{\infty},\quad w_{23}^{\infty}=w_{21}^{\infty}=w_{12}^{\infty},\quad w_{22}^{\infty}=w_{11}^{\infty},
\]  
\[
  w_{11}^{\infty}\Ga(\Ga+1)-(n-1)(\Ga+1)w_{12}^{\infty}+(n-1)(\Ga+1)(w_{12}^{\infty}-w_{11}^{\infty})=0.
\]
Dividing both sides by $\Ga+1$ (since $\Ga=-1$ corresponds to the first term $f_{\infty}r$ in the asymptotics (\ref{etasym})) we obtain
\[
w_{11}^{\infty}\Ga-(n-1)w_{12}^{\infty}+(n-1)(w_{12}^{\infty}-w_{11}^{\infty})=0.
\]
This becomes $w_{11}^{\infty}(\Ga-n+1)=0$, and therefore, $\Ga=n-1$, as claimed.

\subsection{Outer solution for a spherical void in the linear elastic media}
 \label{app:A3}
Here we motivate the assumption that
\begin{equation}
  \label{far1}
  \lim_{\Ge\to 0}\frac{\By_{\Ge}(\Bx)-\By(\Bx)}{\Ge^{n}}=\Bw(\Bx),
\end{equation}
by the observed behavior of solutions in a
radially symmetric isotropic linear elastic example, where
\[
W_{\rm lin}(\BF)=\frac{\Gl}{2}(\Trc\BF)^{2}+\mu|\BF_{\rm sym}|^{2},
\]
$\GO=B(0,1)$ is a unit ball under uniform tension $\BGs\Bn=p\Bn$, $|\Bx|=1$, and
$\Go=B(0,1)$ is a spherical cavity. In this case the displacements in the original unablated medium are 
\begin{equation}
  \label{ulin0}
  \Bu(\Bx)=\frac{p\Bx}{n\Gk}, 
\end{equation}
where $\Gk$ is the bulk modulus of the material, while after ablation
they are given by
\begin{equation}
  \label{ulin}
  \Bu^{\rm lin}_{\Ge}(\Bx)=\frac{p\Bx}{n\Gk(1-\Ge^{n})}
  +\frac{p\Ge^{n}\Bx}{2\mu(n-1)(1-\Ge^{n})|\Bx|^{n}}. 
\end{equation}
We easily verify that
\begin{equation}
  \label{wlin}
  \Bw^{\rm lin}(\Bx)=\lim_{\Ge\to   0}\frac{\Bu_{\Ge}(\Bx)-\Bu(\Bx)}{\Ge^{n}}=\frac{p\Bx}{n\Gk}+
  \frac{p\Bx}{2\mu(n-1)|\Bx|^{n}},\quad\Bx\in B(0,1)\setminus\{0\},
\end{equation}
where the convergence is uniform on compact subsets of $\bra{B(0,1)}\setminus\{0\}$.

\subsection{Completion of the proof of Theorem~\ref{th:nlGrif}}
\label{app:A5}
To finalize the proof of the
Theorem~\ref{th:nlGrif} we show here that indeed   $J=0$.  
We first observe that in (\ref{J}) we can replace $\Grad_{\dOm}\BP$ with $\Grad\BP$ and $\Grad_{\dOm}\cdot\BP$ with $\Div\BP=0$. Hence,
\begin{equation}
  \label{Jw}
    J=\int_{\dOm}\dif{P_{i}^{\Ga}(\Grad\By)}{x^{\Gb}}x^{\Gb}n_{\Ga}w^{i}dS(\Bx). 
\end{equation}
We first observe that
\[
\int_{\dOm}\dif{P_{i}^{\Ga}(\Grad\By)}{x^{\Gb}}x^{\Gb}n_{\Ga}dS=\int_{\GO}\dif{}{x_{\Ga}}\left(\dif{P_{i}^{\Ga}(\Grad\By)}{x^{\Gb}}x^{\Gb}\right)dS=0,
\]
since $\Div\BP(\Grad\By(\Bx))=0$ in $\GO$. It follows that we can replace $\Bw$ in the formula for $J$ with $\Bw+\Bw_{0}$ for any constant $\Bw_{0}$. This constant will be chosen later as convenient.
For now $\Bw$ will stand for $\Bw+\Bw_{0}$. 

Now we want to apply the divergence theorem to the \rhs\ of (\ref{Jw}), except we need to keep in mind that $\Bw(\Bx)$ has a singularity at $\Bx=0$. Therefore, we can apply the divergence theorem to $\GO\setminus(B(0,M\Ge))$, where $M$ is a large but fixed constant.
We have
\[
  J=\lim_{\Ge\to 0}\int_{\GO\setminus(B(0,M\Ge))}\dif{P_{i}^{\Ga}(\Grad\By)}{x^{\Gb}}
  \dif{w^{i}}{x^{\Ga}}x^{\Gb}d\Bx+\int_{\Md(B(0,M\Ge))}\dif{P_{i}^{\Ga}(\Grad\By)}
  {x^{\Gb}}x^{\Gb}n_{\Ga}w^{i}dS(\Bx),
\]
where in the last integral $\Bn(\Bx)$ stands for the outward unit normal to $\Md(B(0,M\Ge))$. 
Taking into account that
\[
\dif{P_{i}^{\Ga}(\Grad\By)}{x^{\Gb}}=\SFL_{ij}^{\Ga\Gg}(\Grad\By)\mix{y^{j}}{x^{\Gg}}{x^{\Gb}},\quad
\SFL_{ij}^{\Ga\Gg}(\BF)=\mix{W(\BF)}{F^{i}_{\Ga}}{F^{j}_{\Gg}},
\]
we have
\[
  J=\lim_{\Ge\to 0}\int_{\GO\setminus(B(0,M\Ge))}(\SFL(\Grad\By)\Grad\Bw)_{j}^{\Gg}
  \mix{y^{j}}{x^{\Gg}}{x^{\Gb}}x^{\Gb}+\int_{\Md(B(0,M\Ge))}\SFL_{ij}^{\Ga\Gg}(\Grad\By)
  \mix{y^{j}}{x^{\Gg}}{x^{\Gb}}x^{\Gb}n_{\Ga}w^{i}dS(\Bx).
\]
Observing that
\[
\mix{y^{j}}{x^{\Gg}}{x^{\Gb}}x^{\Gb}=\Grad(\Grad\By(\Bx)\Bx-\By(\Bx)+\By(0))^{j}_{\Gg},
\]
and that $\Div (\SFL(\Grad\By)\Grad\Bw)=0$, while $(\SFL(\Grad\By)\Grad\Bw)\Bn=0$ on $\dOm$, we obtain
\begin{equation}
  \label{preanswer}
  J=\lim_{\Ge\to 0}\int_{\Md(B(0,M\Ge))}\SFL_{ij}^{\Ga\Gg}(\Grad\By)
  \left(\mix{y^{j}}{x^{\Gg}}{x^{\Gb}}x^{\Gb}n_{\Ga}w^{i}-(\dif{y^{j}}{x^{\Gb}}x^{\Gb}
    -y^{j}+y(0))n_{\Gg}\dif{w^{i}}{x^{\Ga}}\right)dS.
\end{equation}
We note that in order to proceed further we need to apply the matching principle, whereby
$\Grad\By_{\Ge}(\Bx)\approx\Grad\By(\Bx)+\Ge^{n}\Grad\Bw(\Bx)$, when $|\Bx|$ is not too small, must match $\Grad\By_{\Ge}(\Ge\Bz)\approx\Grad\By_{\rm in}(\Bz)$, when $|\Bz|$ is not too large.
Hence,
$\Bx=\Ge\Bz$, and $\Bz\in B(0,M)$ for some large constant $M$, represents the region of validity of both approximations. Thus, in order to evaluate $J$, we replace
\[
\Grad\Bw(\Ge\Bz)\to\Ge^{-n}(\Grad\By_{\rm in}(\Bz)-\Grad\By(\Ge\Bz)),\quad\Bz\in B(0,M),
\]
and
\[
\Bw(\Ge\Bz)\to\Ge^{-n}(\Ge\By_{\rm in}(\Bz)-\By(\Ge\Bz)+\By(0)).
\]
where the free constant in $\Bw$ is now chosen so that the above approximation is valid.
Making the replacement in (\ref{preanswer}) and changing variables $\Bx=\Ge\Bz$, we obtain
\begin{multline*}
  J=\lim_{M\to\infty}\lim_{\Ge\to 0}\int_{\Md B(0,M)}\SFL_{ij}^{\Ga\Gg}(\Grad\By(\Ge\Bz))
  \left[\mix{y^{j}}{x^{\Gg}}{x^{\Gb}}(\Ge\Bz)z^{\Gb}n_{\Ga}(\Bz)(\Ge y^{i}_{\rm in}(\Bz)-y^{i}(\Ge\Bz)+y^{i}(0))-\right.\\\left. \left(\dif{y^{j}}{x^{\Gb}}(\Ge\Bz)z^{\Gb} -\frac{y^{j}(\Ge\Bz)-y^{j}(0)}{\Ge}\right)\left(\dif{y_{\rm in}^{i}}{z^{\Ga}}(\Bz)-\dif{y^{i}}{x^{\Ga}}(\Ge\Bz)\right)n_{\Gg}(\Bz)\right]\,dS(\Bz)=0,
\end{multline*}
because
\[
  \lim_{\Ge\to 0}\left(\dif{y^{j}}{x^{\Gb}}(\Ge\Bz)z^{\Gb}-\frac{y^{j}(\Ge\Bz)-y^{j}(0)}{\Ge}\right)=
  \dif{y^{j}}{x^{\Gb}}(0)z^{\Gb}-\dif{y^{j}}{x^{\Ga}}(0)z^{\Ga}=0.
\]
We conclude that (\ref{DeltaEfin}) holds.

\subsection{Completion of the alternative proof of the Theorem~\ref{th:nlGrif}}
\label{app:A6}
To finalize the alternative proof of the Theorem~\ref{th:nlGrif} we now show that indeed $J^*=0$, where $J^{*}$ is given by (\ref{Jst}).  
First we observe that we can write $-nJ^*=\int_{\Md\Go}\Bv(\Bz)\cdot\Bn(\Bz)dS$,
where
\[
  v^{\Ga}=P_{0i}^{\Ga}\left(\dif{y^{i}}{z^{\Gb}}z^{\Gb}-y^{i}\right)
  -P_{0i}^{\Gb}\dif{y^{i}}{z^{\Gb}}z^{\Ga}+nP_{0i}^{\Ga}y^{i}.
\]
Hence,
\[
  -nJ^*=\lim_{R\to\infty}\left[\int_{\Md(B(0,R)\setminus\Go)}\Div\Bv d\Bz-
    \int_{\Md B(0,R)}\Bv\cdot\Bn\,dS\right].
\]
We compute
\[
  \Div\Bv=P_{0i}^{\Ga}\mix{y^{i}}{z^{\Gb}}{z^{\Ga}}z^{\Gb}
  -nP_{0i}^{\Gb}\dif{y^{i}}{z^{\Gb}}-P_{0i}^{\Gb}\mix{y^{i}}{z^{\Gb}}{z^{\Ga}}z^{\Ga}
  +nP_{0i}^{\Ga}\dif{y^{i}}{z^{\Ga}}=0.
\]
Therefore,
\[
nJ^*=\lim_{R\to\infty}\int_{\Md B(0,R)}\left\{\BP_{0}\Bn\cdot(\Grad\By_{\rm in}\Bz-\By_{\rm in})-\av{\BP_{0},\Grad\By_{\rm in}}(\Bn,\Bz) +n\BP_{0}\Bn\cdot\By_{\rm in}\right\}dS.
\]
We observe that we can replace $\By_{\rm in}(\Bz)$ in the formula for $nJ$ with $\Tld{\By}(\Bz)=\By_{\rm in}(\Bz)-\BF_{0}\Bz$. We also note that there exists an $n\times n$ matrix $\BS$ (``polarization'' tensor \cite{amka2007}), depending on the shape of $\Go$, such that
\begin{equation}
  \label{farfield}
    \Tld{\By}(\Bz)=\frac{\BS\Bz}{|\Bz|^{n}}+O(|\Bz|^{-n}),\text{ as }|\Bz|\to\infty. 
\end{equation}
Substituting this asymptotics into the formula for $nJ$ we obtain
\begin{equation}
  \label{polarII}
  nJ^*=\int_{\bb{S}^{n-1}}\{n\BP_{0}\BGx\cdot\BS\BGx-\av{\BP_{0},\BS}\}dS(\BGx).
\end{equation}
It is clear that
\[
  \BM=\int_{\bb{S}^{n-1}}\tns{\BGx}dS(\BGx)=\Gl\BI_{n},
\]
since $\BR\BM\BR^{T}=\BM$ for any rotation matrix $\BR$. Taking traces we obtain $n\Gl=|\bb{S}^{n-1}|$.
It follows that
\[
  \int_{\bb{S}^{n-1}}n\BP_{0}\BGx\cdot\BS\BGx\,dS(\BGx)=n\av{\BS^{T}\BP_{0},\Gl\BI_{n}}=
n\Gl\av{\BP_{0},\BS}=\int_{\bb{S}^{n-1}}\av{\BP_{0},\BS}dS.
\]
Thus, we have proved that $nJ^*=0$. Hence, formula (\ref{DeltaEfin}) is
established by a different method, where the Clapeyron theorem has once again played a key role.

\subsection{Reconciliation of Theorems \ref{th:nlGrif} and \ref{th:mazya}}
\label{app:A7}
Here we show that the  statements of a  general  Theorem~\ref{th:nlGrif} and  a special  Theorem~\ref{th:mazya}  are in full agreement. To this end we temporarily denote the \rhs\ of the nonlinear result (\ref{DeltaEfin}) by $\GD_{\rm NL}E$ and the \rhs\ of the linear result (\ref{mazya}) by $\GD_{L}E$.
Our starting point  will be again  the Weierstrass $\CE$-function that ensures the proper decay at infinity for the Clapeyron theorem to operate in an exterior domain. When $W(\BF)$ is given  y (\ref{linel}) with constant tensor of material moduli $\SFC$, we compute
\begin{multline*}
  E_{\infty}=\int_{\bb{R}^{n}\setminus\Go}\CE(\BGve_{0},\BGve_{\rm in})d\Bz=
  \int_{\bb{R}^{n}\setminus\Go}\left\{\hf\av{\SFC\BGve_{\rm in}\BGve_{\rm in}}
    -\hf\av{\SFC\BGve_{0}\BGve_{0}}-\av{\SFC\BGve_{0},\BGve_{\rm in}-\BGve_{0}}\right\}dS=\\
  \hf\int_{\bb{R}^{n}\setminus\Go}\av{\SFC\Tld{\BGve},\Tld{\BGve}}d\Bz,
\end{multline*}
where $\Tld{\BGve}=\BGve_{\rm in}-\BGve_{0}$. The energy density $\hf\av{\SFC\Tld{\BGve},\Tld{\BGve}}$ is both 2-homogeneous and scale-invariant function of $\Tld{\BF}=\Grad\Tld{\Bu}$, where $\Tld{\Bu}=\Bu_{\rm in}-\BF_{0}\Bz-\Bu_{0}$, is a $C^{2}$ function, which solves
\[
  \begin{cases}
    \Div\SFC e(\Tld{\Bu})=0,&\Bz\not\in\Go\\
    (\SFC e(\Tld{\Bu}))\Bn=-\BGs_{0}\Bn,&\Bz\in\Md\Go,\\
    \Tld{\Bu}(\Bz)\to 0,&|\Bz|\to\infty.
  \end{cases}
\]
Hence, both versions (\ref{LinClapew}) and (\ref{classclap}) of Clapeyron's theorem are applicable.
Moreover, our use of the Weierstrass $\CE$-function allows us to use formulas (\ref{LinClapew}) and (\ref{classclap}) in the exterior domain. Therefore, on the one hand we have
\begin{equation}
  \label{GCTuse}
E_{\infty}=-\nth{n}\int_{\Md\Go}\left\{(\SFC\Tld{\BGve})\Bn\cdot\Tld{\Bu}+\hf\av{\SFC\Tld{\BGve},\Tld{\BGve}}(\Bz\cdot\Bn)-(\SFC\Tld{\BGve})\Bn\cdot\Grad\Tld{\Bu}\Bz\right\} d\Bz,
\end{equation}
due to (\ref{LinClapew}), and on the other
\begin{equation}
  \label{CTuse}
  E_{\infty}=-\hf\int_{\Md\Go}(\SFC\Tld{\BGve})\Bn\cdot\Tld{\Bu}dS=
  \hf\int_{\Md\Go}\BGs_{0}\Bn\cdot\Tld{\Bu}dS,
\end{equation}
due to (\ref{classclap}). The negative signs in (\ref{GCTuse}) and (\ref{CTuse}) come from
our convention that $\Bn$ denotes the \emph{outward} unit normal to $\Md\Go$, which the the inward-pointing normal for the boundary of $\bb{R}^{n}\setminus\Go$.
The similarity between formula (\ref{CTuse}) for $E_{\infty}$ and formula (\ref{mazya}) for $\GD_{L}E$ is apparent. Therefore, we can rewrite (\ref{CTuse}) as
\begin{equation}
  \label{EinfIII}
  E_{\infty}=\GD_{L}E+\hf\int_{\Md\Go}\BGs_{0}\Bn\cdot(\Tld{\Bu}+\Bu_{\rm in})dS=
 \GD_{L}E+\int_{\Md\Go}\BGs_{0}\Bn\cdot\Tld{\Bu}dS+\hf\av{\BGs_{0},\BGve_{0}}|\Go|.
\end{equation}

Replacing $(\SFC\Tld{\BGve})\Bn$ by $-\BGs_{0}\Bn$ in (\ref{GCTuse}) and expanding
\[
  \av{\SFC\Tld{\BGve},\Tld{\BGve}}=\av{\SFC\BGve_{\rm in},\BGve_{\rm in}}
  -2\av{\SFC\BGve_{\rm in},\BGve_{0}}+\av{\SFC\BGve_{0},\BGve_{0}} =
  \av{\SFC\BGve_{\rm in},\BGve_{\rm in}}-2\av{\SFC\Tld{\BGve},\BGve_{0}}-\av{\SFC\BGve_{0},\BGve_{0}},
\]
formula (\ref{GCTuse}) becomes
\begin{equation}
  \label{EinfII}
    E_{\infty}=\GD_{NL}E-\nth{n}\int_{\Md\Go}\BGs_{0}\Bn\cdot(\Grad\Tld{\Bu}\Bz-\Tld{\Bu})
  -\av{\BGs_{0},\Tld{\BGve}}(\Bz\cdot\Bn)dS+\hf\av{\BGs_{0},\BGve_{0}}|\Go|. 
\end{equation}
Hence, formulas (\ref{EinfIII}) and (\ref{EinfII}) show that $\GD_{NL}E=\GD_{L}E$ \IFF $J=0$, where
\[
J=\int_{\Md\Go}\{\BGs_{0}\Bn\cdot(\Grad\Tld{\Bu}\Bz-\Tld{\Bu})
  -\av{\BGs_{0},\Tld{\BGve}}(\Bz\cdot\Bn)+n\BGs_{0}\Bn\cdot\Tld{\Bu}\}dS.
\]
As in the analysis in Appendix~\ref{app:A6} we first write write $J=\int_{\Md\Go}\Bv\cdot\Bn\,dS$ and then verify that $\Div\Bv=0$.
Indeed,
\[ v^{\Ga}=\Gs_{0i}^{\Ga}\left(\dif{u^{i}}{z^{\Gb}}z^{\Gb}-u^{i}\right)-
  \Gs_{0i}^{\Gb}\dif{u^{i}}{z^{\Gb}}z^{\Ga}+n\Gs_{0i}^{\Ga}u^{i}.
\]
Hence,
\[ \Div\Bv=\Gs_{0i}^{\Ga}\left(\mix{u^{i}}{z^{\Ga}}{z^{\Gb}}z^{\Gb}+\dif{u^{i}}{z^{\Gb}}\Gd_{\Ga}^{\Gb}
    -\dif{u^{i}}{z^{\Ga}}\right)-\Gs_{0i}^{\Gb}\mix{u^{i}}{z^{\Ga}}{z^{\Gb}}z^{\Ga}
  -n\Gs_{0i}^{\Gb}\dif{u^{i}}{z^{\Gb}}+n\Gs_{0i}^{\Ga}\dif{u^{i}}{z^{\Ga}}=0.
\]
It follows that
\[
J=-\lim_{R\to\infty}\int_{\Md B_{R}}\{\BGs_{0}\Bn\cdot(\Grad\Tld{\Bu}\Bz-\Tld{\Bu})
  -\av{\BGs_{0},\Tld{\BGve}}(\Bz\cdot\Bn)+n\BGs_{0}\Bn\cdot\Tld{\Bu}\}dS.
\]
As in Appendix~\ref{app:A6} we use (\ref{farfield})
and the calculation fully analogous to (\ref{polarII}) shows that  $J=0$, establishing the equality of $\GD_{NL}E$ and $\GD_{L}E$.
\subsection{Rice and Drucker's analysis of void formation}
\label{app:RD}
In this section we revisit, for the sake of completeness, the calculation of the total energy release upon the removal of a subregion $\Go\subset\GO$, due to Rice and Drucker \cite{ridr67} because their result is in some sense intermediate between the asymptotic formulas (\ref{DeltaEfin}) and (\ref{mazya}) and because their argument is conducted almost entirely in the context of nonlinear elasticity, while leading to a formula resembling (\ref{mazya}).

Consider a body $\GO$ loaded by a combination of ``fixed grips'' (displacement \bc s) on $\dOm_{D}$, tractions on $\dOm_{N}$, and body forces in $\GO$. The total energy functional is\footnote{In \cite{ridr67} Rice and Drucker were working in the context of geometrically linear, physically nonlinear elasticity. However, nothing in their analysis, reproduced below appealed to the assumption of geometric linearity.}
\[
E[\By]=\int_{\GO}\{W(\Bx,\Grad\By)-\Bb\cdot\By\}d\Bx-\int_{\dOm_{N}}\By\cdot\Bt\,dS,
\]
where for notational convenience we used $\By(\Bx)$ instead of the displacements $\Bu(\Bx)=\By(\Bx)-\Bx$ in the work of external body forces and boundary tractions.
We emphasize that the explicit dependence of the energy density $W$ on $\Bx$ permits heterogeneous materials and location-dependent prestress.

Let $\By_{0}$ be the energy minimizer solving
\begin{equation}
  \label{eq0}
  \begin{cases}
    \Div\BP(\Bx,\BF_{0})=\Bb,&\Bx\in\GO,\\
    \By_{0}(\Bx)=\Bg(\Bx),&\Bx\in\dOm_{D},\\
 \BP(\Bx,\BF_{0})\Bn=\Bt,&\Bx\in\dOm_{N},   
  \end{cases}
\end{equation}
where $\Bg(\Bx)-\Bx$ are the prescribed displacments of the Dirichlet part $\dOm_{D}$ of the boundary.
Let $\By_{d}$ be the energy minimizer of
\[
E_{d}[\By]=\int_{\GO\setminus\Go}\{W(\Bx,\Grad\By)-\Bb\cdot\By\}d\Bx-\int_{\dOm_{N}}\By\cdot\Bt\,dS,
\]
where $\Go\subset\GO$ is the damaged region that has been
destroyed. The deformation $\By_{d}$ solves the \bvp
\begin{equation}
  \label{eqd}
  \begin{cases}
    \Div\BP(\Bx,\BF_{d})=\Bb,&\Bx\in\GO\setminus\Go,\\
    \By_{d}(\Bx)=\Bg(\Bx),&\Bx\in\dOm_{D},\\
    \BP(\Bx,\BF_{d})\Bn=\Bt,&\Bx\in\dOm_{N},\\
    \BP(\Bx,\BF_{d})\Bn=0,&\Bx\in\Md\Go.
  \end{cases}
\end{equation}
Our goal is to find the formula for the energy increment
\[
  \GD E=E[\By_{0}]-E_{d}[\By_{d}]
\]
expressed entirely in terms of the fields on $\Md\Go$.

We begin with the obvious expression
\[
  E[\By_{0}]-E[\By_{d}]=E_{\Go}+\int_{\GO\setminus\Go}\{W(\Bx,\BF_{0})-W(\Bx,\BF_{d})
  -\Bb\cdot(\By_{0}-\By_{d})\}d\Bx-\int_{\dOm_{N}}\Bt\cdot(\By_{0}-\By_{d})dS,
\]
where we used the shorthand $\BF_{0}=\Grad\By_{0}$ and $\BF_{d}=\Grad\By_{d}$, and where
\[
E_{\Go}=\int_{\Go}\{W(\Bx,\BF_{0})-\Bb\cdot\By_{0}\}d\Bx.
\]
In what follows we use the shorthand
\[\GD'E=E[\By_{0}]-E[\By_{d}]-E_{\Go}
\]
in order not to write $E_{\Go}$ in
every formula.
The first observation of Rice and Drucker is that for any smooth test field $\Bu(\Bx)$ vanishing on $\dOm_{D}$, we have, via integration by parts and equilibrium equations (\ref{eqd}),
\begin{equation}
  \label{five}
 \int_{\GO\setminus\Go}\av{\BP(\Bx,\BF_{d}),\Grad\Bu}d\Bx=\int_{\GO\setminus\Go}\Bb\cdot\Bu d\Bx+\int_{\dOm_{N}}\Bt\cdot\Bu\,dS.
\end{equation}
We then use formula (\ref{five}) with $\Bu=\By_{0}-\By_{d}$ and obtain 
\[
  \GD' E=\int_{\GO\setminus\Go}\{W(\Bx,\BF_{0})-W(\Bx,\BF_{d})-
  \av{\BP(\Bx,\BF_{d}),\BF_{0}-\BF_{d}}\}d\Bx=\int_{\GO\setminus\Go}\CE(\BF_{d},\BF_{0})d\Bx,
\]
where the Weierstrass $\CE$-function is defined in (\ref{Weierstrass}).

Next, Rice and Drucker use the following construction.
Let $\Bt_{\Go}(\Bx)=\BP(\Bx,\BF_{0}(\Bx))\Bn_{\Md\Go}(\Bx)$,
$\Bx\in\Md\Go$, where $\Bn_{\Md\Go}$ denotes the \emph{inward}
unit normal to $\Md\Go$ (i.e., the outward unit normal for $\Md(\GO\setminus\Go)$).
For any $s\in[0,1]$ let $\By_{s}(\Bx)$ be the the \emph{equilibrium}
in $\GO\setminus\Go$ with the same \bc s on $\dOm$, the same body
force loading on $\GO\setminus\Go$ as $\By_{0}$, and satisfying
$\BP(\Bx,\BF_{s})\Bn_{\Md\Go}=\Bt_{s}(\Bx)$, where $\Bt_{s}$ depends
smoothly on $s$ and is such that $\Bt_{0}=\Bt_{\Go}$, while
$\Bt_{1}=0$. For example, $\Bt_{s}=(1-s)\Bt_{\Go}$ is the a natural choice. We note that $\By_{s}=\By_{0}$, when $s=0$, and
$\By_{1}=\By_{d}$. 
Next Rice and Drucker use the observations:
\[
  W(\Bx,\BF_{0})-W(\Bx,\BF_{d})=-\int_{0}^{1}\av{\BP(\Bx,\BF_{s}),\Grad\dot{\By}_{s}}ds,\quad
 \BF_{0}-\BF_{d}=-\int_{0}^{1} \Grad\dot{\By}_{s}ds,
\]
where $\dot{\By}_{s}=\dif{\By_{s}}{s}$, so that
\[
  \GD' E=\int_{\GO\setminus\Go}\int_{0}^{1}\{\av{\BP(\Bx,\BF_{d})-
    \BP(\Bx,\BF_{s}),\Grad\dot{\By}_{s}}dsd\Bx.
\]
Now, recalling that $\By_{s}$ is an equilibrium configuration for each $s\in[0,1]$, switching the order of integration, and integrating by parts we obtain 
\[
  \GD' E=\int_{0}^{1}\int_{\Md\Go}(\BP(\Bx,\BF_{d})\Bn_{\Md\Go}
  -\BP(\Bx,\BF_{s})\Bn_{\Md\Go})\cdot\dot{\By}_{s}dS ds
\]
Recalling that $\BP(\Bx,\BF_{d})\Bn_{\Md\Go}=0$ we obtain formula~(12) from \cite{ridr67} for the the energy increment in the fully nonlinear and heterogeneous setting.
\begin{equation}
  \label{nonlin}
  \GD E=\int_{\Go}\{W(\Bx,\BF_{0})-\Bb\cdot\By_{0}\}d\Bx
  -\int_{\Md\Go}\int_{0}^{1}\Bt_{s}\cdot\dot{\By}_{s}dsdS.
\end{equation}
Integrating by parts in the integral in $s$, and using the formula
\[
\int_{\Md\Go}\Bt_{\Go}\cdot\By_{0}=\int_{\Go}\{\Bb\cdot\By_{0}-\av{\BP(\Bx,\BF_{0}),\BF_{0}}\}d\Bx,
\]
we obtain our final result---a simplification of Rice-Drucker formula (\ref{nonlin})
\begin{equation}
  \label{nonlin0}
  \GD E=\int_{\Go}\{W(\Bx,\BF_{0})-\av{\BP(\Bx,\BF_{0}),\BF_{0}}\}d\Bx
  -\int_{\Md\Go}\int_{0}^{1}\Bt_{\Go}\cdot\By_{s}dsdS,
\end{equation}
when $\Bt_{s}=(1-s)\Bt_{\Go}$.

If one does not look too closely, formula (\ref{nonlin}) resembles more (\ref{mazya}) for linear elasticity than the fully nonlinear fromula (\ref{DeltaEfin}). Taking a more detailed look,
formula (\ref{nonlin0}) involves a virtually intractable object $\By_{s}[\Bt_{s}]$, solving 
\begin{equation}
  \label{eqds}
  \begin{cases}
    \Div\BP(\Bx,\BF_{s})=\Bb,&\Bx\in\GO\setminus\Go,\\
    \By_{s}(\Bx)=\Bg(\Bx),&\Bx\in\dOm_{D},\\
    \BP(\Bx,\BF_{s})\Bn=\Bt,&\Bx\in\dOm_{N},\\
    \BP(\Bx,\BF_{s})\Bn=\Bt_{s},&\Bx\in\Md\Go.
  \end{cases}
\end{equation}
This prevents any further simplification of (\ref{nonlin0}) in the general case.

In the case of linear elasticity, equations (\ref{eqds}) become linear, and if we
set  $\Bt_{s}=(1-s)\Bt_{\Go}$, then, by linearity of equations (\ref{eqds}),
$\By_{s}=(1-s)\By_{0}+s\By_{d}$. In that case equation (\ref{nonlin}) gives
\[
  \GD E=E_{\Go}-\int_{\Md\Go}\int_{0}^{1}(1-s)\Bt_{\Go}\cdot(\By_{d}-\By_{0})dsdS=
E_{\Go}-\hf\int_{\Md\Go}\Bt_{\Go}\cdot(\By_{d}-\By_{0})dS. 
\]
In the presence of arbitrary inhomogeneities, but in the absence of prestress (at least in $\Go$), the original Clapeyron theorem (\ref{Clap}) applies to the solution $\By_{0}$ in $\Go$ and we have
\[
E_{\Go}=-\hf\int_{\Md\Go}\Bt_{\Go}\cdot\By_{0},
\]
where the minus sign is due to the defition of $\Bt_{\Go}$ in terms of
the inward unit normal to $\Md\Go$. In this case, we obtain the final formula
for linear elasticity with no prestress in $\Go$,
\begin{equation}
  \label{lin}
  E[\By_{0}]-E[\By_{d}]=-\hf\int_{\Md\Go}\Bt_{\Go}\cdot\By_{d}dS.
\end{equation}
This formula is virtually identical to (\ref{mazya}), and in fact easily reduces to it in the volume-rescaled limit, when $\Go$ is replaced with $\Ge\Go$ and when the body force vanishes at 0. In fact, the Rice-Drucker analysis presented here shows that the asymptotic formula  (\ref{mazya}) is valid in far greater generality than in Section~\ref{sub:linel}.

In the case of nonlinear elasticity, when the material is homogeneous, body forces are absent and
$\Go$ is replaced by $\Ge\Go$, formula (\ref{nonlin0}) implies
\begin{equation}
  \label{RDeltaE}
  \lim_{\Ge\to 0}\frac{E[\By_{0}]-E[\By_{\Ge}]}{\Ge^{n}}=|\Go|\{W_{0}-\av{\BP_{0},\BF_{0}(0)}\}
  -\int_{\Md\Go}\int_{0}^{1}\BP_{0}\Bn_{\Md\Go}(\Bz)\cdot\By^{\rm in}_{s}(\Bz)dsdS(\Bz), 
\end{equation}
where $W_{0}=W(\BF_{0}(0))$, and
$\By^{\rm in}_{s}(\Bz)$ solves the exterior problem
\begin{equation}
  \label{bvps}
    \begin{cases}
    \Div\BP(\Grad\By_{s}^{\rm in}(\Bz))=0,&\Bz\in\bb{R}^{n}\setminus\Go,\\
    \BP(\Grad\By_{s}^{\rm in}(\Bz))\Bn_{\Md\Go}=(1-s)\BP(\BF_{0}(0))\Bn_{\Md\Go},&\Bz\in\Md\Go,\\
\Grad\By^{\rm in}_{s}(\Bz)\to\BF_{0}(0),&\text{ as }|\Bz|\to\infty.
  \end{cases}
\end{equation}
We observe that Theorem~\ref{th:nlGrif} gives another formula for $\GD E$ above and therefore we obtain the following theorem.
\begin{theorem}
  \label{th:RD}
  Suppose $\By^{\rm in}_{s}$, $s\in[0,1]$ solves the exterior \bvp\ (\ref{bvps}), while $\By_{0}$ solves the \bvp\ (\ref{tracbc}). Then
  \[
    \int_{\Md\Go}\BP_{0}\Bn\cdot\int_{0}^{1}\By^{\rm in}_{s}dsdS(\Bz)=
    \nth{n}\int_{\Md\Go}W(\Grad\By^{\rm in}_{1})(\Bn\cdot\Bz)dS(\Bz)-
   |\Go|\{W_{0}-\av{\BP_{0},\BF_{0}(0)}\}, 
\]
where $\Bn$ is the \emph{outward} unit normal to $\Md\Go$, which we can also rewrite as
\begin{equation}
  \label{creation}
  \int_{\Md\Go}\BP_{0}\Bn\cdot\int_{0}^{1}(\By^{\rm in}_{s}(\Bz)-\BF_{0}(0)\Bz)dsdS(\Bz)=
    \nth{n}\int_{\Md\Go}(W(\Grad\By^{\rm in}_{1})-W_{0})(\Bn\cdot\Bz)dS(\Bz). 
\end{equation}
\end{theorem}
The physical meaning of formula (\ref{creation}), i.e. of equivalence of formulas (\ref{DeltaEfin}) and (\ref{nonlin0}) is that the energy increment $\GD E$ can be represented as a portion of the work of configurational forces that ``create'' a void.

\subsection{Computation of the Griffith's discrepancy term (\ref{GrerrG})}
\label{app:A8}

In order to compute the limit in (\ref{GrerrG}) we need to use the far-field asymptotics (\ref{farfield}), obtaining
\[
  (\BP(\Grad\By_{\rm in})-\BP_{0})\Bn\cdot\BF_{0}\Bz= \frac{(\SFC\BS)\Bn\cdot\BF_{0}\Bz}{|\Bz|^{n}}
  -n\frac{(\SFC(\BS\Bz\otimes\Bz))\Bn\cdot\BF_{0}\Bz}{|\Bz|^{n+2}}+O(|\Bz|^{-n}),
\]
where $\SFC=W_{\BF\BF}(\BF_{0})$, and $\BS$ is
the polarization tensor \cite{amka2007} that depends in the shape of $\Go$.
Using the fact that $|\Bz|=R$ and $\Bn=\Hat{\Bz}=\Bz/|\Bz|$ on $\Md B(0,R)$, we obtain
\[
  G=\int_{\bb{S}^{n-1}}\{n(\SFC(\BS\Hat{\Bz}\otimes\Hat{\Bz}))\Hat{\Bz}\cdot
  \BF_{0}\Hat{\Bz}-(\SFC\BS)\Hat{\Bz}\cdot\BF_{0}\Hat{\Bz}\}\,dS(\Hat{\Bz}).
\]
Thus, taking into account that
\[
  \int_{\bb{S}^{n-1}}z^{i}z^{j}dS=\frac{|\bb{S}^{n-1}|}{n}\Gd^{ij},\qquad
  \int_{\bb{S}^{n-1}}z^{i}z^{j}z^{k}z^{l}dS=\frac{|\bb{S}^{n-1}|}{n(n+2)}\left(
  \Gd^{ij}\Gd^{kl}+\Gd^{ik}\Gd^{jl}+\Gd^{il}\Gd^{jk}\right),
\]
we obtain
\begin{equation}
  \label{G}
G=\frac{|B(0,1)|}{n+2}(C_{ijkl}(nS_{kj}F_{il}-2S_{kl}F_{ij})
  +nC_{ijkj}S_{kl}F_{il}), 
\end{equation}
where the summation over repeated indices is assumed, and where we used the relation $|\bb{S}^{n-1}|=n|B(0,1)|$ between the surface and the volume of the unit ball.

If $\SFC$ is an isotropic elastic tensor with Lam\'e moduli $\Gl$ and $\mu$, then
\[
  \SFC(\BS\Hat{\Bz}\otimes\Hat{\Bz})\Hat{\Bz}\cdot\BF_{0}\Hat{\Bz}=
  (\Gl+\mu)(\BS\Hat{\Bz},\Hat{\Bz})(\BF_{0}\Hat{\Bz},\Hat{\Bz})+
  \mu(\BS\Hat{\Bz},\BF_{0}\Hat{\Bz}),
\]
\[
  (\SFC\BS)\Hat{\Bz}\cdot\BF_{0}\Hat{\Bz}=
  \Gl(\BF_{0}\Hat{\Bz}\cdot\Hat{\Bz})\Trc\BS+2\mu\BS\Hat{\Bz}\cdot\BF_{0}\Hat{\Bz},
\]
and thus,
\begin{equation}
  \label{Gisotr}
    G=\frac{|B(0,1)|}{n+2}((n\mu-2\Gl)\Trc\BF_{0}\Trc\BS
  +(2\Gl n+\mu(n^{2}+2n-4))\av{\BS,\BF_{0}}). 
\end{equation}
When $n=2$ we have
\[
G=\frac{\pi}{2}((\mu-\Gl)\Trc\BF_{0}\Trc\BS+2(\Gl+\mu)\av{\BS,\BF_{0}}),
\]
When $n=3$ we have
\[
  G=\frac{4\pi}{15}((3\mu-2\Gl)\Trc\BF_{0}\Trc\BS
  +(6\Gl+11\mu)\av{\BS,\BF_{0}}).
\]
If $\BF_{0}=\frac{p}{n\Gk}\BI_{n}$, then
\[
  G=\frac{2(n-1)p\mu|B(0,1)|}{n\Gk}\Trc\BS=
\begin{cases}
  \frac{\pi\mu p}{\Gk}\Trc\BS,& n=2,\\
  \frac{16\pi\mu p}{9\Gk}\Trc\BS,& n=3.
\end{cases}
\]
The case of nucleating a small spherical cavity in a body occupied  by isotropic linearly elastic material is exactly solvable. The solution of (\ref{Linf}) with $\BF_{0}=\frac{p}{n\Gk}\BI_{n}$ is
\[
\By_{\rm in}(\Bz)=\left(\frac{p}{n\Gk}+\frac{p|\Bz|^{-n}}{2(n-1)\mu}\right)\Bz,
\]
resulting in $\BS=\frac{p}{2(n-1)\mu}\BI_{n}$. Thus,
\[
  G=\frac{p^{2}|B(0,1)|}{\Gk}=
  \begin{cases}
    \frac{\pi p^{2}}{\Gk},&n=2,\\
    \frac{4\pi p^{2}}{3\Gk},&n=3.
  \end{cases}
\]
In this case all of our results in this section can also be verified by a direct calculation,
using the explicit solutions (\ref{ulin0}), (\ref{ulin}) of (\ref{uueps}). Needless to say,    the obtained results  specialized to the appropriate value of $n$ and using the relevant  polarization tensor $\BS$,  agree with the correct formula  presented by Griffith 
in his second  paper  \cite{grif24}.
\def\cprime{$'$} \ifx \cedla \undefined \let \cedla = \c \fi\ifx \cyr
  \undefined \let \cyr = \relax \fi\ifx \cprime \undefined \def \cprime
  {$\mathsurround=0pt '$}\fi\ifx \prime \undefined \def \prime {'}
  \fi\def\Ya{Ya}

\end{document}